\documentclass{article}
\usepackage{fullpage}

\usepackage{xcolor}

\usepackage{subfiles}

\usepackage{graphicx}

\usepackage{amsrefs}
\usepackage{amsmath}
\usepackage{amsthm}
\usepackage{amssymb}
\usepackage{mathrsfs}
\usepackage{enumitem}
\usepackage{esint}
\usepackage{stmaryrd}
\usepackage{slashed}
\usepackage{hyperref}
\usepackage{comment}

\usepackage{bbm}

\usepackage[nottoc,notlot,notlof]{tocbibind}

\numberwithin{equation}{section}
\newtheorem{thm}{Theorem}[section]
\newtheorem{prop}[thm]{Proposition}
\newtheorem{lem}[thm]{Lemma}

\newtheorem{cor}[thm]{Corollary}

\theoremstyle{definition}

\theoremstyle{remark}
\newtheorem{rmk}[thm]{Remark}

\DeclareMathOperator{\Vol}{Vol}

\let\div\relax
\DeclareMathOperator{\div}{div}

\let\d\relax
\newcommand{\d}{\partial}

\newcommand{\slnabla}{\slashed{\nabla}}
\newcommand{\slDelta}{\slashed{\Delta}}

\newcommand{\rslnabla}{\mathring{\slashed{\nabla}}}
\newcommand{\rslDelta}{\mathring{\slashed{\Delta}}}
\newcommand{\rsldiv}{\mathring{\slashed{\div}}}

\newcommand{\RR}{\mathbb{R}}

\newcommand{\ZZ}{\mathbb{Z}}

\newcommand{\TT}{\mathbb{T}}
\let\SS\relax
\newcommand{\SS}{\mathbb{S}}

\newcommand{\R}{\mathcal{R}}
\newcommand{\C}{\mathcal{C}}
\newcommand{\D}{\mathcal{D}}
\newcommand{\I}{\mathcal{I}}

\newcommand{\E}{\mathcal{E}}
\newcommand{\V}{\mathcal{V}}
\newcommand{\U}{\mathcal{U}}
\newcommand{\W}{\mathcal{W}}
\newcommand{\M}{\mathcal{M}}
\newcommand{\X}{\mathcal{X}}

\let\S\relax
\newcommand{\S}{\mathcal{S}}

\newcommand{\fC}{\mathfrak{C}}

\title{Boundedness and decay of waves on spatially flat decelerated FLRW spacetimes}
\author{Mahdi Haghshenas\thanks{Imperial College London, UK. Email: m.haghshenas23@imperial.ac.uk}  \thanks{University College London, UK. Email: mahdi.haghshenas.22@ucl.ac.uk} }

\date{\today}

\begin{document}

\maketitle

\begin{abstract}
    We study the linear wave equation on a class of spatially homogeneous and isotropic Friedmann--Lemaître--Robertson--Walker (FLRW) spacetimes in the decelerated regime  with spatial topology $\mathbb{R}^3$. Employing twisted $t$-weighted multiplier vector fields, we establish uniform energy bounds and derive integrated local energy decay estimates across the entire range of the decelerated expansion regime. Furthermore, we obtain a hierarchy of $r^p$-weighted energy estimates à la the Dafermos--Rodnianski $r^p$-method, which leads to energy decay estimates. 
    As a consequence, we demonstrate pointwise decay estimates for solutions and their derivatives. In the wave zone, this pointwise decay is optimal in the ``radiation'' and ``sub-radiation'' cases, and almost optimal around the radiation case.
\end{abstract}

\tableofcontents
\section{Introduction}\label{sec:intro}
The standard cosmological models describing a homogeneous and isotropic universe in general relativity are represented by the Friedmann--Lemaître--Robertson--Walker (FLRW) spacetimes \cite{Robertson1933}. The current article concerns FLRW spacetimes with $\RR^3$ spatial topology, for which the Lorentzian metrics, defined on the 4-dimensional manifold  $M=(0,\infty)_t\times \RR^3$, read 
\begin{equation}\label{eq:intro:Metric-x}
    g_q=-dt^2  + t^{2q}\left((dx^1)^2+(dx^2)^2+(dx^3)^2\right)\,,
\end{equation}
 for $0<q<1$ corresponding to the \emph{decelerated regime}. We refer to $a(t)=t^q$ as the scale factor, which encodes the expansion of the universe and the big bang-type singularity as $t\to 0$. This paper addresses the dynamics of the scalar wave equation on the FLRW background metric towards the future.
 The covariant wave equation associated to $g_q$ in $(t,x^1,x^2,x^3)$ coordinates is given by
\begin{equation}\label{eq:intro:wave-equation}
    \Box_{g_q} \psi:= -\d_{tt} \psi -\frac{3q}{t}\d_t \psi + \frac{1}{t^{2q}} \Delta_{\RR^3} \psi=0\,.
\end{equation}

The spacetimes \eqref{eq:intro:Metric-x} belong to the more general class of FLRW spacetimes $(M,g)$ with
\begin{align}\label{eq:intro:general_FLRW-metric}
    M=(0,\infty)_t\times S\,,\qquad g:=-dt^2 + a^2(t) g_S\,,
\end{align}
where $(S,g_S)$ is a three-dimensional Riemannian manifold of constant curvature $K$. The above metrics can be derived from the \textit{cosmological principle}, that is the homogeneity and isotropy of the universe on large cosmological scales \cites{Robertson29,Robertson35Kinematics}. 

In the spatially flat case, when $K=0$, $S$ can be taken to be $\RR^3$ or $\TT^3$. This article focuses primarily on the $\RR^3$ topology where the dispersive features of waves are apparent. The FLRW spacetimes \eqref{eq:intro:Metric-x} are spatially homogeneous and spherically symmetric with \emph{area radius} $R=t^q r$ (where  $r=\sqrt{(x^1)^2+(x^2)^2+(x^3)^2}$), and one formally recovers the Minkowski metric by letting $q=0$ in \eqref{eq:intro:Metric-x}. The FLRW metrics, in contrast to Minkowski space, are neither asymptotically flat nor stationary. Notably, the timelike vector field $\d_t$ is not a Killing vector field for $q\neq0$, and as a result, the energy boundedness of waves is not inherently ensured. However, the conformal structure and asymptotic features of the decelerated FLRW spacetimes can be compared to those of Minkowski spacetime. For instance, the metric \eqref{eq:intro:Metric-x} is conformally flat and a null conformal boundary $\I^+$ called \emph{future null infinity}, which represents far away observers, can be attached to the FLRW spacetime; see Figure \ref{fig:Sigma-tau} or refer to \cite{HawkingEllis}*{Chapter~5.3}.

Many results in the current work also hold in Minkowski space, formally corresponding to $q=0$, and are well known in this case. For simplicity of notation, however, the case $q=0$ is omitted from the statements in the following sections. 

The metric $g_q$ \eqref{eq:intro:Metric-x} for all $\frac{1}{3}\leq q \leq \frac{2}{3}$ arises as a solution of the Einstein--Euler equation for a perfect fluid with a linear equation of state and with cosmological constant $\Lambda=0$. In particular, the cases $q=\frac{1}{2}$ and $q=\frac{2}{3}$ correspond to the \textit{radiation} and the \textit{dust} model, respectively. Furthermore, $g_q$ also solves the Einstein--massless Vlasov equation when $q=\frac{1}{2}$. In fact, all solutions of Einstein--massless Vlasov with an isotropic and ``irrotational'' particle density are either stationary or described by the FLRW metric  \cite{EhlersGerenSachs}. We remark that the scalar curvature of \eqref{eq:intro:Metric-x} vanishes when $q=\frac{1}{2}$.

The study of equation \eqref{eq:intro:wave-equation} is relevant for the analysis of wave-type systems near FLRW spacetimes. In particular, \eqref{eq:intro:wave-equation} can be viewed as a ``poor man's'' linearisation of Einstein equations around \eqref{eq:intro:Metric-x}. 

\subsection{Initial value problem}\label{sec:IVP}

In the present article, we consider an initial data hypersurface $\Sigma_{\tau_0}$ that is partially null and connects the centre $\{r=0\}$ to future null infinity $\I^+$. The spacelike part and the null part, respectively denoted by $\S_{\tau_0}$ and $\C_{\tau_0}$, are separated by $\{r=\rho\}$ for an arbitrary $\rho>0$, as depicted in Figure \ref{fig:Sigma-tau}. 

 Moreover, the function $\tau$ and the hypersurface $\Sigma_\tau$, defined in Section \ref{sec:metric-wave-foliation}, are such that $\Sigma_{\tau'}$ is the hypersurface of all points $p$ with $\tau(p)=\tau'$. We also set $\R^{\tau'}_{\tau}=\cup_{\bar \tau\in [\tau,\tau']} \Sigma_{\bar\tau}$; see Figure \ref{fig:Sigma-tau}.
\begin{figure}
    \centering\includegraphics[width=0.3\textwidth]{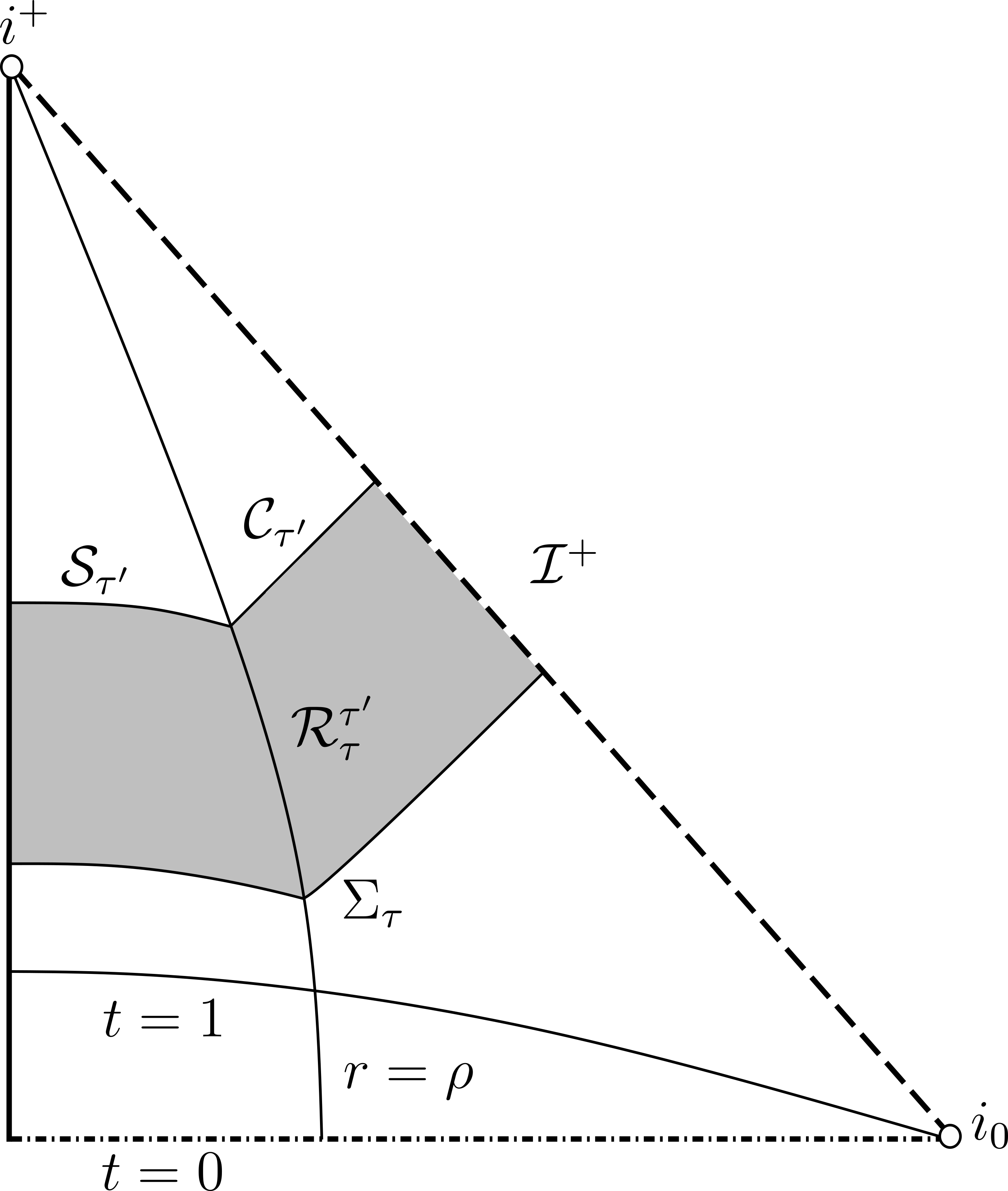}
        \caption{Penrose diagram of FLRW spacetimes \eqref{eq:intro:Metric-x} for $0<q<1$, with the spacelike hypersurface $\{t=1\}$, the hypersurface $\Sigma_\tau$, and the spacelike singularity at $\{t=0\}$.}\label{fig:Sigma-tau}
    \end{figure}
    
The pair $(\psi_0,\psi_1)$ is called a \emph{compactly supported smooth initial data set} on $\Sigma_{\tau_0}$ if $\psi_0$ is a smooth function on both $\S_{\tau_0}$ and $\C_{\tau_0}$ with compact support, $\psi_1$ is a compactly supported smooth function on $\S_{\tau_0}$, and moreover, there exists a smooth function $\Psi$ on $\R^{\tau_1}_{\tau_0}$ for some $\tau_1$ such that 
\begin{align*}
    \Psi|_{\Sigma_{\tau_0}}=\psi_0\,,\qquad \d_t\Psi|_{\S_{\tau_0}}=\psi_1\,.
\end{align*}
It is worth mentioning that compactly supported smooth Cauchy data on $\{t=\text{constant}\}$, for example, induces such smooth initial data on $\Sigma_{\tau_0}$ by considering $\rho$ large enough. Note that the assumption of compact support is introduced to simplify the arguments in Section \ref{sec:decay} but it may be relaxed by instead requiring that an appropriate weighted energy flux of $\psi_0$ and $\psi_1$ along $\Sigma_{\tau_0}$ is bounded.

In this paper, a solution $\psi$ to the wave equation always means a smooth function satisfying the initial value problem
\begin{align}\label{eq:intro:IVP-wave}
    \begin{cases}
    \Box_{g_q}\psi=0 \,,\\
     \psi|_{\Sigma_{\tau_0}}=\psi_0\,,\\
    \d_t\psi|_{\S_{\tau_0}}=\psi_1\,,
    \end{cases}
\end{align}
for a compactly supported smooth initial data set $(\psi_0,\psi_1)$ on $\Sigma_{\tau_0}$.

\subsection{Main results}
The main results of the present paper concern boundedness and decay properties of solutions of \eqref{eq:intro:IVP-wave}, following the Dafermos--Rodnianski $r^p$-method \cite{DafRodNewMethod}. The main ingredients of this method are 
\begin{enumerate}
    \item Energy boundedness;
    \item Integrated local energy decay (ILED);
    \item A hierarchy of $r^p$-weighted estimates.
\end{enumerate}
The above results for the FLRW spacetimes \eqref{eq:intro:Metric-x} are presented in Theorems \ref{thm:intro:boundedness}, \ref{thm:itro:Morawetz}, and \ref{thm:intro:rp}, respectively. Consequently, the pointwise and energy decay of $\psi$, as stated in Corollary \ref{cor:general:decay}, follows from these three theorems.

The Dafermos--Rodnianski $r^p$-method has been generalised to a broad class of asymptotically flat spacetimes by Moschidis \cite{Moschidis16rp}. Further refinements of the $r^p$-method can be found in \cites{AAG18late,AAG18vector, Schlue13HigherDim}. See also the recent paper \cite{VdMoortel25rp}.

The $r^p$-method is particularly well suited to non-linear problems. See, for example, the work of Dafermos--Holzegel--Rodnianski--Taylor \cite{DHRTQuasilinearWave} on non-linear wave equations around various stationary asymptotically flat spacetimes. In the absence of \emph{trapped null geodesics}, the above three ingredients for linear waves lead to a proof of small data global solutions (see case (i) in \cite{DHRTQuasilinearWave}*{Section~3.4.1}, where the $r^p$-estimate is required only for some fixed $p>0$). In the presence of trapped null geodesics, the method of \cite{DHRTQuasilinearWave} requires an additional ``physical space top order identity''. Note also the work of Keir \cite{Keir18WeakNull} on non-linear wave equations around Minkowski space satisfying \emph{weak null condition}.

The present setting considered with a non-linearity satisfying a suitable \emph{null condition}, is comparable to the former case of \cite{DHRTQuasilinearWave}.

The energy flux considered in this paper, denoted by $\E_q^\varepsilon[\psi](\tau)$, is a ``twisted'' weighted energy flux of $\psi$ through $\Sigma_\tau$, where $q$ encodes a $t$-weight corresponding to the scale factor $t^q$ in \eqref{eq:intro:Metric-x}. The energy flux $\E_q^\varepsilon[\psi](\tau)$ includes  an arbitrarily small $\varepsilon\geq 0$ loss in the $t$-weight factors. We write $\E_q[\psi]:=\E_q^0 [\psi]$ in the case of no loss. The energy $\E_q^\varepsilon[\psi](\tau)$ is referred to as a twisted energy in view of the fact that the $\d_t$ derivative of $\psi$ appears only after $\psi$ is multiplied by $t^q$. See already \eqref{eq:intro:energy-sigma-tau} for the precise definition of $\E_q^\varepsilon[\psi](\tau)$.

In the radiation case, corresponding to $q=\frac{1}{2}$, the wave equation can be transformed to the standard wave equation on the Minkowski space, and many results can be directly invoked from Minkowski space. A representation formula, for instance, is available in the radiation case. The approach taken in this paper does not rely on this transformation to Minkowski space, however. 

Nevertheless, the case $q=\frac{1}{2}$ is somewhat the simplest.  As such, the following results, particularly Theorem \ref{thm:intro:rp}, are sharpest for $q$ close to $\frac{1}{2}$ and $0$. In fact, $q=\frac{1}{2}$ can be viewed as a threshold, with the statements and proofs differing for $q>\frac{1}{2}$ and $q\leq\frac{1}{2}$. 

The following theorem establishes the boundedness of the twisted energy flux $\E_q^\varepsilon[\psi] (\tau)$ associated to the solution $\psi$. It is important to emphasise that $\E_q^\varepsilon[\psi] (\tau)$ already includes $t$-dependent weight factors, with $q=\frac{1}{2}$ exhibiting the strongest $t$-weight, see Remark \ref{rmk:DefEnergy}. 
\begin{thm}[Energy boundedness of waves on FLRW]\label{thm:intro:boundedness}
    Suppose $0<q<1$, and $\psi$ solves the initial value problem \eqref{eq:intro:IVP-wave}. Then for any $\tau_2>\tau_1\geq \tau_0$, we have
    \begin{equation}\label{eq:intro:EB}
          \E_q[\psi] (\tau_2) \lesssim   \E_q[\psi] (\tau_1)\,. 
    \end{equation}
    Moreover, for any $\varepsilon>0$, we also have $\E_q^\varepsilon[\psi] (\tau_2) \lesssim   \E_q^\varepsilon[\psi] (\tau_1)\,.$
\end{thm}

The next theorem, to be referred to as ILED,  is of \textit{Morawetz-type} and \textit{local} in the sense that it is more effective when applied to regions with bounded $r$. Before stating the theorem, for any $\varepsilon\geq 0$, define
    \begin{align}\label{eq:intro:mu}
        \mu_q^\varepsilon:=\begin{cases}
        0& \qquad q\leq\frac{1}{2}\\
        4q-2+\varepsilon & \qquad q>\frac{1}{2}
    \end{cases}\,,
    \end{align}
    which captures the difference of $t$-weight factors in different cases. Also, set $\mu_q=\mu_q^0$.
\begin{thm}[Integrated local energy decay for waves on FLRW]\label{thm:itro:Morawetz}
    Let $\delta>0$ and $0<q<1$. If $\psi$ solves the initial value problem \eqref{eq:intro:IVP-wave}, then for any $\tau_2>\tau_1\geq \tau_0$, the following estimate holds\footnote{Here $\rslnabla$ denotes the standard covariant derivative on the unit sphere $\SS^2$, and the norm of $\rslnabla\psi$ is also measured via the standard metric on $\SS^2$.}
    \begin{align}
          \int_{\R^{\tau_2}_{\tau_1}} t^{-\mu_q} \left(\frac{t^{q}|\d_t(t^q\psi)|^2}{1+r^{3+\delta}} +\frac{t^{q}|\d_r \psi|^2}{(1+r)^{1+\delta}} + \frac{t^q}{r^3}|\rslnabla\psi|^2+\frac{t^{q}\psi^2}{1+r^{3+\delta}} \right) \, dxdt&\lesssim  \E_q[\psi] (\tau_1)\,.\label{eq:intro:ILED}
      \end{align}
      Moreover, the above estimate remains valid if we replace $\mu_q$ with $\mu_q^\varepsilon$ and $\E_q[\psi] (\tau_1)$ with $\E_q^\varepsilon[\psi] (\tau_1)$ for any $0<\varepsilon\ll 1-q$.
\end{thm}
We remark that the above theorem also holds with improved weight factors in certain ranges of $q$. For the best proven estimates in each case, refer to Propositions \ref{prop:sub:ILED} and \ref{prop:sup:Morawetz}.

The next theorem concerns a hierarchy of $r^p$-weighted estimates in a region near future null infinity $\I^+$. There is a restriction on the allowed range of $p$ depending on $q$, encoded by
\begin{equation}\label{eq:intro:sigma_q}
    \sigma_q:= 2(\frac{q|1-2q|}{|1-q|^2})^\frac{1}{2}\,.
\end{equation}
Observe that $\sigma_q$ vanishes for Minkowski spacetime and the radiation-filled universe, that is $q=0$ and $q=\frac{1}{2}$. Also, note that $\sigma_q\geq1$ for $q\geq \frac{1+2\sqrt 2}{7}$ and $q=\frac{1}{3}$.

\begin{thm}[$r^p$-estimates for waves on FLRW]\label{thm:intro:rp}
For $0<q<1$, assume $\psi$ is a solution to the initial value problem \eqref{eq:intro:IVP-wave}. Moreover suppose $\tau_2>\tau_1\geq \tau_0$, and 
\begin{equation}\label{eq:intro:range-p}
   0<p< \max\{1, 2-\sigma_q\}.
\end{equation}
If $\D^{\tau_2}_{\tau_1}=\R^{\tau_2}_{\tau_1}\cap \{r>\rho\}$ and $\varphi=rt^q\psi$, then we have\footnote{Here $\d_v=t^q \d_t +\d_r$ is the derivative in the null direction tangent to $\C_\tau$. Refer to Section \ref{sec:metric-wave-foliation}.}
    \begin{align}\label{eq:intro:rpEstimate}
    \begin{split}
        \int_{\D^{\tau_2}_{\tau_1}} t^{-\mu_q}r^{p-1} \left( |\d_v \varphi|^2 +r^2|\d_v(t^q \psi)|^2+ \frac{1}{r^2}|\rslnabla\varphi|^2+ \frac{1}{r^2}\varphi^2\right) \, dudvd\omega + \int_{\C_{\tau_2}}t^{-\mu_q} r^p|\d_v\varphi|^2 \, dvd\omega \\ \lesssim \int_{\C_{\tau_1}} t^{-\mu_q}r^p|\d_v\varphi|^2 \, dvd\omega +  \E_q[\psi] (\tau_1)\,,
    \end{split}
\end{align}
in which $\mu_q=\mu^0_q$ and $\E_q[\psi] (\tau_1)=\E^0_q[\psi] (\tau_1)$.
\end{thm}

We remark that in the radiation case, \textit{i.\@e.\@}\ $q=\frac{1}{2}$, and the Minkowski spacetime, suitably adjusted $r^p$-estimates are valid for all $0\leq p \leq 2$; see estimate \eqref{eq:radiation:r2Estimate}. The above theorem, if restricted to $q\sim \frac{1}{2}$, can be viewed as a perturbation of the radiation case with small loss. 

Given that the three components of the $r^p$-method have now been established, the corollary below follows directly from general functional inequalities discussed in Section \ref{sec:functional_ineq}.
\begin{cor}
    [Decay of waves on FLRW]\label{cor:general:decay}
    Suppose $\varepsilon>0$, $0<q<1$ and $q\neq \frac{1}{3}$, and consider $\sigma>\sigma_q$. If $\psi$ satisfies the initial value problem \eqref{eq:intro:IVP-wave}, then the following statements hold for all $\tau\geq \tau_0$.
    \begin{enumerate}
        \item \textbf{(Decay of energy)}. The energy flux $  \E_q^\varepsilon[\psi] (\tau)$ and $\E_q[\psi] (\tau)=\E_q^0[\psi] (\tau)$ decay as
        \begin{align*}
              \E_q[\psi] (\tau)&\lesssim \frac{\E_q[\psi_0,\psi_1]}{\tau^{2}}& {\rm{ if }} \quad& q=\frac{1}{2}\,,\\
              \E_q[\psi] (\tau)&\lesssim \frac{\E_q[\psi_0,\psi_1]}{\tau^{2-\sigma}}&{\rm{ if }}\quad & q<\frac{1+2\sqrt2}{7}\,,\\
              \E_q^\varepsilon[\psi] (\tau)&\lesssim \frac{\E_q[\psi_0,\psi_1]}{\tau}&{\rm{ if }}\quad &\frac{1+2\sqrt2}{7}\leq q<1\,,
        \end{align*}
        where $\E_q[\psi_0,\psi_1]$ denotes a weighted initial energy of $\psi$ defined in \eqref{eq:decay:initial-energy}.
        \item \textbf{(Pointwise decay of \boldmath$\psi$)}. The solution $\psi$ satisfies
        \begin{align*}
            |\psi|&\lesssim \frac{\bar\E_q[\psi_0,\psi_1]^\frac{1}{2}}{  t^q \tau \sqrt{1+r}}\,, &|\psi|&\lesssim\frac{\bar\E_q[\psi_0,\psi_1]^\frac{1}{2}}{(1+r)  t^q  \sqrt \tau}&{\rm{ if }}\quad &q=\frac{1}{2}\,,\\
            |\psi|&\lesssim\frac{ \bar\E_q[\psi_0,\psi_1]^\frac{1}{2} }{ t^q  \tau^{1-\frac{\sigma}{2}}\sqrt{1+r}}\,,&|\psi|&\lesssim \frac{\bar\E_q[\psi_0,\psi_1]^\frac{1}{2}}{(1+r)  t^q  \sqrt {\tau^{1-\frac{\sigma}{2}}}}&{\rm{ if }}\quad&q<\frac{1}{2}\,,\\
            |\psi|&\lesssim \frac{\bar\E_q[\psi_0,\psi_1]^\frac{1}{2}}{ t^{1-q}\tau^{1-\frac{\sigma}{2}}  \sqrt{1+r}}\,,&|\psi|&\lesssim \frac{\bar\E_q[\psi_0,\psi_1]^\frac{1}{2}}{(1+r)  t^{1-q}  \sqrt {\tau^{1-\frac{\sigma}{2}}}}&{\rm{ if }}\quad&\frac{1}{2}<q<\frac{1+2\sqrt 2}{7}\,,\\
            |\psi|&\lesssim \frac{\bar\E_q[\psi_0,\psi_1]^\frac{1}{2}}{ t^{1-q-\frac{\varepsilon}{2}}\sqrt \tau  \sqrt{1+r}}\,,&|\psi|&\lesssim \frac{ \bar\E_q[\psi_0,\psi_1]^\frac{1}{2}}{(1+r)  t^{1-q-\varepsilon}  }&{\rm{ if }}\quad&\frac{1+2\sqrt 2}{7}\leq q<1\,,
            \end{align*}
        where $\bar \E_q[\psi_0,\psi_1]$ is a high order weighted initial energy of $\psi$ containing derivatives up to order three defined in \eqref{eq:decay:h-initial-energy}.
        In particular, we have
        \begin{equation}\label{eq:intro:decay-pointwise-t}
            \begin{aligned}
                |\psi|&\lesssim \frac{1}{  t}\bar\E_q[\psi_0,\psi_1]^\frac{1}{2} &{\rm{ if }}\quad &q=\frac{1}{2}\,,\\
            |\psi|&\lesssim\frac{  \tau^{\frac{\sigma}{2}}}{ t}\bar\E_q[\psi_0,\psi_1]^\frac{1}{2}&{\rm{ if }}\quad&q<\frac{1}{2}\,,\\
            |\psi|&\lesssim\frac{  \tau^{\frac{\sigma}{2}}}{ t^{2-2q}}\bar\E_q[\psi_0,\psi_1]^\frac{1}{2}&{\rm{ if }}\quad&\frac{1}{2}<q<\frac{1+2\sqrt 2}{7}\,,\\
            |\psi|&\lesssim\frac{  \tau^{\frac{1}{2}}}{ t^{2-2q-\varepsilon}}\bar\E_q[\psi_0,\psi_1]^\frac{1}{2}&{\rm{ if }}\quad&\frac{1+2\sqrt 2}{7}\leq q<1\,.
            \end{aligned}
        \end{equation}
        \item \textbf{(Pointwise decay of derivatives of $\psi$)}. The time derivative of the solution $\psi$ satisfies 
        \begin{align*}
            |\d_t\psi|&\lesssim \frac{\bar{\bar\E}[\psi_0,\psi_1]^\frac{1}{2} }{ \sqrt \tau t^{2q} (1+r)}&{\rm{ if }}\quad &q=\frac{1}{2}\,,\\
            |\d_t\psi|&\lesssim\frac{\bar{\bar\E}_q[\psi_0,\psi_1]^\frac{1}{2}}{ \tau^{-\frac{\sigma}{2}} t^{2q}  (1+r)} &{\rm{ if }}\quad&q\leq\frac{1}{2}\,,\\
            |\d_t\psi|&\lesssim \frac{\bar{\bar\E}_q[\psi_0,\psi_1]^\frac{1}{2} }{\tau^{-\frac{\sigma}{2}}  t^{1}  (1+r)}&{\rm{ if }}\quad&\frac{1}{2}<q<\frac{1+2\sqrt 2}{7}\,,\\
            |\d_t\psi|&\lesssim \frac{\bar{\bar\E}_q[\psi_0,\psi_1]^\frac{1}{2} }{\tau^{-\frac{1}{2}}  t^{1-\frac{\varepsilon}{2}}  (1+r)}&{\rm{ if }}\quad&\frac{1+2\sqrt 2}{7}\leq q<1\,,
            \end{align*}
        with $\bar{\bar \E}[\psi_0,\psi_1]$ and $\bar{\bar \E}_q[\psi_0,\psi_1]$ denoting high order weighted initial energy fluxes of $\psi$ containing derivatives up to order four and five defined in \eqref{eq:decay:hh-initial-energy-radiation} and \eqref{eq:decay:hh-initial-energy}, respectively. Moreover, the following estimates hold for null derivatives $\d_v$ and $\d_u$.
        \begin{align*}
            |\d_v\psi|&\lesssim \frac{\bar{\bar\E}_q[\psi_0,\psi_1]^\frac{1}{2} \log \tau } {t^q (1+r)^2 }\,, &|\d_u\psi|&\lesssim\frac{\bar{\bar\E}_q[\psi_0,\psi_1]^\frac{1}{2} \log \tau}{t^q (1+r) }&{\rm{if }}\quad &q=\frac{1}{2}\,,\\
            |\d_v\psi|&\lesssim\frac{ \bar{\bar\E}_q[\psi_0,\psi_1]^\frac{1}{2} }{\tau^{-\frac{\sigma}{2}}  t^q  (1+r)^2}\,,&|\d_u\psi|&\lesssim \frac{\bar{\bar\E}_q[\psi_0,\psi_1]^\frac{1}{2}}{ \tau^{-\frac{\sigma}{2}}  t^q  (1+r)}&{\rm{ if }}\quad&q\leq\frac{1}{2}\,,\\
            |\d_v\psi|&\lesssim\frac{ \bar{\bar\E}_q[\psi_0,\psi_1]^\frac{1}{2} }{\tau^{-\frac{\sigma}{2}} t^{1-q}  (1+r)^2}\,,&|\d_u\psi|&\lesssim \frac{\bar{\bar\E}_q[\psi_0,\psi_1]^\frac{1}{2}}{ \tau^{-\frac{\sigma}{2}} t^{1-q}  (1+r)}&\rm{ if }\quad&\frac{1}{2}<q<\frac{1+2\sqrt 2}{7}\,,\\
            |\d_v\psi|&\lesssim\frac{ \bar{\bar\E}_q[\psi_0,\psi_1]^\frac{1}{2} }{\tau^{-\frac{1}{2}} t^{1-q-\frac{\varepsilon}{2}}  (1+r)^2}\,,&|\d_u\psi|&\lesssim \frac{\bar{\bar\E}_q[\psi_0,\psi_1]^\frac{1}{2}}{ \tau^{-\frac{1}{2}} t^{1-q-\frac{\varepsilon}{2}}  (1+r)}&\rm{ if }\quad&\frac{1+2\sqrt 2}{7}\leq q<1\,.
            \end{align*}
    \end{enumerate}
\end{cor}
Note that, in the ``wave zone'' (that is, for bounded $\tau$), the estimates \eqref{eq:intro:decay-pointwise-t} are optimal for $q\leq\frac{1}{2}$ and is almost optimal for $q\sim \frac{1}{2}$.
\begin{rmk}[The case $q=\frac{1}{3}$]\label{rmk:1/3}
    The case $q=\frac{1}{3}$ is excluded in Corollary \ref{cor:general:decay} as the proof given in this article requires $\sigma_q < 1$ (recall that $\sigma_{\frac{1}{3}}=1$). However, since Theorem \ref{thm:intro:rp} remains valid for $0<p<1$ when $q=\frac{1}{3}$, a modified version of Corollary \ref{cor:general:decay}, with an arbitrarily small loss, can also be established for $q=\frac{1}{3}$. This is analogous to the situation when $q\geq\frac{1+2\sqrt 2}{7}$. See Remark \ref{rmk:decay:1/3}.
\end{rmk}
\begin{rmk}[Definition of the energy flux]\label{rmk:DefEnergy}
    The twisted energy flux considered in this paper is defined as
\begin{align}\label{eq:intro:energy-sigma-tau}
        \begin{split}
        \E_q^\varepsilon[\psi] (\tau):= &\int_{\S_{\tau}}t^{-\mu_q^\varepsilon}t^{2q} \left( |\d_t(t^{q}\psi)|^2+|\d_r\psi|^2+\frac{1}{r^2}|\d_r(r\psi)|^2+\frac{1}{r^2}|\rslnabla\psi|^2+  t^{2q-2}\psi^2 + \frac{1}{1+r^2}\psi^2\right)\,dx \\
            &+\int_{\C_{\tau}} t^{-\mu_q^\varepsilon}\left(\frac{1}{r^2}|\d_v(rt^q\psi)|^2+|\d_v(t^q\psi)|^2+\frac{t^{2q}}{r^2}|\rslnabla\psi|^2+t^{4q-2} \psi^2+ \frac{t^{2q}}{r^2}\psi^2\right)\,r^2dvd\omega\,,
        \end{split}
\end{align}
with $\mu_q^\varepsilon$ as in \eqref{eq:intro:mu} for any $\varepsilon\geq 0$. For definition of null and angular derivatives refer to Section \ref{sec:metric-wave-foliation}.
\end{rmk}
\begin{rmk}[Non-sharp decay from energy boundedness]\label{rmk:intro:decay-from-EB}
    Note that a non-sharp pointwise decay can be induced just from the energy boundedness, Theorem \ref{thm:intro:boundedness}. For the sake of simplicity, we consider a simplified formulation of energy flux defined on the spacelike hypersurfaces $\{t=\text{constant}\}$ rather than on $\Sigma_\tau$. Then, for $0<q<1$, any $t_2>t_1$, and for any solution $\psi$ originating from compactly supported initial data on $\{t=\text{constant}\}$, the following energy boundedness holds\footnote{with $|\nabla \psi|^2=\sum_{i=1,2,3}|\d_{x^i}\psi|^2$.}
    \begin{align}\label{eq:intro:standard-energy}
    \begin{split}
        \int_{\{t=t_2\}} &t^{2q-\mu_q}\left(|\d_t (t^q\psi)|^2+|\nabla \psi|^2+t^{2q-2}|2q-1|\psi^2 \right)\,dx \\&\lesssim\int_{\{t=t_1\}} t^{2q-\mu_q}\left(|\d_t (t^q\psi)|^2+|\nabla \psi|^2+t^{2q-2}|2q-1|\psi^2 \right)\,dx \,.
    \end{split}
    \end{align}
    The estimate above obviously implies the decay of an unweighted energy flux. Furthermore, commuting with $\d_{x^i}$ for $i=1,2,3$ and using Sobolev inequalities, Proposition \ref{prop:general:Sobolev}, the above estimate immediately yields the following pointwise decay
    \begin{align*}
        |\psi(t,x)|&\lesssim \frac{
            \tilde \E_q[\psi_0,\psi_1]}{ t^{q}}\quad &\text{ if }\quad 0<q\leq \frac{1}{2}\,,\\
          |\psi(t,x)|&\lesssim\frac{ 
            \tilde \E_q[\psi_0,\psi_1]}{ t^{1-q}}\quad &\text{ if }\quad \frac{1}{2}<q<1\,,
    \end{align*}
    with $\tilde \E_q[\psi_0,\psi_1]$ denoting a suitable weighted higher order initial energy on $\{t=t_0\}$. Observe that the decay rate described above is weaker than those established in Corollary \ref{cor:general:decay}. Since the standard energy fluxes along the spacelike hypersurfaces $\{t=\text{constant}\}$ in Minkowski spacetime are conserved,
    the energy decay \eqref{eq:intro:standard-energy} in FLRW spacetimes can be interpreted as a consequence of \emph{cosmological expansion}. On the other hand, the stronger decay presented in Corollary \ref{cor:general:decay} additionally incorporates the \emph{dispersive} properties of waves. See also Remark \ref{rmk:torus}.

   Moreover, by refining the use of Sobolev inequalities and commutators, the pointwise decay estimates above can be improved while still relying only on the energy boundedness result, namely Theorem \ref{thm:intro:boundedness}. Indeed we have,
    \begin{align*}
        |\psi(t,x)|&\lesssim \frac{
            \tilde \E_q[\psi_0,\psi_1]}{ t^{q}\sqrt{1+r}}\quad &\text{ if }\quad 0<q\leq \frac{1}{2}\,,\\
          |\psi(t,x)|&\lesssim\frac{ 
            \tilde \E_q[\psi_0,\psi_1]}{ t^{1-q}\sqrt{1+r}}\quad &\text{ if }\quad \frac{1}{2}<q<1\,.
    \end{align*}
    See Remark \ref{rmk:decay:pointwise-general}.
\end{rmk}
\begin{rmk}[Waves on FLRW with $\TT^3$ spatial topology]\label{rmk:torus}
    Estimate \eqref{eq:intro:standard-energy} remains valid even when the spatial topology is changed from $\RR^3$ to \emph{compact} $\TT^3$ topology, where dispersive effects are absent. Again commuting with $\d_{x^i}$ and using the Sobolev inequality \eqref{eq:general:Sobolev:TT} yield
    \begin{align*}
        \Big\vert\psi(t,x)-\frac{1}{|\TT^3|}\int_{\TT^3}\psi\,dx\Big\vert&\lesssim \frac{
            \tilde \E_q[\psi_0,\psi_1]}{ t^{q}}\quad &\text{ if }\quad 0<q\leq \frac{1}{2}\,,\\
          \Big\vert\psi(t,x)-\frac{1}{|\TT^3|}\int_{\TT^3}\psi\,dx\Big\vert&\lesssim\frac{ 
            \tilde \E_q[\psi_0,\psi_1]}{ t^{1-q}}\quad &\text{ if }\quad \frac{1}{2}<q<1\,.
    \end{align*}
    We remark that subtracting the average is essential, as the constant functions satisfy the wave equation.
\end{rmk}
\subsection{Related works}
There have been a number of previous works on wave-type equations on FLRW spacetimes \eqref{eq:intro:Metric-x}, some of which we briefly mention below.
\subsubsection{Explicit representation formulae via transformation}
The mathematical study of the wave equation on FLRW spacetimes can be traced back to Klainerman--Sarnak's work \cite{KlainermanSarnak} on the dust model, \textit{i.\@e.\@}\ $q=\frac{2}{3}$, with spatial curvature $K\in\{-1,0,1\}$, where they provided explicit representation formula for the wave equation. In the spatially flat case, they considered the transformation $\Psi=t^\frac{1}{3}\d_t(t\psi)$ and reduced the wave equation on FLRW to the standard wave equation on Minkowski, where a representation formula in terms of spherical means is given. Building on this representation formula for the FLRW dust model, \cite{AbbasiCraig} later showed that waves do not necessarily satisfy the strong Huygens principle, and moreover, in the flat case, they decay at the rate $t^{-1}$. See also \cites{GKY10WaveDust, GalstianYagdjian14WaveDust}.

The transformation method in \cite{KlainermanSarnak} has been applied to more general FLRW spacetimes in \cite{NatarioRossetti}, where Nat{\'a}rio--Rossetti found explicit representation formula for solutions of wave equation in a number of cosmological spacetimes including the radiation-filled FLRW universe, \textit{i.\@e.\@}\ $q=\frac{1}{2}$. In particular, they proved the $t^{-1}$ decay rate for waves in the spatially flat radiation case by transforming the equation to the wave equation on Minkowski spacetime where representation formulas are available. Furthermore, Nat{\'a}rio--Rossetti, using Wirth's work \cite{Wirth04} on ``weakly dissipative waves'', showed that solutions to the wave equation $\Box_{g_q} \psi=0$ on the spatially flat FLRW spacetimes arising from initial data belonging to appropriate Sobolev spaces decay as
\begin{align}\label{eq:intro:decay-rate}
    |\psi(t,\cdot)|\lesssim \begin{cases}
        t^{-1}\qquad &0\leq q\leq \frac{2}{3}\,,\\
        t^{-3(1-q)}\qquad &\frac{2}{3}<q<1\,,\\
        (\log t)^{-\frac{2}{3}}\qquad &q=1\,.
    \end{cases}
\end{align}
They also formulated a similar conjecture in the spatially hyperbolic case. The decay rate \eqref{eq:intro:decay-rate} across the entire decelerated regime was proved in \cite{NatarioRossetti} by transforming the wave equation on FLRW spacetimes into a wave equation with \emph{weak dissipation} on Minkowski space, for which the results were already known from \cite{Wirth04}. The details of this approach are presented in the following lines.

\subsubsection{Conformal weakly dissipative wave equation}
Observe that by defining the \emph{conformal time} coordinate 
$$t^*=\int \frac{dt}{a(t)}=\frac{t^{1-q}}{1-q},$$
the metric \eqref{eq:intro:Metric-x} becomes
\begin{align*}
    g_q=a^2(t^*)\left(-(dt^*)^2  + (dx^1)^2+(dx^2)^2+(dx^3)^2\right)\,,
\end{align*}
which clearly shows that FLRW spacetimes are \emph{conformally flat}. In $(t^*,x^1,x^2,x^3)$ coordinates, 
the wave operator turns to
\begin{align}\label{eq:intro:wave-operator-conformal}
    \Box_{g_q} \Psi = \frac{1}{t^{2q}}\left( -\d_{t^*}\d_{t^*} \Psi -\frac{2q}{(1-q) t^*}\d_{t^*} \Psi +\Delta_{\RR^3}\Psi \right)\,.
\end{align}
One therefore can study the ``weakly dissipative'' wave equation
\begin{equation}\label{eq:intro:weakly-dissipative-wave}
    -\d_{t^*}\d_{t^*} \Psi -\frac{\alpha}{ t^*}\d_{t^*} \Psi +\Delta_{\RR^3}\Psi=0\,,
\end{equation}
 on Minkowski spacetime for $t^*\geq 1$ and with $\alpha=\frac{2q}{1-q}$, instead of studying $\Box_{g_q} \psi=0$. The study of equation \eqref{eq:intro:weakly-dissipative-wave} goes back to \cite{Matsumura} and \cite{Uesaka}. In fact, Uesaka \cite{Uesaka} proved that if $\Psi$ solves equation \eqref{eq:intro:weakly-dissipative-wave} with compactly supported initial data on $\{t^*={\text{constant}}\}$, then the standard energy decays as
 $$\int_{\RR^3} \left(|\d_{t^*}\Psi|^2+|\nabla \Psi |^2\right)\,dx\lesssim  \frac{\fC}{{t^*}^{\min \{2,\alpha\}}}\,,$$
 which corresponds to $t^{-\min\{2q,2(1-q)\}}$ decay rate for energy of solution $\psi$ of the FLRW wave equation.
 It is worth mentioning that Uesaka's result and method are comparable and compatible with Theorem \ref{thm:intro:boundedness} and its proof. See Remark \ref{rmk:intro:decay-from-EB}.

Wirth \cite{Wirth04} showed that  the Fourier image $\hat \Psi (t^*,\xi)$ of the solution $\Psi$ solves an ODE that can be transformed to Bessel's differential equation, and therefore,  an explicit representation of solutions to \eqref{eq:intro:weakly-dissipative-wave} in terms of Bessel functions was derived. In \cite{Wirth04}, the special structure of the representation and the properties of the Bessel functions were carefully exploited to obtain $L^{\bar p}$–$L^{\bar q}$ estimates for the solution. The decay rate \eqref{eq:intro:decay-rate} was proved in \cite{NatarioRossetti} by letting $\bar p=1$ and $\bar q=\infty$ in this $L^{\bar p}$–$L^{\bar q}$ estimate.

 For equation \eqref{eq:intro:weakly-dissipative-wave} with $|\Psi|^\beta$-non-linearity, refer to \cite{He25SemiLinearDamping}*{and references therein}.
\subsubsection{Numerical works}
The numerical work of Rossetti--Va{\~n}{\'o}-Vi{\~n}uales \cite{RossettiNumerical} considers the spherically symmetric solutions of the FLRW wave equation with Gaussian pulses as initial data. For the spatially flat case, their numerical data is in agreement with the decay rate presented in \eqref{eq:intro:decay-rate} and suggests that the time derivative of $\psi$ decays as
\begin{align*}
    |\d_t\psi(t,\cdot)|\sim \begin{cases}
        t^{-1-q}\qquad &0\leq q\leq \frac{3}{4}\,,\\
        t^{3q-4}\qquad &\frac{3}{4}<q<1\,,\\
        t^{-1}(\log t)^{-\frac{5}{2}}\qquad &q=1\,,\\
        t^{1-2q}\, \qquad &q>1\,.
    \end{cases}
\end{align*}
In addition, \cite{RossettiNumerical} suggests the following decay rate for the ``tail'' of the wave in the interior region for $q\neq \frac{1}{2}$
\begin{align*}
    \sup_{\{r<\delta t^*\}} |\psi|\sim \frac{1}{t^{3-3q}}\,, \qquad  \sup_{\{r<\delta t^*\}} |\d_t\psi|\sim \frac{1}{t^{4-3q}}\,, 
\end{align*}
with $\delta \in (0,1)$. Note that in the radiation case, the strong Huygens principle holds.

Also see the recent paper \cite{Rosseti25hyperboloidal}.
\subsubsection{Non-linear stability and instability results}
We briefly mention results on non-linear wave-type problems around FLRW spacetimes that are particularly relevant to the study of equation \eqref{eq:intro:wave-equation}.

It has been shown by Taylor \cite{Taylor24FLRWRadiation} that spacetime \eqref{eq:intro:Metric-x} (with $q=\frac{1}{2})$ is future stable as a solution to the spherically symmetric Einstein--massless Vlasov equation.

Speck \cite{Speck2013RelEuler} considered the relativistic Euler equation with a linear equation of state on fixed FLRW spacetimes $M=[1,\infty)_t\times \RR^3$ with metrics including \eqref{eq:intro:Metric-x}. In the decelerated regime, where the scale factor $a(t)$ satisfies $\ddot a<0$, he showed that when the \emph{sound speed} $c_s$ vanishes (the dust model), homogeneous isotropic ``background'' solutions\footnote{If the metric background is given by \eqref{eq:intro:Metric-x}, then the homogeneous isotropic background solutions are given by the four-velocity $v=\d_t$ and energy density $\rho=\bar\rho t^{-3q(1+c_s^2)}$ with $\bar \rho >0$ and $c_s$ denoting the sound speed.} are stable if the scale factor $a(t)$ grows suitably fast (including \eqref{eq:intro:Metric-x} with $q> \frac{1}{2}$). He also noted that the case of the radiation equation of state can be conformally transformed to Minkowski space where the celebrated work of Christodoulou \cite{Christodoulou07ShockFormation} on the \emph{shock formation} can be applied. In particular, when $c_s^2=\frac{1}{3}$, the background solutions are unstable for the entire family of decelerated expanding metrics \eqref{eq:intro:Metric-x}.
Moreover, Speck \cite{Speck2013RelEuler} showed the stability for a wider range of sound speeds in the \emph{accelerated regime}, where $\ddot a(t)>0$.

Similar results in a slightly different setting have been recently proved. On the background manifold $M=(1,\infty)_t\times \TT^3$ with the metric \eqref{eq:intro:Metric-x}, Fajman, Ofner, Oliynyk, and Wyatt  \cite{FOOW25StabilityEuler} showed that the homogenous background solutions are \emph{stable} if $c_s^2<1-\frac{2}{3q}$. When this condition is violated, numerical evidence suggests instability; \cite{FMOOW24PhaseTransition}. See also the numerical study \cite{Marshall25NumericEE} of the coupled Einstein--Euler equation.

In the accelerated regime, $\ddot a>0$, where the fast expansion can help with the ``stabilisation'', more stability results are known. For the stability of Euler--Einstein near FLRW spacetimes with positive $\Lambda$ refer to \cites{RodnianskiSpeck09irrotational, Speck12StabilityEE, HadvzicSpeck15Dust, LubbeKroon13Conformal, Oliynyk16PositiveLambda}. See also \cite{Fournodavlos22ScalarField} and \cite{Ringstorm13FutureUniverse}. For stability of the flat Milne spacetime with linear expansion, see \cites{AnderssonMoncrief11VacuumMilne, AnderssonFajman20MilneEV, BFK19KaluzaKlienMilne, BarzegarFajman22Milne,  FOW2024MilneDust}.

See also the recent papers \cite{TW25SemilinearWave} and \cite{NY25KlienGordon} on semi-linear wave and Klien--Gordon equations in FLRW spacetimes. See \cite{CNO19WaveAccelerated} for the study of the wave equation in FLRW spacetimes with accelerated expansion, \textit{i.\@e.\@}\ $\dot a(t)>0$.
\subsection{Overview of proofs of main results}
In this section, we provide an overview of the proof and the fundamental ideas underlying the main results in the current paper, notably Theorems \ref{thm:intro:boundedness}, \ref{thm:itro:Morawetz} and \ref{thm:intro:rp}. Moreover, we compare the radiation case to the well known results on Minkowski space. Then, we demonstrate the differences between the cases $q<\frac{1}{2}$ and $q>\frac{1}{2}$.

For the sake of simplicity, we here consider simplified energy fluxes on hypersurfaces of the form $\{t=\text{constant}\}$, similar to those in Remark \ref{rmk:intro:decay-from-EB}. Accordingly, set $\R=[t_1,t_2]\times\RR^3$.
\subsubsection{Waves on Minkowski space}
To begin with, we recall some established results regarding the boundedness and decay of waves on Minkowski space, formally corresponding to $q=0$ in \eqref{eq:intro:Metric-x}. Suppose $\psi$ is a solution to the standard wave equation on Minkowski space originated from compactly supported initial data on $\{t=t_0\}$. The notation we use here for the coordinate functions and regions in Minkowski space is the same as that used in FLRW spacetimes, formally setting $q=0$.

We first recall the \emph{energy-momentum tensor} $\TT_{\mu\nu}[\psi]$ associated to $\psi$, defined as 
\begin{equation*}
    \TT_{\mu\nu}[\psi]=\nabla_\mu \psi \nabla_\nu \psi -\frac{1}{2}g_{\mu\nu}\nabla^\gamma \psi \nabla_\gamma \psi\,.
\end{equation*}
If one lets 
$J^X_\mu [\psi ] = \TT_{\mu \nu} [\psi] X^\nu$ for any smooth vector field $X$, then one is left with the \emph{divergence identity}
\begin{equation}\label{eq:intro:divergence-identity}
     \nabla^\mu J^X_\mu [\psi] =  \nabla_X\psi \Box_g \psi+ {}^{(X)}\pi^{\mu \nu} \TT_{\mu \nu} [\psi]  \,,
\end{equation}
with ${}^{(X)}\pi_{\mu \nu}$ denoting the deformation tensor of $X$, defined in Section \ref{sec:TEMT}. The former term on the right-hand side obviously vanishes for solutions to the wave equation. After suitable choices of vector field $X$, divergence identity \eqref{eq:intro:divergence-identity} leads to energy estimates for solutions $\psi$ to the wave equation after integration over suitable spacetime regions.

By letting $X$ to be the Killing timelike vector field $T=\d_t$, \eqref{eq:intro:divergence-identity} leads to the following conservation law  after integrating over the region $\R$,
\begin{equation}\label{eq:intro:EB-Mink}
        \int_{\{t=t_2\}}  \left( |\d_t  \psi|^2 +  |\nabla \psi|^2\right) \, dx= \int_{\{t=t_1\}}  \left( |\d_t \psi|^2 +  |\nabla \psi|^2\right)\, dx\,.
\end{equation}
The above identity should be compared to \eqref{eq:intro:standard-energy} and Theorem \ref{thm:intro:boundedness}.
Similarly, a modified version of the divergence identity \eqref{eq:intro:divergence-identity} can also be used to derive the ILED estimate for $\psi$. Indeed, vector field multipliers $X_1=\d_r$ and $X_2=\frac{1}{(1+r)^\delta} \d_r$ for an arbitrarily  small $\delta>0$ along with suitable \emph{modification functions}, give rise to
\begin{align}\label{eq:intro:ILED-Mink}
    \int_{\R} \left(\frac{1}{r^3}|\rslnabla \psi|^2 + \frac{1}{1+r^{1+\delta}} \left( |\d_t  \psi|^2 + |\d_r\psi|^2 \right) + \frac{1}{1+ r^{3+\delta}} \psi ^2 \right)\; dxdt&\lesssim \int_{\{t=t_1\}}  \left( |\d_t \psi|^2 +  |\nabla \psi|^2\right)\,dx\,.
\end{align}
Note that the energy boundedness \eqref{eq:intro:EB-Mink} and the ILED estimate \eqref{eq:intro:ILED-Mink} are consequences of the (modified) divergence identity \eqref{eq:intro:divergence-identity}, which is a pointwise identity. Therefore, the choice of the integration domain is flexible, provided that its boundary components are either \emph{null} or \emph{spacelike}. In particular, estimates \eqref{eq:intro:EB-Mink} and \eqref{eq:intro:ILED-Mink} can likewise be formulated for partially null hypersurfaces $\Sigma_\tau$, the bulk region $\R^{\tau_2}_{\tau_1}$, and the energy flux $\E[\psi](\tau)$ associated to energy current $J^T_\mu[\psi]$. 

To derive the Dafermos-Rodnianski $r^p$-estimates on Minkowski \cite{DafRodNewMethod}, one can use the divergence identity \eqref{eq:intro:divergence-identity} for the energy current $J^X_\mu[r\psi]$ with $X=r^{p-1}\d_v$. After integration over the region $\D^{\tau_2}_{\tau_1}$, as defined in Theorem \ref{thm:intro:rp}, this choice of vector field multiplier along with suitable modification functions yields
\begin{align}\label{eq:intro:rp-Mink}
    \begin{split}
        \int_{\D^{\tau_2}_{\tau_1}} r^{p-1} \left( p|\d_v (r\psi)|^2 + \frac{p-2}{r^2}|\rslnabla(r\psi)|^2\right) \, dudvd\omega + \int_{\C_{\tau_2}} r^p|\d_v(r\psi)|^2 \, dvd\omega \\ \lesssim \int_{\C_{\tau_1}} r^p|\d_v(r\psi)|^2 \, dvd\omega +  \E[\psi] (\tau_1)\,,
    \end{split}
\end{align}
for any $0\leq p \leq 2$, in which $\E[\psi] (\tau_1)$ is the energy flux associated to vector field $T$ on $\Sigma_{\tau_1}$. Note that in the estimate above, other terms arising from boundary components of $\D^{\tau_2}_{\tau_1}$ have been controlled by \eqref{eq:intro:EB-Mink} and \eqref{eq:intro:ILED-Mink}.

Once estimates \eqref{eq:intro:EB-Mink}, \eqref{eq:intro:ILED-Mink} and \eqref{eq:intro:rp-Mink} are achieved, a general hierarchical lemma, see Lemma \ref{lem:general:hierarchy-estimate}, gives 
\begin{equation*}
    \E[\psi] (\tau)\lesssim \frac{1}{\tau^2}\E[\psi_0,\psi_1]\,,
\end{equation*}
for any $\tau>\tau_0$ with $\bar\E[\psi_0,\psi_1]$ being an appropriate weighted initial energy of $\psi$. The pointwise decay of $\psi$ then follows from standard Sobolev inequalities (see Proposition \ref{prop:general:Sobolev} later) and commuting with angular derivatives $\Omega_{i}$ for $i=1,2,3$. In fact, we have
\begin{equation*}
    |\psi|\lesssim \frac{\bar\E[\psi_0,\psi_1]}{(1+r)\sqrt \tau},\qquad |\psi|\lesssim \frac{\bar\E[\psi_0,\psi_1]}{\sqrt {1+r} \tau}\,,
\end{equation*}
in which $\bar\E[\psi_0,\psi_1]$ is an appropriate weighted initial energy containing derivatives of $\psi$ up to order three.
\subsubsection{The radiation case $q=\frac{1}{2}$}\label{subsubsec:radiation}
In FLRW spacetime, the energy flux under consideration is modified by introducing a twisted derivative, replacing $\d_t\psi$ with $\frac{1}{t^q}\d_t(t^q \psi)$.\footnote{This adjustment is motivated by the underlying structure of the wave equation in the radiation case, where such a formulation naturally arises. Indeed, the wave equation \eqref{eq:intro:wave-equation} can be rewritten for $q=\frac{1}{2}$ as
\begin{equation*}
    \Box_{\frac{1}{2}} \psi= -\frac{1}{t} \d_t(\sqrt t \d_t(\sqrt t \psi)) + \frac{1}{t} \Delta_{\RR^3} \psi=0\,.
\end{equation*}}
This is done by introducing the \emph{twisted energy-momentum tensor}
\begin{align}\label{eq:intro:overview:TEMT}
    \tilde{\TT}^\beta_{\mu\nu}[\psi]:=\beta^2 \left[ {\nabla}_\mu (\beta^{-1}\psi) {\nabla}_\nu (\beta^{-1}\psi) - \frac{1}{2} g_{\mu \nu} \left( {\nabla}^\gamma (\beta^{-1} \psi) {\nabla}_\gamma (\beta^{-1} \psi) + 
    \V_\beta \beta^{-2} \psi^2 \right)\right]\,,
\end{align}
for a \emph{twisting function} $\beta$ and with the \emph{potential} $\V_\beta=-\frac{\Box_g \beta}{\beta}$. The twisted version of the divergence identity \eqref{eq:intro:divergence-identity} for vector field $X$ and the \emph{twisted energy current} $\tilde{J}^{\beta,X}_\mu [\psi]=\tilde{\TT}^\beta_{\mu \nu} [\psi] X^\nu$ then reads
\begin{equation}\label{eq:intro:twisted-divergence-identity}
    \nabla^\mu \tilde{J}^{\beta,X}_\mu [\psi] = \beta \nabla_X(\frac{\psi}{\beta}) \Box_g \psi+ {}^{(X)}\pi^{\mu \nu} \tilde{\TT}^\beta_{\mu \nu} [\psi] - \frac{1}{2} X^\nu \left(\frac{\nabla_\mu (\beta^2 \V_\beta)}{\beta^2} \psi^2 + \nabla_\mu( \beta^2)  {\nabla}^\gamma (\beta^{-1} \psi) {\nabla}_\gamma (\beta^{-1} \psi)\right) \,,
\end{equation}
which again leads to energy estimates for solutions $\psi$ after suitable choices of vector field $X$ and twisting function $\beta$ and integration over suitable spacetime regions.

For the case $q=\frac{1}{2}$, in fact, if we let $\beta=\frac{1}{a(t)}=\frac{1}{\sqrt t}$ and $X=\sqrt t \d_t$, then the terms on the right-hand side of \eqref{eq:intro:twisted-divergence-identity} sum to zero, yielding the following conservation law
\begin{equation}\label{eq:intro:radiation:conservedE}
        \int_{\{t=t_2\}} t \left( |\d_t(\sqrt t \psi)|^2 +  |\nabla \psi|^2\right) \, dx= \int_{\{t=t_1\}} t \left( |\d_t(\sqrt t \psi)|^2 +  |\nabla \psi|^2\right)\, dx\,.
\end{equation}
Note that the potential $\V_\beta$ vanishes for $\beta=\frac{1}{\sqrt t}$.

In order to establish the ILED estimate \eqref{eq:intro:ILED} for $q=\frac{1}{2}$, we again use \eqref{eq:intro:twisted-divergence-identity} with $\beta=\frac{1}{\sqrt t}$, vector field multipliers $X_1=\d_r$ and $X_2=\frac{1}{(1+r)^\delta} \d_r$, and suitable modification functions.

The $r^p$-estimates can also be derived  in a manner similar to the Minkowski case. Indeed, using \eqref{eq:intro:twisted-divergence-identity} for the twisted energy current $\tilde J^X_\mu[r\psi]$ with $X=\sqrt t r^{p-1}\d_v$---where $\partial_v$ is defined with respect to the FLRW metric, and thus differs from its Minkowski counterpart---together with suitable modification functions, yields \eqref{eq:intro:rpEstimate} for $0 \leq p \leq 2$.

The energy decay estimate and the pointwise decay then follow similarly from the abstract hierarchical lemma \ref{lem:general:hierarchy-estimate} and Sobolev inequalities.
\subsubsection{The case $q\neq\frac{1}{2}$}
In the general case, there is no obvious choice of $\beta$ and $X$ that makes the right-hand side of \eqref{eq:intro:twisted-divergence-identity} vanish for solutions of the wave equation---a property that makes the case $q = \frac{1}{2}$ comparatively the simplest. If we let $\beta=\frac{1}{a(t)}=\frac{1}{t^q}$, the potential $\V_\beta= -\frac{q(2q-1)}{t^2}$ is no longer vanishing. Furthermore, $V_\beta$ changes sign as $q$ crosses $\frac{1}{2}$. This change in the sign of the potential $\V_\beta$ constitutes the primary distinction between the cases $q\leq\frac{1}{2}$ and $q>\frac{1}{2}$. In particular, if one let $X$ in \eqref{eq:intro:twisted-divergence-identity} to be $T=t^q\d_t$, then we have
\begin{equation*}
     \nabla^\mu \tilde{J}^{\beta,X}_\mu [\psi]= (1-q) t^{q-1}\V_\beta \psi^2\,,
\end{equation*}
for any wave solution $\psi$. After integration over $\R$, we arrive at
\begin{align*}
    &(1-q)\int_\R  t^{q-1} \V_\beta\psi^2\,t^{3q} dtdx+ \frac{1}{2}\int_{\{t=t_2\}} \left(t^{2q} |\d_t(t^q\psi)|^2 + t^{2q}|\nabla \psi|^2 + t^{4q}\V_\beta  \psi^2 \right)\, dx\\
        &=\frac{1}{2}\int_{\{t=t_1\}} \left(t^{2q} |\d_t(t^q\psi)|^2 +t^{2q} |\nabla \psi|^2 + t^{4q}\V_\beta  \psi^2\right) \, dx\,.
\end{align*}
The above identity clearly establishes a uniform energy bound for the case $q\leq\frac{1}{2}$, as $\V_\beta\geq0$, and coincides with the conservation law \eqref{eq:intro:radiation:conservedE} when $q=\frac{1}{2}$. However, it falls short to establish the energy boundedness estimate in the case $q>\frac{1}{2}$. So, we introduce an alternative twisting function for $q>\frac{1}{2}$. In fact, if we let $\beta'=t^{q-1}$, then \eqref{eq:intro:twisted-divergence-identity} for $X=t^{2-3q}\d_t$ along with suitable modifications leads to the following energy estimate
\begin{align*}
    &\int_\R t|\d_t \psi|^2 \, dtdx+ \int_{\{t=t_2\}} \left(t^{2q}|\d_t(t^{1-q}\psi)|^2 + t^{2-2q}|\nabla \psi|^2 + \psi^2 \right)\, dx\\
        &\lesssim \int_{\{t=t_1\}}  \left(t^{2q} |\d_t(t^{1-q}\psi)|^2 + t^{2-2q}|\nabla \psi|^2 + \psi^2 \right) \, dx\,.
\end{align*}
Note that once the above estimate is established, it can be rewritten in terms of the twisting function $\beta=\frac{1}{t^q}$ with different weight factors involving $\mu_q$, as in \eqref{eq:intro:energy-sigma-tau} and \eqref{eq:intro:standard-energy}.

The derivation of the ILED estimate \eqref{eq:intro:ILED} again relies on the divergence identity \eqref{eq:intro:twisted-divergence-identity} for $\beta=\frac{1}{t^q}$, using vector field multipliers $X_1=\d_r$ and $X_2=\frac{1}{(1+r)^\delta} \d_r$ in the case $q<\frac{1}{2}$. When $q>\frac{1}{2}$, we use multipliers $X'_1=t^{2-4q}\d_r$ and $X'_2=\frac{t^{2-4q}}{(1+r)^\delta} \d_r$, along with suitable modifications.

To derive the $r^p$-estimates, once again we use \eqref{eq:intro:twisted-divergence-identity} for the twisted energy current $\tilde J^X_\mu[r\psi]$ with $X=t^{-\mu_q}t^q r^{p-1}\d_v$. After integrating over $\D^{\tau_2}_{\tau_1}$, introducing suitable modifications, and controlling the boundary terms, we arrive at 
\begin{align*}
    &\int_{\D^{\tau_2}_{\tau_1}}t^{-{\mu_q}}r^{p-1} \left( p|\d_v \varphi|^2 + \frac{2-p}{r^2}|\rslnabla\varphi|^2+\frac{\gamma (2-p)}{r^2}\varphi^2\right) \, dudvd\omega+ \int_{\C_{\tau_2}} t^{-{\mu_q}}r^p|\d_v\varphi|^2 \, dvd\omega\\&\leq  \int_{\D^{\tau_2}_{\tau_1}} 2t^{-\mu_q}r^{p-2}(\gamma + r^2t^{2}\V_\beta) |\varphi|| \d_v \varphi| \, dudvd\omega  +\int_{\C_{\tau_1}} t^{-{\mu_q}}r^p|\d_v\varphi|^2 \, dvd\omega+\fC\E_q[\psi](\tau_1)\,,
\end{align*}
in which $\varphi=t^qr\psi$, $\gamma=\frac{q|1-2q|}{|1-q|^2}$, and $\fC>0$. The former term on the right-hand side should be viewed as an error term that needs to be controlled by the spacetime terms on the left-hand side. This error term involves both $t$ and $r$-dependent weight factors, requiring us to consider it separately in different spacetime regions. To do so, in the region $\{u>0\}$, we use the fact that $rt^{q-1}<\frac{1}{1-q}$. While in the region $\{u\leq0\}$ with bounded $\tau$, we use Grönwall's inequality.

Observe that the power of $r$- and $t$-weight factors are functions of $p$ and $q$, respectively. As a result, estimating terms with weights depending on both $t$ and $r$ imposes a dependence of the admissible range of $p$ on $q$, a relation encoded by $\sigma_q$ in Theorem \ref{thm:intro:rp}. Recall that in the cases $q=0$ and $q=\frac{1}{2}$, there is no error term and $\sigma_q=0$. For more details, refer to Section \ref{sec:EnergyEstimates}.

The energy decay and the pointwise decay estimates then follow as before. See Section \ref{sec:decay} for more details.

\subsection{Organisation of the paper}
Section \ref{se:preliminaries} provides all the required material for the rest of the paper, such as necessary information about the wave equation on FLRW spacetimes, definitions of coordinate functions, and various general functional inequalities.
A detailed proof of Theorems \ref{thm:intro:boundedness}, \ref{thm:itro:Morawetz} and \ref{thm:intro:rp} is given in Section \ref{sec:EnergyEstimates}. The proofs are divided into cases $q\leq\frac{1}{2}$ and $q>\frac{1}{2}$. Once Theorems \ref{thm:intro:boundedness}, \ref{thm:itro:Morawetz} and \ref{thm:intro:rp} are obtained, the results in Corollary \ref{cor:general:decay} then follow from general functional inequalities provided in Section \ref{sec:functional_ineq}. A detailed explanation of this derivation is given in Section \ref{sec:decay}.

 \subsection*{Acknowledgments}
 I am grateful to my advisor Martin Taylor for his invaluable support and guidance throughout this project and for introducing me to this area.
 This work was supported by the Engineering and Physical Sciences Research Council [EP/S021590/1]. The EPSRC Centre for Doctoral Training in Geometry and Number Theory (The London School of Geometry and Number Theory), University College London. I am also thankful to Imperial College London.
\section{Preliminaries}\label{se:preliminaries}
In this section, we present the prerequisites of the paper. Section \ref{sec:TEMT} introduces the (twisted) energy-momentum tensor and the associated divergence identities, which serve as key tools in the following sections. We then define the coordinate functions and the foliation of the spacetime in Section \ref{sec:metric-wave-foliation}. Section \ref{sec:wave-FLRW} is dedicated to the wave equation on FLRW spacetimes and its twisted representation. Lastly, in Section \ref{sec:functional_ineq}, we present all the general functional inequalities that we use in other sections.

\subsection{The twisted energy-momentum tensor}\label{sec:TEMT}
In this section, we consider the covariant wave equation $\Box_g\psi=g^{\mu\nu}\nabla_\mu\nabla_\nu \psi=0$ for a general Lorentzian metric $g$ on a 4-dimensional manifold $\M$. Recall the \emph{energy-momentum tensor} $\TT_{\mu\nu}[\psi]$ associated to $\psi$, defined as 
\begin{equation*}
    \TT_{\mu\nu}[\psi]=\nabla_\mu \psi \nabla_\nu \psi -\frac{1}{2}g_{\mu\nu}\nabla^\gamma \psi \nabla_\gamma \psi\,,
\end{equation*}
and the \emph{divergence identity} $\nabla^\nu \TT_{\mu\nu}[\psi]= \d_\mu \psi \Box_g \psi$ which is of great importance for proving energy estimates in the this paper; see \cite{DHRTQuasilinearWave}*{Section~3.3}.
We here introduce the twisted energy-momentum tensor $\tilde \TT^\beta_{\mu\nu}[\psi]$ associated with $\psi$ and the twisting function $\beta$, which will play a crucial role in the following sections. See \cite{DHRTQuasilinearWave}*{Appendix~ A.2} for the proof of identities stated in this section.

To proceed, consider a non-zero \emph{twisting function} $\beta$ and define the twisted covariant derivatives as 
\begin{equation*}\label{eq:twisted-covariant-derivatives}
    \tilde{\nabla}_\mu (\,\cdot \,):= \beta \nabla_\mu (\beta^{-1} \,\cdot \,), \qquad \tilde{\nabla}^\dagger_\mu(\,\cdot \,):= - \beta^{-1}\nabla_\mu (\beta \,\cdot \,)\,.
\end{equation*}
With these new operators, one can rewrite the wave equation as
\begin{equation}\label{eq:general:twisted-wave}
    0=\Box_g \psi= - \tilde\nabla_\mu^\dagger \tilde \nabla^\mu \psi - \V_\beta \psi=g^{\mu\nu}\beta^{-1}\nabla_\mu (\beta^2 \nabla_\nu(\beta^{-1}\psi))- \V_\beta \psi\,,
\end{equation}
in which $\V_\beta=-\frac{\Box_g \beta}{\beta}$. We now define the \emph{twisted energy-momentum tensor} for any function $\psi:\M \to \RR$ as 
\begin{align}\label{eq:twisted-energy-momentum-tensor}
    \tilde{\TT}^\beta_{\mu\nu}[\psi]:=&\tilde{\nabla}_\mu \psi \tilde{\nabla}_\nu \psi - \frac{1}{2}g_{\mu\nu} \left( \tilde{\nabla}^\gamma \psi \tilde{\nabla}_\gamma \psi +\V_\beta \psi^2\right)\\
    =&\beta^2 \left[ {\nabla}_\mu (\beta^{-1}\psi) {\nabla}_\nu (\beta^{-1}\psi) - \frac{1}{2} g_{\mu \nu} \left( {\nabla}^\gamma (\beta^{-1} \psi) {\nabla}_\gamma (\beta^{-1} \psi) + 
    \V_\beta \beta^{-2} \psi^2 \right)\right] \notag.
\end{align}
Note that by letting $\beta=1$, we recover the usual energy-momentum tensor $\TT_{\mu\nu}$. Similar to $\TT_{\mu\nu}[\psi]$, the twisted energy-momentum $\tilde \TT^\beta_{\mu\nu}[\psi]$ benefits from satisfying a corresponding divergence property. Indeed, $\nabla^\mu \tilde \TT^\beta_{\mu\nu}[\psi]= \tilde\nabla_\nu \psi \Box_g \phi + \tilde S^\beta_\nu [\psi]$ where 
\begin{align*}
    \tilde{S}^\beta_\mu [\psi] = \frac{\tilde{\nabla}^\dagger_\mu (\beta \V_\beta)}{2\beta} \psi^2 + \frac{\tilde{\nabla}^\dagger_\mu \beta}{2\beta}  \tilde{\nabla}^\nu \psi \tilde{\nabla}_\nu \psi \,.
\end{align*}
In particular, if $\beta=1$, then $\tilde S ^\beta_\mu[\psi]=0$ and $\TT_{\mu\nu}[\psi]$ is  divergence free provided that $\psi$ solves the wave equation.

For a smooth \textit{vector field multiplier} $X$, we next define the associated \emph{twisted energy current} 
\begin{align*}
    \tilde{J}^{\beta,X}_\mu [\psi ] &= \tilde{\TT}^\beta_{\mu \nu} [\psi] X^\nu \, ,
\end{align*}
and
\begin{equation*}
    \tilde{K}^{\beta,X} [\psi] = {}^{(X)}\pi^{\mu \nu} \tilde{\TT}^\beta_{\mu \nu} [\psi] + X^\nu \tilde{S}^\beta_\nu [\psi] \,  ,
\end{equation*}
with ${}^{(X)}\pi^{\mu \nu}=\frac{1}{2}(\mathcal{L}_X g)_{\mu\nu}$ denoting the \emph{deformation tensor} of $X$, to infer
\begin{align}\label{eq:TEMT:twisted-X-identity}
    \nabla^\mu \tilde{J}^{\beta,X}_\mu [\psi] = \tilde \nabla_X\psi \Box_g \psi+ \tilde{K}^{\beta,X} [\psi] \,.
\end{align}
Moreover, for a smooth \textit{modification function} $w:\M\to \RR$ we have
\begin{align*}
    \nabla^\mu \tilde{J}_\mu^{\beta,w}[\psi]  = w\psi \Box_g\psi + \tilde{K}^{\beta,w}[\psi]\,,
\end{align*}
with
\begin{align*}
     \tilde{K}^{\beta,w}[\psi]  &:= w \beta^2 \nabla^\mu (\beta^{-1} \psi) \nabla_\mu (\beta^{-1} \psi) + \left(\V w  -\beta^{-1} \nabla^\mu \beta \nabla_\mu w - \frac{1}{2}  \Box_g w \right) \psi^2 \, ,  \notag \\
     \tilde J^{\beta,w}_\mu[\psi] &:= w \beta \psi \nabla_\mu (\beta^{-1} \psi) - \frac{1}{2} \psi^2 \nabla_\mu w\,. \notag
\end{align*}
Finally for a vector field multiplier $X$ and a modification function $w$, we define
$\tilde{K}^{\beta,X,w} [\psi]:=\tilde{K}^{\beta,X} [\psi] +\tilde{K}^{\beta,w} [\psi]$, and similarly $\tilde{J}^{\beta,X,w}_\mu[\psi]:=\tilde J^{\beta,X}_\mu [\psi]+\tilde J^{\beta,w}_\mu [\psi]$ to obtain the divergence identity
\begin{equation}\label{eq:TEMT:DivergenceIdentity}
    \nabla^\mu \tilde J ^{\beta, X,w}_\mu [\psi]=\left( \tilde\nabla_X \psi + w\psi \right)\Box_g \psi +\tilde K ^{\beta,X,w} [\psi]\,.
\end{equation}
In the rest of the paper, we drop the dependence on $\beta$ in $\V_\beta$, $\tilde{\TT}^\beta_{\mu\nu}$, $\tilde J^{\beta,X,w}$ and $\tilde K^{\beta,X,w}$ whenever the twisting function $\beta$ is clearly mentioned.

It is worth mentioning that divergence identities above can also be recovered in an untwisted framework, that is, using $\TT_{\mu\nu}[\psi]$ instead of $ \tilde\TT_{\mu\nu}[\psi]$. In fact, this is achieved by modifying the energy currents with $q$ and $\varpi$ as defined in 
\cite{DHRTQuasilinearWave}*{Section~ 3.3}. 
\subsection{Foliation and geometry of FLRW spacetimes}\label{sec:metric-wave-foliation}
The rest of the current article focuses on the manifold
$M= (0,\infty)_t\times \RR^3_x$ with the Lorentizian metric \eqref{eq:intro:Metric-x} which in standard spherical coordinates $(t,r,\theta,\phi)$ can be rewritten as
\begin{align*}\label{eq:intro:qMetric-x}
    g_q=-dt^2  + t^{2q}dr^2+ R^2 dw^2\,,
\end{align*}
in which $d\omega^2$ denotes the round metric on $\SS^2$ and $R(t,r)=t^{q} r$ is the \emph{area radius}. We henceforth suppress the $q$ subscript in $g_q$.
\subsubsection{Double null coordinates}
We define the null coordinate functions $u$ and $v$ as
\begin{align*}
    v=\frac{1}{2}(\frac{t^{1-q}}{1-q}+r)\,,\qquad u=\frac{1}{2}(\frac{t^{1-q}}{1-q}-r)\,.
\end{align*}
Note that $(u,v,\theta,\phi)$ defines a double null coordinate system on $M$, and function $t$ and $r$ can be recovered as
\begin{align*}
    t=\left((1-q)(v+u)\right)^{\frac{1}{1-q}}\,,\qquad r=v-u\,.
\end{align*}
The null coordinates lead to
\begin{align*}
    \d_v=t^q \d_t+\d_r\,, \qquad \d_u=t^q\d_t -\d_r\,, \qquad \d_t=\frac{1}{2t^q}(\d_v+\d_u)\,, \qquad \d_r=\frac{1}{2}(\d_v-\d_u)\,,
\end{align*}
and
\begin{align*}
    dv=\frac{1}{2t^q}dt+\frac{1}{2}dr\,,\qquad du=\frac{1}{2t^q}dt-\frac{1}{2}dr\,.
\end{align*}
Thus, in double null coordinates, the metric takes the form
\begin{equation*}
    g=-4t^{2q}dvdu+R^2 d\omega^2=t^{2q}\left(-4dvdu + r^ d\omega^2\right)\,,
\end{equation*}
which shows that FLRW metrics are conformal to the Minkowski metric. Moreover, the volume form is given by
\begin{align*}
    d\Vol_{\text{FLRW}}=t^{3q} dt dx d\omega= 2t^{4q}r^2 dudvd\omega\,,
\end{align*}
where $d\omega$ is the standard volume form of the unit sphere $\SS^2$.
\subsubsection{$\Sigma_\tau$ and foliation of $M$}
As discussed in Section \ref{sec:IVP}, we foliate the manifold $M$ with hypersurfaces that are partially null. To define this foliation, first let $\rho>0$ and define
\begin{equation*}
    \tau:= \begin{cases}
        \frac{1}{2} \frac{t^{1-q}}{1-q} & r\leq \rho,\\
        u+\frac{1}{2} \rho & r\geq \rho\,.
    \end{cases}
\end{equation*}
Now we foliate the spacetimes by level sets of $\tau$, that is
$$\Sigma_{\tau}:=\underbrace{\left(\{ \frac{t^{1-q}}{2(1-q)}= \tau\}\times\{0\leq r \leq \rho\} \times \SS^2\right)}_{\S_{\tau}} \cup \underbrace{\left(\{u=\tau -\frac{\rho}{2}\}\times\{\rho \leq r \}\times \SS^2\right)}_{\C_{\tau}}\,.$$ 
See Figure \ref{fig:Sigma-tau}. Note that $\rho$ can be arbitrarily large. However, in Section \ref{sec:EnergyEstimates}, we modify $\rho$ by considering $\bar \rho$, which is arbitrarily close to $\rho$. We also let $\R^{\tau_2}_{\tau_1}$ to be the region enclosed by ${\I^{\tau_2}_{\tau_1}}$, $\Sigma_{\tau_1}$ and $\Sigma_{\tau_2}$ with ${\I^{\tau_2}_{\tau_1}}$ denoting $\I^+ \cap \{\tau_1 - \frac{\rho}{2} \leq u \leq \tau_2 - \frac{\rho}{2}\}$. Thus, $$\R^{\tau_2}_{\tau_1}=\cup_{\tau_1 \leq\tau\leq \tau_2}\Sigma_\tau= J^-(\Sigma_{\tau_2})\cap J^+(\Sigma_{\tau_1}).$$
Finally, let $\D^{\tau_2}_{\tau_2}$ denote the far region near null infinity $\R^{\tau_2}_{\tau_1}\cap \{r>\rho\}$.

Since our primary focus is on the initial value problem \eqref{eq:intro:IVP-wave}, suppose that the initial data are considered on $\Sigma_{\tau_0}$ with $\tau_0=\frac{1}{2}\frac{t_0^{1-q}}{1-q}= u_0 +\frac{\rho}{2}$. Therefore, we would only consider the set $\M\subset M$ of points with $\tau\geq \tau_0 $, that is, the \emph{future domain of dependence} of the initial hypersurface $\Sigma_{\tau_0}$.

Multiple times in the next section, we integrate the divergence identity \eqref{eq:TEMT:DivergenceIdentity} over $\R^{\tau_2}_{\tau_1}$ and apply the divergence theorem to obtain energy estimates. The following lemma provides the appropriate volume forms on each component of the boundary of $\R^{\tau_2}_{\tau_1}$. The proof is straightforward and follows from the general divergence theorem and the definition of $\R^{\tau_2}_{\tau_1}$.
\begin{lem}[Divergence theorem]\label{lem:general:divergenceThm}
    Suppose $\psi$ satisfies the wave equation. For any smooth twisting function $\beta$, smooth vector field multiplier $X$, and smooth modification function $w$, if we integrate of \eqref{eq:TEMT:DivergenceIdentity} over $\R^{\tau_2}_{\tau_1}$ with respect to $d\Vol_{\rm{FLRW}}$, we have
    \begin{align*}
        \int_{\R^{\tau_2}_{\tau_1}} \tilde K^{X, w} [\psi] \,t^{3q} dtdx= &-\int_{\S_{\tau_2}} \tilde J^{X,w}_t [\psi] \,t^{3q}dx+\int_{\S_{\tau_1}} \tilde J^{X,w}_t [\psi] \,t^{3q}dx \\&-\int_{\C_{\tau_2}} \tilde J^{X,w}_v [\psi] \, t^{2q}r^2 dvd\omega+ \int_{\C_{\tau_1}} \tilde J^{X,w}_v [\psi] \, t^{2q}r^2 dvd\omega-\int_{\I^{\tau_2}_{\tau_1}} \tilde J^{X,w}_u [\psi] \, t^{2q}r^2 dud\omega\,.
    \end{align*}
\end{lem}

The next lemma provides a method for estimating weighted terms with factors depending on both $t$ and $r$, across different regions. To this end, for any $\epsilon>0$, we first cover $\M$ with three regions as
$$\M=\{\tau_0\leq\tau\}\subseteq \U^+ \cup \U^\epsilon_{u_0} \cup\U_\text{cmp} \,,$$
where
\begin{align*}
    \U^+:=& \M\cap\{u>0\}\,,\\
    \U_{u_0}^\epsilon:=& \M\cap\{ -|u_0|\leq u\leq 0\}\cap\{2|u_0|(1+\frac{1}{\epsilon})\leq r\}\,,\\
    \U_\text{cmp}:=& \M\cap\{\tau_0<\tau\}\cap \{u\leq 0\}\cap \{r\leq2|u_0|(1+\frac{1}{\epsilon})\}\,.
\end{align*}
Note that $\U_{u_0}^\epsilon$ and $\U_\text{cmp}$ might be empty depending on the choice of $\tau_0$. See Figure \ref{fig:three-regions}.
\begin{figure}
\centering\includegraphics[width=0.3\textwidth]{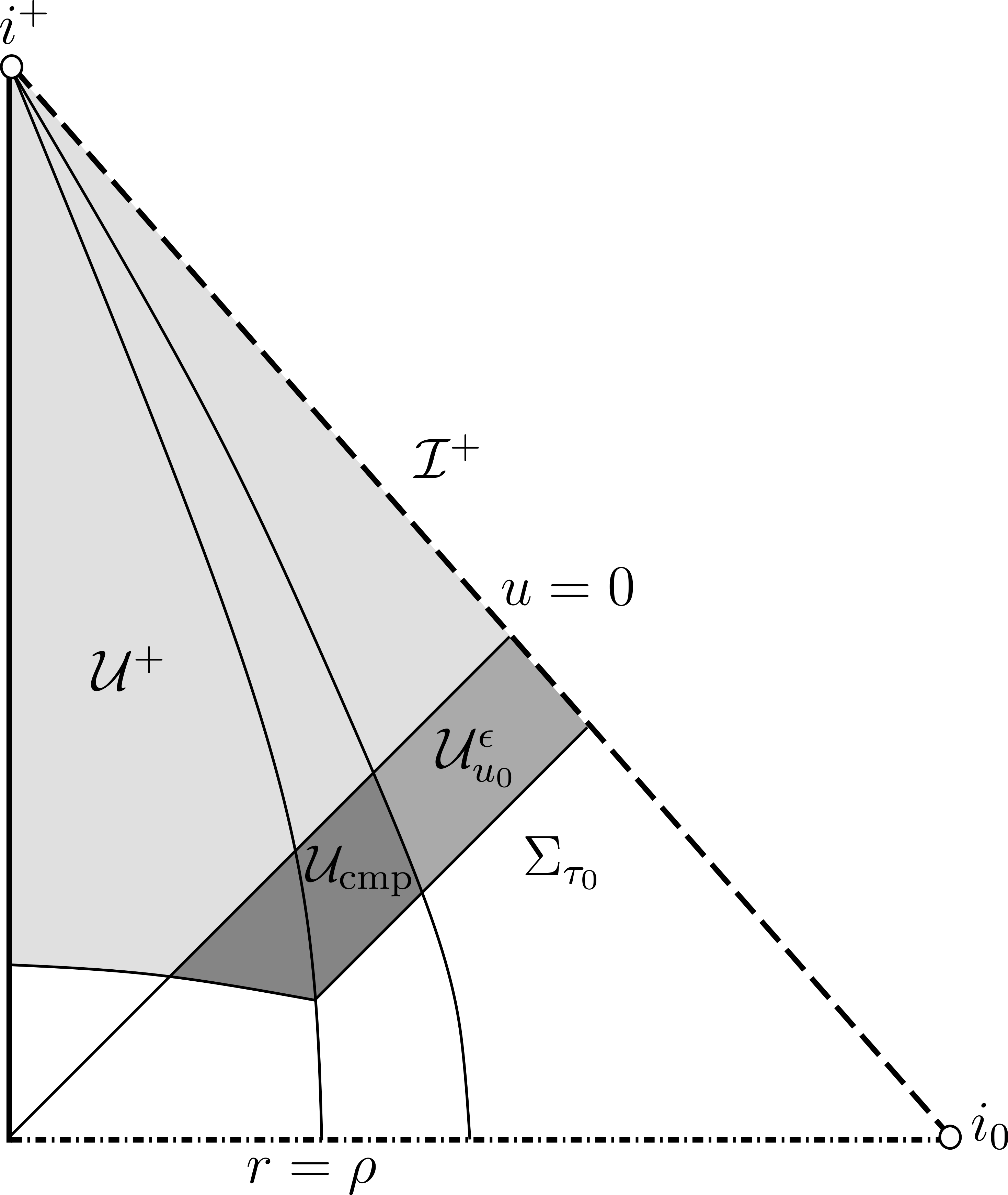}
        \caption{The future domain of dependence of the initial hypersurface $\Sigma_{\tau_0}$ is covered by the regions $\U^+$, $\U^\epsilon_{u_0}$, and the compact set $\U_\text{cmp}$.}\label{fig:three-regions}
\end{figure}
\begin{lem}\label{lem:general:t-r} For any initial data hypersurface $\Sigma_{\tau_0}$, there exists $\fC>0$ such that for any point in $\M$ we have
\begin{equation}\label{eq:general:t-r-lemma}
    rt^{q-1}\leq \fC\,.
\end{equation}
Moreover, for any $\epsilon>0$ the following holds.
\begin{enumerate}
    \item For all points in $\U^+$
\begin{equation}\label{eq:sub:uPositive}
        rt^{q-1}\leq \frac{1}{1-q}\,.
    \end{equation}
    \item On $\U_{u_0}^\epsilon$  
    \begin{equation*}
        rt^{q-1}\leq \frac{1+\epsilon}{1-q}.
    \end{equation*}
    \item The set $\U_\text{cmp}$ is compact and for any points in $\U_\text{cmp}$
    \begin{equation}\label{eq:sub:exchange-bounded}
        \frac{1}{1-q}\leq rt^{q-1}\leq 2|u_0|(1+\frac{1}{\epsilon}) t^{q-1}_0\,.
    \end{equation}
\end{enumerate}
\end{lem}
\begin{proof}
    The lemma follows easily from the definitions.
\end{proof}
\subsubsection{Angular momentum derivatives}\label{sec:angular-derivatives}
We introduce the \emph{angular momentum} Killing vector fields
\begin{equation*}\label{eq:general:angular-operators}
    \Omega_1=-x^3\d_{x^2}+x^2\d_{x^3}\,,\qquad\Omega_2=-x^1\d_{x^3}+x^3\d_{x^1}\,,\qquad\Omega_3=-x^2\d_{x^1}+x^1\d_{x^2}\,.
\end{equation*}
Observe that $\Omega_i$ with $i=1,2,3$ commutes with $\d_r$, $\d_u$, and $\d_v$, and moreover\footnote{Throughout this paper, we use the notation $A\lesssim B$ to indicate that there exists a universal constant  $C>0$ such that $A\leq C B$.}

\begin{align}\label{eq:general:angularderivatives-sphere}
    &|\rslnabla \psi|^2 = \sum_{i=1,2,3} |\Omega_i \psi|^2\,,\\
    &|\rslnabla{}^2 \psi|^2 \lesssim \sum_{i,j=1,2,3} |\Omega_i \Omega_j \psi|^2\,,
\end{align}
where $\rslnabla$ is the covariant derivative on the unit sphere $\SS^2$ and $|\cdot|$ denotes the standard norm on $\SS^2$.

Note that in FLRW spacetimes, for the sphere $S_{r,t}=\{t\}\times \{r\}\times \SS^2$ in the standard spherical coordinates $(t,r,\theta,\phi)$, we have 
$$|\slnabla_{S_{r,t}}\psi|_{\slashed{g}_{S_{r,t}}}= \frac{1}{r t^q}|\rslnabla\psi|\,,$$
where $\slashed{g}_{S_{r,t}}$ is the induced metric on $S_{r,t}$.
\subsection{Wave equation on FLRW spacetimes}\label{sec:wave-FLRW}
In this paper, the wave equation on FLRW is mainly represented in a twisted form. If we let $\beta=a(t)^{-1}=t^{-q}$ in \eqref{eq:general:twisted-wave}, we then are left with
\begin{equation*}
    \Box_g \psi= t^q g^{\mu\nu} \nabla_\mu(\frac{1}{t^{2q}}\nabla_\nu(t^q\psi)) - \V \psi=0\,,
\end{equation*}
in which $\V=-\frac{1}{\beta}\Box_g \beta=-\frac{q(2q-1)}{t^2}$. The above equation can also be rewritten as
\begin{equation}\label{eq:Setup:qWaveTwisted}
    \Box_g \psi= -\frac{1}{t^{2q}} \d_t(t^q \d_t(t^q \psi)) + \frac{1}{t^{2q}} \Delta_{\RR^3} \psi - \V \psi\,,
\end{equation}
 which exhibits the twisted structure of the equation. We emphasise the change in the sign of $\V$,
 \begin{align*}
     \V &=0 \qquad &\text{ if }\, q=\frac{1}{2}\,, \\
     \V &>0 \qquad &\text{ if }\, q<\frac{1}{2}\,, \\
     \V &<0 \qquad &\text{ if }\, q>\frac{1}{2}\,.
 \end{align*}
 Furthermore, the wave equation can be represented in double null coordinates as
\begin{equation}\label{eq:general:wave-double-null}
    \Box_g \psi= -\frac{1}{t^{3q}r}\d_{uv}(t^q r\psi)+\frac{1}{t^{2q}r^2}\rslDelta \psi - \V\psi\,,
\end{equation}
with $\rslDelta$ denoting the Laplacian induced on $\SS^2$.

The main twisting function considered in this article is $\beta=\frac{1}{t^q}$. This function serves as the default twisting function throughout the next sections unless explicitly stated otherwise.

\begin{rmk}[Commutation with twisted time derivative]\label{rmk:radiation:commuting}
    It is worth remarking that equation \eqref{eq:Setup:qWaveTwisted} indicates that in the radiation case, we have
    \begin{equation*}
        \Box_g ( \d_t(\sqrt t \psi))= \d_t(\sqrt t \Box_g \psi) + \frac{1}{\sqrt t} \Box_g \psi.
    \end{equation*}
    In other words, if $\psi$ solves the wave equation, then $\d_t(\sqrt t \psi)$ is also a solution --- a property that holds specifically for $q=\frac{1}{2}$. In general, we instead have
    \begin{equation*}
        \Box_g (\d_t(t^q\psi))=\d_t(t^q \Box_g \psi) + 2q t^{q-1}\Box_g \psi + 2(q-1) \V t^{q-1} \psi\,.
    \end{equation*}
\end{rmk}
\subsection{Functional inequalities}\label{sec:functional_ineq}
In this section, we present the general functional inequalities that are used throughout the paper.
\begin{lem}[Abstract hierarchical estimates]\label{lem:general:hierarchy-estimate}
Let $0\leq\sigma<1$. Assume $f$, $h_1$, and $h_{2-\sigma}$ are positive continuous functions, satisfying 
    \begin{align}
    f(\tau_2)&\lesssim f(\tau_1)\,,\label{eq:general:0-EE}\tag{0-EE}\\
    \int^{\tau_2}_{\tau_1} f(\tau) \, d\tau + h_1(\tau_2)&\lesssim h_1(\tau_1) + f(\tau_1)\,,\label{eq:general:1-EE}\tag{1-EE}\\
    \left(\int^{\tau_2}_{\tau_1} h_1^{\frac{1}{1-\sigma}}(\tau) \, d\tau\right)^{1-\sigma} + h_{2-\sigma}(\tau_2)&\lesssim h_{2-\sigma}(\tau_1) + f(\tau_1)\label{eq:general:2-EE}\tag{2-EE}\,,
\end{align}
for any $\tau_2>\tau_1>\tau_0$. Then, for any $\tau>\tau_0$ we have
\begin{equation*}
    f(\tau)\lesssim \frac{1}{\tau^{2-\sigma}}\fC\,,
\end{equation*}
and 
\begin{equation*}
    h_1(\tau)\lesssim \frac{1}{\tau^{1-\sigma}}\fC\,,
\end{equation*}
where $\fC=h_1(\tau_0) + f(\tau_0)+\left(  h_{2-\sigma}(\tau_0) + f(\tau_0)\right)^{\frac{1}{1-\sigma}
   } $.
\end{lem}
\begin{proof}
    First note that for $\tau\geq 2\tau_0$, \eqref{eq:general:0-EE} and \eqref{eq:general:1-EE} together imply
\begin{equation*}
    f(\tau)\lesssim\frac{1}{\tau - \tau_0} \int^{\tau}_{\tau_0} f(\tau') \, d\tau' \lesssim \frac{1}{\tau}   \underbrace{\left( h_1(\tau_0) + f(\tau_0) \right)}_{:=\fC_1}.
\end{equation*}
Also, for bounded $\tau\leq 2\tau_0$, \eqref{eq:general:0-EE} gives $f(\tau)\lesssim \frac{1}{\tau}\fC_1$.
To improve this decay rate, consider a dyadic sequence $\{\tau_n\}_n$, that is, $\tau_{n+1}\sim \tau_n\sim \tau_{n+1}-\tau_n\sim 2^n$. From the mean value theorem and \eqref{eq:general:2-EE}, we deduce that there exists $\bar \tau \in (\tau_n, \tau_{n+1})$ such that
\begin{equation*}
   \bar \tau_n h_1^{\frac{1}{1-\sigma}}(\bar \tau_n)\lesssim  \int^{\tau_{n+1}}_{\tau_n} h_1^{\frac{1}{1-\sigma}}(\tau') \, d\tau'\lesssim \left(  h_{2-\sigma}(\tau_n) + f(\tau_n)\right)^{\frac{1}{1-\sigma}
   } \lesssim  (h_{2-\sigma}(\tau_0) + f(\tau_0))^{\frac{1}{1-\sigma}}=:\fC_2^{\frac{1}{1-\sigma}},
\end{equation*}
which implies $h_1(\tau_n)\lesssim \frac{1}{\tau_n^{1-\sigma}} \fC_2$. Applying \eqref{eq:general:0-EE} and \eqref{eq:general:1-EE} again gives
\begin{equation*}
    (\tau_{n+1}-\tau_{n}) f(\tau_{n+1})\lesssim \int^{\tau_{n+1}}_{\tau_n} f(\tau') \, d\tau' \lesssim h_1(\tau_n) + f(\tau_n) \lesssim \frac{1}{\tau_n^{1-\sigma}} \fC_2 + \frac{1}{\tau_n} \fC_1\leq \frac{\fC}{\tau_n^{1-\sigma}}\,,
\end{equation*}
with $\fC=\fC_1+\fC_2$.
Finally, since for any $\tau$ there exists $n$ such that $\tau_{n+1}\leq \tau <\tau_{n+2}$, we have
$$f(\tau) \lesssim f(\tau_{n+1}) \lesssim \frac{\fC}{\tau^{2-\sigma}}\,,$$
and
$$h_1(\tau) \lesssim h_1(\tau_{n+1}) \lesssim \frac{\fC}{\tau^{2-\sigma}}\,,$$
as desired.
\end{proof}
\begin{lem}[Interpolation]\label{lem:interpolation-r^p}
    Let $\varphi:\M\to \RR$ be any differentiable function and $\Sigma_\tau$ and $\d_v$ as in \ref{sec:metric-wave-foliation}. Define
    $$h_p(\tau)=\int_{\Sigma_\tau} r^p |\d_v\varphi|^2\, d\Sigma_\tau\,.$$
    Then the interpolation inequality
    \begin{equation}
        h_1(\tau)\leq h_{1-\sigma}^{1-\sigma}(\tau) h_{2-\sigma}^{\sigma}(\tau) \,,
    \end{equation}
    holds for any $0<\sigma<1$.
\end{lem}
\begin{proof}
    The result follows from a standard Hölder inequality with exponents $\frac{1}{1-\sigma}$ and $\frac{1}{\sigma}$.
\end{proof}
\begin{prop}[Sobolev inequalities]\label{prop:general:Sobolev}
Suppose $f:\RR^3\to \RR$, $g:\SS^2\to\RR$, $h:\TT^3\to \RR$ are smooth functions, and $f$ has a compact support. Then, the following Sobolev inequalities hold
    \begin{align}
        &\|f\|^2_{L^\infty(\RR^3)} \lesssim \int_{\RR^3} \left(|\nabla f(x)|^2+|\nabla^2 f(x)|^2\right)\,dx\,,\label{eq:general:Sobolev:RR}\\
        &\|g\|^2_{L^\infty(\SS^2)}\lesssim \int_{\SS^2} \left(|g(\omega)|^2+|\rslnabla g(\omega)|^2+|\rslnabla{}^2 g(\omega)|^2\right)\,d\omega\,,\label{eq:general:Sobolev:SS}\\
        &\Big\Vert h(x)-\frac{1}{|\TT^3|}\int_{\TT^3}h(y)\,dy\Big\Vert^2_{L^\infty(\SS^2)}\lesssim \int_{\TT^3} |\nabla^2 h(x)|\,dx\,.\label{eq:general:Sobolev:TT}
    \end{align}
    Moreover, for any fixed $s>0$, the following local version of Sobolev inequality holds as well
    \begin{align}
        \|f\|^2_{L^\infty(B_s)} \lesssim s^{-3} \int_{B_{2s}} \left(|f(x)|^2+s^2|\nabla f(x)|^2+s^4|\nabla^2 f(x)|^2\right)\,dx\,.\label{eq:general:Sobolev:Local}
    \end{align}
\end{prop}
\begin{proof}
    To prove the first estimate, let $\hat f (\xi)$ denotes the Fourier transform of $f$. Observe that for $\alpha<\frac{3}{2}$ and $\beta>\frac{3}{2}$, we have 
    \begin{align*}
        |f(x)|&\lesssim \int_{\RR^3} |\hat f(\xi)| \,d\xi = \int_{B_1} |\hat f(\xi)| \,d\xi+\int_{\RR^3\setminus B_1} |\hat f(\xi)| \,d\xi\\
        &\lesssim \left(\int_{B_1} |\xi|^{-2\alpha}\,d\xi\int_{B_1} |\hat f(\xi)| |\xi|^{2\alpha} \,d\xi \right)^\frac{1}{2}+\left( \int_{\RR^3\setminus B_1} |\xi|^{-2\beta}\,d\xi\int_{\RR^3\setminus B_1} |\hat f(\xi)| |\xi|^{2\beta} \,d\xi)\right)^\frac{1}{2}\\
        &\lesssim \Vert f\Vert _{\mathring{H}^\alpha(\RR^3)} (\int_0^1 r^{2-2\alpha}\,dr)^\frac{1}{2}+ \Vert f\Vert _{\mathring{H}^\beta(\RR^3)} (\int_1^\infty r^{2-2\beta}\,dr)^\frac{1}{2}\lesssim \Vert f\Vert _{\mathring{H}^\alpha(\RR^3)}+\Vert f\Vert _{\mathring{H}^\beta(\RR^3)}\,.
    \end{align*}
    So if $\alpha=1$ and $\beta=2$, one then is left with the required estimate.
    To prove the local version, simply consider $f_s(x)=f(sx)$ and use the Sobolev inequality for $f_s \phi$ where $\phi\in C^\infty_0(B_2)$ is a cut-off function with $f=1$ in the ball $B_1$. 
    
    For the proof of \eqref{eq:general:Sobolev:SS}, refer to \cite{Christodoulou07ShockFormation}*{Section~5.2}.
    
    Finally, to prove the estimate on $\TT^3$, recall the Fourier series of $h(x)$, and observe that for $\alpha>\frac{3}{2}$ we have
    \begin{align*}
        |h(x)-\frac{1}{|\TT^3|} \int_{\TT^3} f(y)\,dy|&\lesssim \sum_{k\in \ZZ^3, k\neq 0} |\hat f(k) e^{-ik\cdot x|}|\lesssim \sum_{k\in \ZZ^3, k\neq 0} |\hat f(k)| |k|^\alpha |k|^{-\alpha}
        \\&\lesssim \left(\sum_{k\in \ZZ^3, k\neq 0} |\hat f(k)|^2 |k|^{2\alpha}\right)^\frac{1}{2}\left(\sum_{k\in \ZZ^3, k\neq 0}|k|^{-2\alpha}\right)^\frac{1}{2}\lesssim \|h\|_{\mathring{H}^\alpha(\TT^3)}\,,
    \end{align*}
    which completes the proof if we let $\alpha=2$.
\end{proof}
\begin{lem}[Grönwall's inequality]\label{lem:Gronwall} Let $f:[a,b]\to\RR$ be a continuous function. Assume that there exist a constant $\fC$ and a non-negative function $g(t)$ such that for all $t\in [a,b]$
\begin{align*}
    f(t)\leq \fC + \int_a^t g(t)f(t)\,dt\,.
\end{align*}
Then, for all $t\in[a,b]$
\begin{align*}
    f(t)\leq \fC \exp\left({\int_a^t g(t)\,dt}\right)\,.
\end{align*}
\end{lem}
The proof of Lemma \ref{lem:Gronwall} is standard.
\begin{lem}\label{lem:general:angular-derivatives}
    Let $f:\RR^3\to \RR$ be a smooth function.
    \begin{align*}
        |\slDelta_{\partial B_r} f|\lesssim \frac{1}{(1+r)^2} \sum_{i=1,2,3}\sum_{j=1,2,3}\left( |\Omega_i\Omega_j f|+|\d_{x^i}\d_{x^j} f|\right)\,.
    \end{align*}
\end{lem}
\begin{proof}
    For large $r$ observe that 
    \begin{align*}
        |\slDelta_{\partial B_r} f |=\frac{1}{r^2} |\rslDelta f| \leq \frac{1}{r^2}|\rslnabla^2 f|\lesssim \frac{1}{(1+r)^2} \sum_{i=1,2,3}\sum_{j=1,2,3}  |\Omega_i\Omega_j f|\,.
    \end{align*}
    On the other hand, for small $r$, we have
    \begin{align*}
        |\slDelta_{\partial B_r} f|\leq |\Delta_{\RR^3} f| \leq \sum_{i=1,2,3}\sum_{j=1,2,3} |\d_{x^i}\d_{x^j} f| \lesssim \frac{1}{(1+r)^2} \sum_{i=1,2,3}\sum_{j=1,2,3} |\d_{x^i}\d_{x^j} f|\,.
    \end{align*}
\end{proof}
\section{Energy Estimates}\label{sec:EnergyEstimates}
This section is dedicated to providing the proof of Theorems \ref{thm:intro:boundedness}, \ref{thm:itro:Morawetz} and \ref{thm:intro:rp}.  The statements and proofs are divided into two cases, namely $0<q\leq \frac{1}{2}$ and $\frac{1}{2}<q<1$.
Propositions \ref{prop:sub:EB} and \ref{prop:sup:energy-bnd} complete the proof of Theorem \ref{thm:intro:boundedness}, Propositions \ref{prop:sub:ILED} and \ref{prop:sup:Morawetz} imply Theorem \ref{thm:itro:Morawetz}, and finally, Theorem \ref{thm:intro:rp} simply follows from Propositions \ref{prop:sub:rp} and \ref{prop:sup:rp}. 

Recall that case $q=\frac{1}{2}$ considered the simplest for deriving energy estimates; see Section \ref{subsubsec:radiation}. In this section, however, the radiation case is treated together with the case $0<q<\frac{1}{2}$.

The proofs of the following propositions rely on the twisted divergence identity \eqref{eq:TEMT:DivergenceIdentity} with suitable choices of the twisting function $\beta$, vector field multiplier $X$, modification function $w$, and the integration region $\R$.
\subsection{The case $0<q\leq\frac{1}{2}$}
Recall $\V=-\frac{q(2q-1)}{t^2}$ and define
\begin{align}\label{eq:sub:energy-sigma-tau}
    \begin{split}
        \E_q[\psi](\tau):= &\int_{\S_{\tau}}   t^{2q} \left( |\d_t(t^q \psi)|^2 +|\d_r\psi|^2+ \frac{1}{r^2}|\d_r(r \psi)|^2+\frac{1}{r^2}|\rslnabla\psi|^2 +t^{2q}\V\psi^2 +\frac{1}{1+r^2}\psi^2\right) \,dx \\
        +&\int_{\C_{\tau}} \left(|\d_v(t^q \psi)|^2+ \frac{1}{r^2}|\d_v(t^qr \psi)|^2 + \frac{t^{2q}}{r^2}|\rslnabla \psi|^2 + t^{4q} \V \psi^2 + \frac{t^{2q}}{r^2}\psi^2\right)  r^2dvd\omega\,.
    \end{split}
    \end{align}
    Notice that \eqref{eq:sub:energy-sigma-tau} is equivalent to energy flux \eqref{eq:intro:energy-sigma-tau} when $0<q< \frac{1}{2}$, as $\mu_q^\varepsilon=0$ in this range. When $q=\frac{1}{2}$, the terms involving the potential $\V$ vanish, making \eqref{eq:sub:energy-sigma-tau} slightly different from \eqref{eq:intro:energy-sigma-tau}. However, in view of Lemma \ref{lem:general:t-r}, we are able to control $t^{2q-2}\psi^2$ with $\frac{1}{1+r^2}\psi^2$. Therefore, a uniform energy bound for \eqref{eq:sub:energy-sigma-tau} is sufficient to prove Theorem \ref{thm:intro:boundedness} for $0<q\leq \frac{1}{2}$. The following proposition is devoted to proving this energy bound in this range.
\begin{prop}[Energy boundedness for $0<q\leq\frac{1}{2}$]\label{prop:sub:EB}
    Let $\psi$ be a solution to the initial value problem \eqref{eq:intro:IVP-wave} with $0<q\leq\frac{1}{2}$. Then for any $\tau_2>\tau_1$ we have
    \begin{align}
          \E_q[\psi](\tau_2)&\lesssim\E_q[\psi](\tau_1)\,,\label{eq:sub:Eng-boundedness}\\
          \int_{\I^{\tau_2}_{\tau_1}} \left(|\d_u(t^q \psi)|^2+ \frac{1}{r^2}|\d_u(t^qr \psi)|^2 + \frac{t^{2q}}{r^2}|\rslnabla \psi|^2 + t^{4q} \V \psi^2 + \frac{t^{2q}}{r^2}\psi^2\right)  \, r^2dvd\omega&\lesssim\E_q[\psi](\tau_1)\label{eq:sub:nullInfty-bound}\,.
      \end{align}
\end{prop}
\begin{proof}
    To begin with, set $\beta_1=t^{-q}$ and $X_1=t^q\d_t$. Observe that
    \begin{align}
        \begin{split}
            &\tilde J^X_t [\psi]=\frac{1}{2t^{q}}\left(|\d_t(t^q\psi)|^2+|\nabla \psi|^2\right)+ \frac{1}{2}t^q\V\psi^2\, ,\\
        & \tilde \TT_{\mu\nu}[\psi]{}^{(X)}\pi^{\mu\nu}= \frac{q}{t^{1+q}} \left(|\d_t(t^q \psi)|^2-|\nabla \psi|^2 \right) - 2q t^{q-1}\V \psi^2\,,\\
        &X^\mu\tilde S_\mu[\psi] =\frac{q}{t^{1+q}} \left(-|\d_t(t^q \psi)|^2+|\nabla \psi|^2 \right) + (q+1) t^{q-1}\V \psi^2\,,\label{eq:sub:pointwise-X-S-boundedness}
        \end{split}
    \end{align}
    and
    \begin{equation}\label{eq:sub:pointwise-K_boundedness}
        \tilde K^X[\psi]= (1-q) t^{q-1}\V \psi^2\,.
    \end{equation}
    We also have
    \begin{align*}
        \tilde J^{\beta_1,X_1}_v [\psi]&= \frac{1}{2}(\tilde \TT_{vv}+\tilde \TT_{vu})= \frac{1}{2}\left( t^{-2q} |\d_v(t^q \psi)|^2 + \frac{1}{r^2}|\rslnabla \psi|^2 + t^{2q} \V \psi^2\right)\,,\\
        \tilde J^{\beta_1,X_1}_u [\psi]&= \frac{1}{2}(\tilde \TT_{uv}+\tilde\TT_{uu})= \frac{1}{2}\left( t^{-2q} |\d_u(t^q \psi)|^2 + \frac{1}{r^2}|\rslnabla \psi|^2 + t^{2q} \V \psi^2\right)\,.
    \end{align*}
    Thus, the twisted divergence identity \eqref{eq:TEMT:DivergenceIdentity}, after integration over the region $\R^{\tau_2}_{\tau_1}$ and using Lemma \ref{lem:general:divergenceThm}, yields
    \begin{align}\label{eq:sub:Conservation1-null}
    \begin{split}
        &\int_{\S_{\tau_2}}   t^{2q} \left( |\d_t(t^q \psi)|^2 + |\nabla\psi|^2 +t^{2q}\V\psi^2 \right) \,dx +\int_{\C_{\tau_2}} \left( |\d_v(t^q \psi)|^2 + \frac{t^{2q}}{r^2}|\rslnabla \psi|^2 + t^{4q} \V \psi^2 \right)  \, r^2dvd\omega \\
        &+ \int_{\I^{\tau_2}_{\tau_1}} \left( |\d_u(t^q \psi)|^2 + \frac{t^{2q}}{r^2}|\rslnabla \psi|^2 + t^{4q} \V \psi^2 \right)  \, r^2dud\omega +\int_{\R^{\tau_2}_{\tau_1}} t^{4q-1}\V\psi^2\, dxdt \\&\lesssim \int_{\S_{\tau_1}}   t^{2q} \left( |\d_t(t^q \psi)|^2 +  |\nabla\psi|^2 +t^{2q}\V\psi^2 \right) \,dx +\int_{\C_{\tau_1}} \left( |\d_v(t^q \psi)|^2 + \frac{t^{2q}}{r^2}|\rslnabla \psi|^2 + t^{4q} \V \psi^2 \right)  \, r^2dvd\omega.
    \end{split}
    \end{align}
    In order to obtain an energy estimate for the remaining terms of $\E_q[\psi](\tau)$, consider a new twisting function $\beta_2=\frac{1}{t^q r}$. We mark that the potential $\V_{\beta_2}$ associated to this new twisting function is still $\V=-\frac{q(2q-1)}{t^2}$. The same multiplier $X_2=t^q\d_t$ leads to the same bulk term, \textit{i.\@e.\@}\ $\tilde K^{\beta_2,X_2}[\psi]= (1-q) t^{q-1}\V \psi^2$ which gives rise to
    \begin{align}\label{eq:sub:rtqTwisted}
    \begin{split}
        \int_{\R^{\tau_2}_{\tau_1}} t^{4q-1}\V\psi^2\, dxdt &+ \int_{\S_{\tau_2}}   t^{2q} \left( |\d_t(t^q \psi)|^2 + \frac{1}{r^2}|\d_r(r \psi)|^2+\frac{1}{r^2}|\rslnabla\psi|^2 +t^{2q}\V\psi^2 \right) \,dx \\
        &+\int_{\C_{\tau_2}} \left( \frac{1}{r^2}|\d_v(t^qr \psi)|^2 + \frac{t^{2q}}{r^2}|\rslnabla \psi|^2 + t^{4q} \V \psi^2 \right)  \, r^2dvd\omega \\
        & + \int_{\I^{\tau_2}_{\tau_1}} \left( \frac{1}{r^2}|\d_u(rt^q \psi)|^2 + \frac{t^{2q}}{r^2}|\rslnabla \psi|^2 + t^{4q} \V \psi^2 \right)  \, r^2dud\omega \\&\lesssim \int_{\S_{\tau_1}}   t^{2q} \left( |\d_t(t^q \psi)|^2 +   \frac{1}{r^2}|\d_r(r \psi)|^2+\frac{1}{r^2}|\rslnabla\psi|^2 +t^{2q}\V\psi^2 \right) \,dx \\
        &+\int_{\C_{\tau_1}} \left(\frac{1}{r^2} |\d_v(t^q r\psi)|^2 + \frac{t^{2q}}{r^2}|\rslnabla \psi|^2 + t^{4q} \V \psi^2 \right)  \, r^2dvd\omega\,.
    \end{split}
    \end{align}
    To estimate the only remaining term, note that
    \begin{equation}\label{eq:sub:engbound-untwst}
        {t^{2q}}\psi^2\lesssim r^2|\d_m(t^q \psi)|^2 +|\d_m(t^q r\psi)|^2\,,
    \end{equation}
    where $m$ is either $v,u,$ or $r$. Combining the above with \eqref{eq:sub:Conservation1-null} and \eqref{eq:sub:rtqTwisted}, we arrive at \eqref{eq:sub:Eng-boundedness} and \eqref{eq:sub:nullInfty-bound}.
\end{proof}
The next proposition proves the ILED estimate \eqref{eq:intro:ILED} for $q\leq\frac{1}{2}$.
\begin{prop}[Integrated local energy decay for $0<q\leq\frac{1}{2}$]\label{prop:sub:ILED}
If $\psi$ solves the initial value problem \eqref{eq:intro:IVP-wave} with $0<q\leq\frac{1}{2}$, then for any $\tau_2>\tau_1$
      \begin{align}
          \int_{\R^{\tau_2}_{\tau_1}} \left( \frac{t^{3q}|\d_t\psi|^2}{(1+r)^{1+\delta}}+\frac{t^{q}|\d_t(t^q\psi)|^2}{1+r^{3+\delta}} +\frac{t^{q}|\d_r \psi|^2}{(1+r)^{1+\delta}} + \frac{t^q}{r^3}|\rslnabla\psi|^2+\frac{t^{q}\psi^2}{1+r^{3+\delta}} \right) \, dxdt&\lesssim \E_q[\psi](\tau_1)\,.\label{eq:sub:Morawetz}
      \end{align}
      Moreover, when $q=\frac{1}{2}$, we have
      \begin{align}
          \int_{\R^{\tau_2}_{\tau_1}} \left(  \frac{\sqrt t}{1+r^{1+\delta}} \left( |\d_t (\sqrt{t} \psi)|^2 + |\d_r\psi|^2 \right) +\frac{\sqrt t}{r^3}|\rslnabla \psi|^2 + \frac{\sqrt t}{1+ r^{3+\delta}} \psi ^2 \right)\; dxdt&\lesssim \E_{\frac{1}{2}}[\psi](\tau_1).\label{eq:radiation:Morawetz}
      \end{align}
\end{prop}
\begin{proof}
    In order to show \eqref{eq:sub:Morawetz}, set $\beta_1=\frac{1}{t^q}$, $Y_1=\d_r$ and $z_1=\frac{1}{r}$.  We then are left with
    \begin{align*}\label{eq:sub:KY1z1}
        \begin{split}
            & Y_1^{\beta_1,\nu}\tilde S_\nu =0\,,\qquad \tilde K^{\beta_1,Y_1,z_1} [\psi] = \frac{1}{t^{2q} r^3}|\rslnabla \psi |^2\,,\\
            &\tilde J^{\beta_1,Y_1,z_1}_t=\frac{1}{r t^q} \psi \d_t(t^q \psi) + \frac{1}{t^q} \d_t(t^q \psi) \d_r \psi\,,\qquad \tilde J ^{\beta_1,Y_1,z_1}_r=\frac{1}{2} \left( |\d_t (t^q \psi)|^2 + |\frac{1}{r} \d_r(r\psi)|^2 - \frac{1}{r^2}|\rslnabla \psi|^2 -t^{2q}\V\psi^2\right)\,,\\
            &\tilde J^{\beta_1,Y_1,z_1}_v=\frac{1}{2t^{2q}}|\d_v(t^q \psi)|^2 - \frac{1}{2r^2}|\rslnabla\psi|^2+\frac{1}{2}t^{2q}\V\psi^2+ \frac{1}{rt^q} \psi\d_v(t^q \psi)+ \frac{1}{2r^2}\psi^2\,,\\
            &\tilde J^{\beta_1,Y_1,z_1}_u=-\frac{1}{2t^{2q}}|\d_u(t^q \psi)|^2 + \frac{1}{2r^2}|\rslnabla\psi|^2-\frac{1}{2}t^{2q}\V\psi^2+\frac{1}{r t^q} \psi\d_u(t^q \psi)- \frac{1}{2r^2}\psi^2\,.
        \end{split}
    \end{align*}
    The divergence identity \eqref{eq:TEMT:DivergenceIdentity} applied on $\R^{\tau_2}_{\tau_1}\setminus \{r<\epsilon\}$ for small $\epsilon\ll\rho$, leads to
    \begin{align*}
         &\int_{\R^{\tau_2}_{\tau_1}} \frac{t^q}{r^3}|\rslnabla \psi |^2 \, dxdt+\int_{\S_{\tau_2}}  \tilde J^{\beta_1,Y_1,z_1}_t  \,t^{3q} dx + \int_{\C_{\tau_2}}  \tilde J^{\beta_1,Y_1,z_1}_v \, t^{2q}r^2 dvd\omega\\&+ \int_{\I^{\tau_2}_{\tau_1}}  \tilde J^{\beta_1,Y_1,z_1}_u \, t^{2q}r^2 dud\omega+ \epsilon^2\int_ {\d \{r<\epsilon\}} \tilde{J}^{\beta_1,Y_1,z_1}_r t^{3q}  \, dt d\omega\\&=  \int_{\S_{\tau_1}}  \tilde J^{\beta_1,Y_1,z_1}_t  \,t^{3q} dx + \int_{\C_{\tau_1}}  \tilde J^{\beta_1,Y_1,z_1}_v \, t^{2q}r^2 dvd\omega \,.
    \end{align*}
    Now note that the only terms in $\tilde{J}^{\beta_1,Y_1,z_1}_r$ with bad sign are bounded on the compact region $\d \{r<\epsilon\}$, and they vanish as $\epsilon\to 0$. Therefore, after passing to the limit as $\epsilon\to 0$, we arrive at
    \begin{align*}
        \int_{\R^{\tau_2}_{\tau_1}} \frac{t^q}{ r^3}|\rslnabla \psi |^2\,dxdt \lesssim \E_q[\psi](\tau_1)\,,
    \end{align*}
    thanks to \eqref{eq:sub:Eng-boundedness} and \eqref{eq:sub:nullInfty-bound} for controlling all the boundary terms. We now use the above estimate involving the angular derivative to control the other derivatives by letting $\beta_1=t^{-q}$, $Y_2=\frac{1}{(1+r)^\delta}\d_r$, and $z_2=\frac{1}{r(1+r)^\delta}$ for any positive $\delta$. In this setting, we have
    \begin{align*}
         \tilde K^{\beta_1,Y_2, z_2} [\psi]=& \frac{1}{t^{2q}r^3(1+r)^\delta} |\rslnabla \psi|^2 +\frac{-\delta}{2t^{2q} (1+r)^{1+\delta}}\left( |\d_t(t^q \psi)|^2 + |\d_r \psi|^2 - \frac{1}{r^2} |\rslnabla \psi|^2 \right) \\&- \frac{1}{2} \psi^2 \Box_g z_2+\frac{\delta}{2(1+r)^{1+\delta}}\V\psi^2\,.
    \end{align*}
    Note that in the radiation case, the last term above vanishes as $\V=0$, and all the terms involving $\rslnabla \psi|^2$ in $\tilde K ^{\beta_1,Y_2,z_2}[\psi]$ have a sign opposite to other terms in $\tilde K ^{\beta_1,Y_2,z_2}[\psi]$. So, after integration and using the divergence identity, we arrive at \eqref{eq:radiation:Morawetz}. However, when $q\neq\frac{1}{2}$, the final term which contains $\V$ has an unfavourable sign. To resolve this issue, we untwist the time derivative. In fact, after taking the volume form into account, we have
    \begin{equation*}
        t^q |\d_t(t^q \psi)|^2-t^{3q}\V \psi^2 = t^{3q}|\d_t\psi|^2+ \d_t(qt^{3q-1}\psi^2)\,.
    \end{equation*}
    Therefore, from the divergence identity \eqref{eq:TEMT:DivergenceIdentity}, we infer
    \begin{align*}
        &\int_{\R^{\tau_2}_{\tau_1}} \left( \frac{t^{q}}{(1+r)^{1+\delta}}(t^{2q}|\d_t\psi|^2 + |\d_r \psi|^2)+t^{3q}\psi^2\Box_g z_2 \right)\, dxdt+ \int_{\S_{\tau_2}} \frac{t^{3q-1}\psi^2}{(1+r)^{1+\delta}}\,dx+\int_{\C_{\tau_2}} \frac{t^{3q-1}\psi^2}{(1+r)^{1+\delta}}\,r^2dvd\omega\\
        &\lesssim \int_{\R^{\tau_2}_{\tau_1}} \frac{t^q}{r^3}|\rslnabla\psi|^2\, dxdt+ \sum_{i=1,2} \left(\int_{\S_{\tau_i}}  |\tilde J^{\beta_1,Y_2,z_2}_t|  \,t^{3q} dx + \int_{\C_{\tau_i}}  |\tilde J^{\beta_1,Y_2,z_2}_v| \, t^{2q}r^2 dvd\omega \right)+ \int_{\I^{\tau_2}_{\tau_1}}  |\tilde J^{\beta_1,Y_2,z_2}_u| \, t^{2q}r^2 dud\omega\\&+\int_{\S_{\tau_1}} \frac{t^{3q-1}}{(1+r)^{1+\delta}}\psi^2\,dx+\int_{\C_{\tau_1}} \frac{t^{3q-1}}{(1+r)^{1+\delta}}\psi^2\,r^2dvd\omega\,.
    \end{align*}
    One more time, estimates \eqref{eq:sub:Eng-boundedness} and \eqref{eq:sub:nullInfty-bound} control the boundary terms in the second line in the above estimate. For the last two boundary term, recall Lemma \ref{lem:general:t-r}, and observe
    \begin{align*}
        \int_{\S_{\tau_1}} \frac{t^{3q-1}}{(1+r)^{1+\delta}}\psi^2\,dx+\int_{\C_{\tau_1}} \frac{t^{3q-1}}{(1+r)^{1+\delta}}\psi^2\,r^2dvd\omega \lesssim \int_{\S_{\tau_1}} \frac{t^{2q}}{1+r^{2}}\psi^2\,dx+\int_{\C_{\tau_1}} \frac{t^{2q}}{r^2}\psi^2\,r^2dvd\omega\lesssim \E_q[\psi](\tau_1)\,.
    \end{align*}
    Thus,
    \begin{align*}
        \int_{\R^{\tau_2}_{\tau_1}} \left( \frac{t^{q}}{(1+r)^{1+\delta}}(t^{2q}|\d_t\psi|^2 + |\d_r \psi|^2)+t^{3q}\psi^2\Box_g z_2 \right)\, dxdt\lesssim\E_q[\psi](\tau_1)\,.
    \end{align*}
     To estimate the zeroth-order term in \eqref{eq:sub:Morawetz}, note that 
    \begin{equation}\label{eq:sub:Boxz2}
        \frac{1}{t^{2q}(1+ r^{3+\delta})} \psi ^2 \lesssim \psi^2 \Box_g z_2.
    \end{equation}
    Finally, in order to bound the twisted time derivative in \eqref{eq:sub:Morawetz}, observe
    \begin{align*}
        \int_{\R^{\tau_2}_{\tau_1}} \frac{t^q}{1+r^{3+\delta}} |\d_t(t^q \psi)|^2\,dxdt&\lesssim \int_{\R^{\tau_2}_{\tau_1}}\left( \frac{t^q}{1+r^{3+\delta}} \psi^2 + \frac{t^{3q}}{(1+r)^{1+\delta}} |\d_t\psi|^2\right)\,dxdt\lesssim\E_q[\psi](\tau_1) \,.
    \end{align*}
\end{proof}
Before stating the next proposition, we recall that $\D^{\tau_2}_{\tau_2}=\R^{\tau_2}_{\tau_1}\cap \{r>\rho\}$ for any $\tau_2>\tau_1\geq \tau_0$ and that $\sigma_q=2\frac{(q|1-2q|)^\frac{1}{2}}{|1-q|^2}$, as defined in \eqref{eq:intro:sigma_q}.
\begin{prop}[$r^p$-estimates for waves on FLRW with $0<q\leq\frac{1}{2}$]\label{prop:sub:rp}
Suppose $\psi$ is a solution to the initial value problem \eqref{eq:intro:IVP-wave} for $0<q\leq\frac{1}{2}$. Then, for any $\tau_2>\tau_1\geq \tau_0$ and  $0<p< 2-\sigma_q$, we have
    \begin{align}\label{eq:sub:rpEstimate}
    \begin{split}
        &\int_{\D^{\tau_2}_{\tau_1}} r^{p-1} \left( |\d_v \varphi|^2 +r^2|\d_v(t^q \psi)|^2+ \frac{1}{r^2}|\rslnabla\varphi|^2+ \frac{1}{r^2}\varphi^2\right) \, dudvd\omega + \int_{\C_{\tau_2}} r^p|\d_v\varphi|^2 \, dvd\omega \\&+  \int_{\I^{\tau_2}_{\tau_1}} r^{p-2}\left(|\rslnabla\varphi|^2 +\varphi^2\right)\, dud\omega
        \lesssim \int_{\C_{\tau_1}} r^p|\d_v\varphi|^2 \, dvd\omega +\E_{q}[\psi](\tau_1)\,,
    \end{split}
\end{align}
in which $\varphi=rt^q\psi$. Moreover, in the case $q=\frac{1}{2}$, in addition to the estimate \eqref{eq:sub:rpEstimate}, we also have
\begin{align}\label{eq:radiation:r2Estimate}
    \begin{split}
        \int_{\D^{\tau_2}_{\tau_1}} r  |\d_v \varphi|^2 \, dudvd\omega + \int_{\C_{\tau_2}} r^2|\d_v\varphi|^2 \, dvd\omega +  \int_{\I^{\tau_2}_{\tau_1}} |\rslnabla\varphi|^2 \, dud\omega
        \lesssim \int_{\C_{\tau_1}} r^2|\d_v\varphi|^2 \, dvd\omega +\E_{\frac{1}{2}}[\psi](\tau_1)\,,
    \end{split}
\end{align}
which corresponds to $p=2$.
\end{prop}
\begin{proof}
    To begin with, recall from \eqref{eq:general:wave-double-null} that $\varphi=t^qr\psi$ satisfies 
    \begin{equation*}
        -\d_{uv} \varphi + \frac{1}{r^2} \rslDelta\varphi - \W \varphi=0\,,
    \end{equation*}
    in which $\W=t^{2q}\V=q(1-2q)t^{2q-2}>0$. Multiply the above by $-2r^p \d_v \varphi$ to get
    \begin{align}
        &\d_u(r^p |\d_v \varphi|^2)  - \rsldiv( 2r^{p-2}\d_v\varphi \rslnabla\varphi) + \d_v\left(r^{p-2}|\rslnabla \varphi|^2\right)\notag \\
        &+ p r^{p-1}|\d_v \varphi|^2- (p-2)r^{p-3}|\rslnabla\varphi|^2+2r^p \W \varphi\d_v \varphi=0\,.\label{eq:sub:rpmethod-pointwise1}
    \end{align}    
    In addition, the following identity holds for any smooth function $\varphi$,
    \begin{equation}\label{eq:sub:rp-zeroth-order-term}
        \d_v(r^{p-2} \varphi^2)+(2-p)r^{p-3}\varphi^2 - 2r^{p-2}\varphi\d_v \varphi=0\,.
    \end{equation}
    Now we add the above equation times $\gamma>0$ to \eqref{eq:sub:rpmethod-pointwise1} to get
    \begin{align}\label{eq:sub:r^p-pointwise}
    \begin{split}
        &\d_u(r^p |\d_v \varphi|^2)  - \rsldiv( 2r^{p-2}\d_v\varphi \rslnabla\varphi) + \d_v(r^{p-2}|\rslnabla \varphi|^2+\gamma r^{p-2} \varphi^2) \\
        &+ p r^{p-1}|\d_v \varphi|^2+(2-p)r^{p-3}\left(|\rslnabla\varphi|^2+\gamma\varphi^2\right)=2r^{p-2}\left(\gamma- r^2\W  \right) \varphi\d_v \varphi\,.
    \end{split}
    \end{align}
    In the identity above, all terms on the left-hand side are either boundary terms or have a favourable  sign. The only error term appears on the right-hand side and should be controlled by the other terms.

    We first assume that $\tau_2>\tau_1\geq\frac{\rho}{2}$, \textit{i.\@e.\@}\  $u\geq 0$.
    Before integrating over $\D^{\tau_2}_{\tau_1}$, recall that Lemma \ref{lem:general:t-r} implies that for any point in $\U^+$, we have 
    $$r^2t^{2q-2}\leq \frac{1}{(1-q)^2}\,.$$
    In this region, if we let
     $\gamma=\frac{q|1-2q|}{|1-q|^2}$, we have
    \begin{align}\label{eq:sub:error-pointwise}
         2r^{p-2}\left(\gamma+ r^2|\W|  \right) |\varphi||\d_v \varphi|\leq 2 r^{p-2}\left(\gamma+ \frac{q|1-2q|}{|1-q|^2}  \right)|\varphi||\d_v \varphi|\leq  2\gamma^\frac{3}{2} r^{p-3}|\varphi|^2+2 \gamma^ \frac{1}{2}r^{p-1}|\d_v \varphi|^2\,.
    \end{align}
     So, as long as $2\sqrt \gamma< p<2 -2\sqrt \gamma$, the error term in \eqref{eq:sub:r^p-pointwise} can be controlled by the left-hand side of \eqref{eq:sub:r^p-pointwise} in $\D^{\tau_2}_{\tau_1}\cap \U^+$. Thus, 
    \begin{align}
    \begin{split}
        &\int_{\D^{\tau_2}_{\tau_1}} r^{p-1} \left( |\d_v \varphi|^2 + \frac{1}{r^2}|\rslnabla\varphi|^2+ \frac{1}{r^2}\varphi^2\right) \, dudvd\omega + \int_{\C_{\tau_2}} r^p|\d_v\varphi|^2 \, dvd\omega + \int_{\I^{\tau_2}_{\tau_1}} r^{p-2}\left(|\rslnabla\varphi|^2 +\varphi^2\right)\, dud\omega\\
        &\lesssim \int_{\C_{\tau_1}} r^p|\d_v\varphi|^2 \, dvd\omega + \underbrace{\int_{\{r=\rho, \tau_1\leq \tau\leq \tau_2\}} r^{p-2}\left(|\rslnabla\varphi|^2 +\varphi^2\right) \, dud\omega- \int_{\{r=\rho, \tau_1\leq \tau\leq \tau_2\}} r^p|\d_v\varphi|^2 \, dvd\omega}_{I}\,,
    \end{split}\label{eq:sub:rp-one-errors}
\end{align}
for any $\tau_2>\tau_1\geq\frac{\rho}{2}$.

We can now use the ILED estimate \eqref{eq:sub:Morawetz} to estimate the boundary term $I$ on $\{r=\rho\}$ as the power of $r$ in the weight factors does not matter when $r=\rho$. To do so, first note that we can modify the value of $\rho$ in order to use the mean value theorem on the region $\{r\sim \rho\}=\{\rho - \epsilon \leq r \leq \rho +\epsilon\}$ for a small $\epsilon>0$. In fact, there exists $\rho'\in\{r\sim \rho\}$ such that
\begin{align}
\begin{split}
    |I|&\lesssim\ \int_{\{r=\rho'\}} \left( \frac{t^q}{r^3} |\rslnabla\psi|^2 + \frac{1}{t^q}|\d_v (t^q \psi)|^2 + t^q \psi^2 \,\right) r^2 dtd\omega \\
   &\lesssim  \int_{\{r\sim\rho\}} \left(\frac{t^q}{r^3}|\rslnabla \psi|^2 + \frac{t^q}{1+r^{1+\delta}}  |\d_r\psi|^2  + \frac{t^q}{1+ r^{3+\delta}} \left(|\d_t (t^q \psi)|^2 + \psi ^2\right)\right) \; r^2dtdrd\omega \\&\stackrel{\eqref{eq:sub:Morawetz}}{\lesssim}  \E_{q}[\psi](\tau_1)\,.\label{eq:sub:MeanValueThm}
\end{split}
\end{align}
For simplicity of notation, we continue to use $\rho$ in place of $\rho'$.

To control term $r^{p+1}\d_v(t^q\psi)$ in \eqref{eq:sub:rpEstimate}, note that
\begin{align*}
    r^{p+1}|\d_v(t^q\psi)|^2=r^{p-1}|\d_v \varphi|^2+ (p-1) r^{p-3}\varphi^2 -\d_v(r^{p-2}\varphi^2)\,,
\end{align*}
leading to
\begin{align}\label{eq:sub:rp-untwst-v}
    \int_{\D^{\tau_2}_{\tau_1}} r^{p+1}|\d_v(t^q\psi)|^2\,dudvd\omega &\lesssim \int_{\D^{\tau_2}_{\tau_1}} r^{p-1}|\d_v \varphi|^2+ r^{p-3}\varphi^2 \,dudvd\omega + I\\
    &\lesssim \int_{\C_{\tau_1}} r^p|\d_v\varphi|^2 \, dvd\omega + \E_{q}[\psi](\tau_1)\,,
\end{align}
which completes the proof for $2\sqrt \gamma< p<2 -2\sqrt \gamma$ and $\tau_2>\tau_1\geq\frac{\rho}{2}$. Note that $\sqrt\gamma =1 $ when $q=\frac{1}{3}$.

Now we assume that $\tau_0\leq \tau_1<\tau_2\leq\frac{\rho}{2}$. Recall again that Lemma \ref{lem:general:t-r} implies that there exists $\fC$ such that $rt^{q-1}\leq \fC$. Similar to \eqref{eq:sub:error-pointwise}, we find that for any small $\epsilon$, there exist $\fC^*$ sufficiently large such that
    \begin{align*}
         2r^{p-2}\left(\gamma+ r^2|\W|  \right) |\varphi||\d_v \varphi|\leq 2 r^{p-2}\left(\gamma+ \fC^2q|1-2q|  \right)|\varphi||\d_v \varphi|\leq \epsilon r^{p-3}|\varphi|^2+ \fC^*r^{p-1}|\d_v \varphi|^2\,.
    \end{align*}
Now if we assume $\epsilon<(2-p)\gamma$, the term $r^{p-3}\varphi^2$ can be controlled by the term on the left-hand side of \eqref{eq:sub:r^p-pointwise}. Integrating over $\D^{\tau_2}_{\tau_1}\setminus \U^+$ yields
\begin{align}
        &\int_{\D^{\tau_2}_{\tau_1}} r^{p-1} \left( |\d_v \varphi|^2 + \frac{1}{r^2}|\rslnabla\varphi|^2+ \frac{1}{r^2}\varphi^2\right) \, dudvd\omega + \int_{\C_{\tau_2}} r^p|\d_v\varphi|^2 \, dvd\omega + \int_{\I^{\tau_2}_{\tau_1}} r^{p-2}\left(|\rslnabla\varphi|^2 +\varphi^2\right)\, dud\omega\\
        &\lesssim \int_{\C_{\tau_1}} r^p|\d_v\varphi|^2 \, dvd\omega + \int_{\{r=\rho\}} r^{p-2}\left(|\rslnabla\varphi|^2 +\varphi^2\right) \, dud\omega- \int_{\{r=\rho\}} r^p|\d_v\varphi|^2 \, dvd\omega+\fC^*\int_{\D^{\tau_2}_{\tau_1}} r^{p-1}|\d_v \varphi|^2\, dudvd\omega\,,\\
        &\stackrel{\eqref{eq:sub:MeanValueThm}}{\leq} \fC^*\left(\int_{\C_{\tau_1}} r^p|\d_v\varphi|^2 \, dvd\omega+\int_{\D^{\tau_2}_{\tau_1}} r^{p-1}|\d_v \varphi|^2\, dudvd\omega+\E_{q}[\psi](\tau_1)\right)\,,\label{eq:sub:rp-Gronwall-term}
\end{align}
for $\tau_0\leq \tau_1<\tau_2\leq\frac{\rho}{2}$ with constant $\fC^*$ that may vary between lines.
Thus, if we let $f(\tau)=\int_{\C_{\tau}} r^p|\d_v\varphi|^2 \, dvd\omega$, then Grönwall's inequality, Lemma \ref{lem:Gronwall}, implies
\begin{align*}
    f(\tau_2)&\leq \fC^*\left(\int_{\C_{\tau_1}} r^p|\d_v\varphi|^2 \, dvd\omega+\int^{\tau_2}_{\tau_1} f(\tau)\,d\tau+\E_{q}[\psi](\tau_1)\right)\leq \fC^*\left(\int_{\C_{\tau_1}} r^p|\d_v\varphi|^2 \, dvd\omega+\E_{q}[\psi](\tau_1)\right) e^{\fC^*(\tau_2-\tau_1)}\,.
\end{align*}
Note that the exponential term is bounded as $\tau_2\leq\frac{\rho}{2}$. Therefore, the spacetime term in \eqref{eq:sub:rp-Gronwall-term} is estimated as follows
\begin{align*}
    \int_{\D^{\tau_2}_{\tau_1}} r^{p-1}|\d_v &\varphi|^2\, dudvd\omega\leq \int_{\D^{\tau_2}_{\tau_1}} r^{p}|\d_v \varphi|^2\, dudvd\omega\leq \int_{\tau_1}^{\tau_2} f(\tau)\,d\tau\\&\leq \fC^*\left(\int_{\C_{\tau_1}} r^p|\d_v\varphi|^2 \, dvd\omega+\E_{q}[\psi](\tau_1)\right) \int_{\tau_1}^{\tau_2} e^{\fC^*(\tau-\tau_1)}\,d\tau\lesssim\int_{\C_{\tau_1}} r^p|\d_v\varphi|^2 \, dvd\omega+\E_{q}[\psi](\tau_1)\,.
\end{align*}
In order to bound the term involving $r^{p+1}\d_v(t^q\psi)$, note that \eqref{eq:sub:rp-untwst-v} holds in general. Thus, the above estimate completes the proof of \eqref{eq:sub:rpEstimate} for $2\sqrt \gamma< p<2 -2\sqrt \gamma$ and all $\tau_0\leq\tau_1<\tau_2$. Note that the above argument fails to prove $\eqref{eq:sub:rpEstimate}$ for any $p$ when $q=\frac{1}{3}$, as in this case, $\gamma=1$.

Nevertheless, the following, we derive \eqref{eq:sub:rpEstimate} with $0<p<1$ for all $0<q<\frac{1}{2}$. To do so, we untwist $\d_v\varphi$ into $t^q\d_v(r\psi)$ in \eqref{eq:sub:rpmethod-pointwise1}. We rearrange the error term as well. In fact, we get
\begin{align*}
    p r^{p-1}|\d_v \varphi|^2=&pt^{2q}r^{p-1}|\d_v(r\psi)|^2+q^2pr^{p+1}t^{4q-2}\psi^2+ qpr^{p-1}t^{3q-1}\d_v(r^2\psi^2)\\=&pt^{2q}r^{p-1}|\d_v(r\psi)|^2+q^2pr^{p+1}t^{4q-2}\psi^2+ \d_v\left( qpr^{p+1}t^{3q-1}\psi^2\right) \\&-qp(p-1)r^pt^{3q-1}\psi^2-pq(3q-1)r^{p+1}t^{4q-2}\psi^2\,,
\end{align*}
and
\begin{align*}
    2r^p \W \varphi\d_v \varphi=&\d_v\left(r^{p+2} t^{2q} \W  \psi^2\right)-pq(1-2q)r^{p+1}t^{4q-2}\psi^2-2(q-1)r^{p+2}t^{3q-1}\W  \psi^2\,,
\end{align*}
which together imply
\begin{align*}
    p r^{p-1}|\d_v \varphi|^2+ 2r^p \W \varphi\d_v \varphi=& \d_v\left(r^{p+2} t^{2q} \W  \psi^2 + qpr^{p+1}t^{3q-1}\psi^2\right)\\&+pr^{p-1}t^{2q}|\d_v(r\psi)|^2 + qp(1-p)r^p t^{3q-1}\psi^2 + 2(1-q)r^{p+2} t^{3q-1}\W \psi^2\,.
\end{align*}
Notice that all terms in the last line above are non-negative if $0\leq p\leq1$. In addition, we can add the untwisted version of \eqref{eq:sub:rp-zeroth-order-term},
\begin{equation*}
        \d_v(r^p t^{2q}\psi^2)+(2-p) r^{p-1} t^{2q}\psi^2 -2 qr^p t^{3q-1}\psi^2-2 r^{p-1}t^{2q}\psi\d_v(r\psi)=0\,,
    \end{equation*}
times $\gamma'>0$ to \eqref{eq:sub:rpmethod-pointwise1} in order to deduce
\begin{align*}
        &\d_u(r^p |\d_v \varphi|^2)+ \d_v\left(r^{p-2}t^{2q}|\rslnabla(r \psi)|^2+r^{p+2} t^{2q} \W  \psi^2 + qpr^{p+1}t^{3q-1}\psi^2+\gamma' r^pt^{2q}\psi^2\right)  - \rsldiv( 2r^{p-2}\d_v\varphi \rslnabla\varphi)\\& +(2-p)r^{p-3}t^{2q}|\rslnabla(r\psi)|^2+pt^{2q}r^{p-1}|\d_v(r\psi)|^2 +\left(qp(1-p) t^{q-1}+ 2(1-q)r^{2} t^{q-1}\W +\gamma' (2-p)r^{-1}\right)r^pt^{2q}\psi^2
        \\&=2\gamma' qr^p t^{3q-1}\psi^2+2\gamma' r^{p-1}t^{2q}\psi\d_v(r\psi)\,,
    \end{align*}
    in which everything on the left-hand side is either a boundary term or non-negative for $0\leq p \leq 1$. Both terms on the right-hand side are error terms. Nevertheless, they can be controlled by the terms on the other side as long as $0<p<1$ and $\gamma'<\frac{p(1-p)}{2}$. So, for any $0<p<1$ there exists $\gamma'>0$ that enables us to estimate the error terms. Therefore, integrating over $\D^{\tau_2}_{\tau_1}$ and discarding non-essential terms leads to
    \begin{align*}
    \begin{split}
        &\int_{\D^{\tau_2}_{\tau_1}} r^{p-1} \left(t^{2q} |\d_v (r\psi)|^2 + \frac{1}{r^2}|\rslnabla\varphi|^2+ \frac{1}{r^2}\varphi^2\right) \, dudvd\omega + \int_{\C_{\tau_2}} r^p|\d_v\varphi|^2 \, dvd\omega + \int_{\I^{\tau_2}_{\tau_1}} r^{p-2}\left(|\rslnabla\varphi|^2 +\varphi^2\right)\, dud\omega\\
        &\lesssim \int_{\C_{\tau_1}} r^p|\d_v\varphi|^2 \, dvd\omega + \underbrace{\int_{\{r=\rho\}} r^{p-2}\left(|\rslnabla\varphi|^2 +\varphi^2 + rt^{q-1}\varphi^2 +r^2 t^{2q-2}\varphi^2\right) \, dud\omega- \int_{\{r=\rho\}} r^p|\d_v\varphi|^2 \, dvd\omega}_{II}\,,
    \end{split}
\end{align*}
for $0<p<1$. To bound $II$, first notice that the power of $r$ and $\frac{1}{t}$ in weights are irrelevant as they are bounded\footnote{Throughout this paper, we use “bounded” to mean bounded up to a universal constant, hidden in the $\lesssim$ notation.} by $\rho$ and $t_0$ in this region. Therefore, the mean-value argument can again be applied to deduce $II\lesssim\E_q[\psi](\tau_1)$.  
The last piece of the proof is to control the twisted derivative $r^{p-1}|\d_v \varphi|^2$. This can easily be obtained by recalling 
\begin{equation*}
    r^{p-1}|\d_v \varphi|^2\lesssim r^{p-1} \left(t^{2q} |\d_v (r\psi)|^2 + t^{2q-2}\varphi^2\right)\stackrel{\eqref{eq:general:t-r-lemma}}{\lesssim} r^{p-1} t^{2q} |\d_v (r\psi)|^2 + r^{p-3}\varphi^2\,.
\end{equation*}
So, we finally obtain \eqref{eq:sub:rpEstimate} for $0<p<2-\sigma_q$ if we let $\sigma_q=2(\frac{q|1-2q|}{|1-q|^2})^\frac{1}{2}$.
Lastly, note that \eqref{eq:radiation:r2Estimate} follows if we directly integrate \eqref{eq:sub:rpmethod-pointwise1} over the region $\D^{\tau_2}_{\tau_1}$.
\end{proof}
\subsection{The case $q>\frac{1}{2}$}
Let $\varepsilon\geq 0$ be any non-negative number, and define 
\begin{align}\label{eq:sup:energy-sigma-tau}
        \begin{split}
            \E_q^\varepsilon[\psi](\tau):= &\int_{\S_{\tau}}t^{2-2q-\varepsilon} \left( |\d_t(t^{q}\psi)|^2+|\d_r\psi|^2+\frac{1}{r^2}|\d_r(r\psi)|^2+\frac{1}{r^2}|\rslnabla\psi|^2+t^{2q-2}\psi^2 + \frac{1}{1+r^2}\psi^2\right)\,dx \\
            &+\int_{\C_{\tau}} t^{2-4q-\varepsilon}\left(\frac{1}{r^2}|\d_v(rt^q\psi)|^2+|\d_v(t^q\psi)|^2+\frac{t^{2q}}{r^2}|\rslnabla\psi|^2+t^{4q-2} \psi^2+ \frac{t^{2q}}{r^2}\psi^2\right)\,r^2dvd\omega\,.
        \end{split}
    \end{align}
Notice that the energy flux \eqref{eq:sup:energy-sigma-tau} is the same as \eqref{eq:intro:energy-sigma-tau} since $\mu_q^\varepsilon=4q-2+\varepsilon$ for $\frac{1}{2}<q<1$. The following proposition establishes the boundedness of the above energy flux.
\begin{prop} [$\Sigma_\tau$-energy boundedness for waves on FLRW with $q>\frac{1}{2}$]\label{prop:sup:energy-bnd}
    We suppose $\psi$ solves the wave equation with $\frac{1}{2}<q<1$. Then for any $\tau_2>\tau_1\geq\tau_0$ and any non-negative $\varepsilon\geq 0$, we have
    \begin{align}
        \E_q^\varepsilon[\psi](\tau_2)&\lesssim\E_q^\varepsilon[\psi](\tau_1)\,,\label{eq:sup:Eng1-boundedness}\\
        \int_{\I^{\tau_2}_{\tau_1}} t^{2-4q-\varepsilon}\left(\frac{1}{r^2}|\d_u(rt^q\psi)|^2+|\d_u(t^q\psi)|^2+\frac{t^{2q}}{r^2}|\rslnabla\psi|^2+t^{4q-2} \psi^2+ \frac{t^{2q}}{r^2}\psi^2\right)\,r^2dud\omega&\lesssim\E_q^\varepsilon[\psi](\tau_1)\,.\label{eq:sup:Eng1-Null}
    \end{align}
    Moreover, for any non-negative $\varepsilon\geq0$, we also have
    \begin{align}
        \label{eq:sup:Eng1-bulk}
        \int_{\R^{\tau_2}_{\tau_1}}  t^{1-2q-\varepsilon}\left( |\nabla \psi|^2+ |\d_t (t^q\psi)|^2 + t^{2q-2} \psi ^2 \right)\, dxdt \lesssim \E_q^\varepsilon[\psi](\tau_1)\,.
    \end{align}
\end{prop}
\begin{proof}
    Although estimates \eqref{eq:sup:Eng1-boundedness}-\eqref{eq:sup:Eng1-bulk} are formulated with twisting function $t^{-q}$, to begin the proof we need to use a different twisting function. Therefore, set $\beta_1=t^{q-1}$ with the associated potential $\V_1=-\frac{\Box_g\beta_1}{\beta_1}=2(q-1)(2q-1)t^{-2}<0$ for $\frac{1}{2}<q<1$, and consider the vector field $X=t^{2-3q-\varepsilon}\d_t$. A simple computation shows

    \begin{align}
        \tilde K^{\beta_1,X}[\psi]=& (4q-2+\frac{\varepsilon}{2})t^{-q-1-\varepsilon}|\d_t( t^{1-q}\psi)|^2+ \frac{\varepsilon}{2} t^{1-5q-\varepsilon}|\nabla\psi|^2 +(1-q+\frac{\varepsilon}{2})t^{1-3q}\V_1 \psi^2\,,\label{eq:sup:engbound1-beta1-bulk}
    \end{align}
    and
    \begin{align}
        &\tilde J^{\beta_1,X}_t[\psi]=\frac{t^{-\varepsilon}}{2}\left( t^{-q}|\d_t(t^{1-q}\psi)|^2+t^{2-5q}|\nabla\psi|^2+t^{2-3q}\V_1 \psi^2\right)\,,\notag\\
        &\tilde J^{\beta_1,X}_v[\psi]=\frac{t^{-\varepsilon}}{2}\left(t^{-2q}|\d_v(t^{1-q}\psi)|^2 + \frac{t^{2-4q}}{r^2}|\rslnabla\psi|^2+ t^{2-2q} \V_1 \psi^2\right)\,,\notag\\
        &\tilde J^{\beta_1,X}_u[\psi]=\frac{t^{-\varepsilon}}{2}\left(t^{-2q}|\d_u(t^{1-q}\psi)|^2 + \frac{t^{2-4q}}{r^2}|\rslnabla\psi|^2+ t^{2-2q} \V_1 \psi^2\right)\,.\notag
    \end{align}
    Note that since $\V_1<0$, none of the above quantities are necessarily non-negative. However, we can untwist the space-time term to obtain non-negative energy flux and bulk terms. Indeed, if we also take the volume factor $t^{3q}$ into account, we have
    \begin{equation}\label{eq:sup-engbound-untwst}
        t^{2q-1-\varepsilon}|\d_t(t^{1-q}\psi)|^2=t^{1-\varepsilon}|\d_t\psi|^2 + \frac{(1-q)(1-q+\varepsilon)}{t^{1+\varepsilon}}\psi^2+(1-q)\d_t(t^{-\varepsilon}\psi^2)\,.
    \end{equation}
    Thus, integration over $\R^{\tau_2}_{\tau_1}$ and the divergence identity \eqref{eq:TEMT:DivergenceIdentity} give
    \begin{align}\label{eq:sup:EngTws1-q}
        \begin{split}
            &\int_{\R^{\tau_2}_{\tau_1}} \left((4q-2+\frac{\varepsilon}{2})t^{1-\varepsilon}|\d_t\psi|^2 + \frac{\varepsilon}{2} t^{1-2q-\varepsilon}|\nabla\psi|^2 +\frac{\varepsilon}{2}(1-q)(3q-1+\varepsilon)t^{-1-\varepsilon}\psi^2\right) \,dtdx \\
        &+ \frac{1}{2}\int_{\S_{\tau_2}}t^{-\varepsilon} \left( t^{2q}|\d_t(t^{1-q}\psi)|^2+t^{2-2q}|\nabla\psi|^2+ 2(1-q)(2q-1+\frac{\varepsilon}{2})  \psi^2\right)\,dx \\
        &+\frac{1}{2}\int_{\C_{\tau_2}} t^{-\varepsilon}\left(|\d_v(t^{1-q}\psi)|^2+\frac{t^{2-2q}}{r^2}|\rslnabla\psi|^2+ 2(1-q)(2q-1+\frac{\varepsilon}{2})\psi^2\right)\,r^2dvd\omega\\
        &+\frac{1}{2}\int_{\I^{\tau_2}_{\tau_1}}t^{-\varepsilon} \left(|\d_u(t^{1-q}\psi)|^2+\frac{t^{2-2q}}{r^2}|\rslnabla\psi|^2+ 2(1-q)(2q-1+\frac{\varepsilon}{2})  \psi^2\right)\,r^2dud\omega\\
        &=\frac{1}{2}\int_{\S_{\tau_1}} t^{-\varepsilon}\left( t^{2q}|\d_t(t^{1-q}\psi)|^2+t^{2-2q}|\nabla\psi|^2+ 2(1-q)(2q-1+\frac{\varepsilon}{2})  \psi^2\right)\,dx \\
        &+\frac{1}{2}\int_{\C_{\tau_1}} t^{-\varepsilon}\left(|\d_v(t^{1-q}\psi)|^2+\frac{t^{2-2q}}{r^2}|\rslnabla\psi|^2+ 2(1-q)(2q-1+\frac{\varepsilon}{2})  \psi^2\right)\,r^2dvd\omega\,.
        \end{split}
    \end{align}
    Note that all terms in the above identity are non-negative. In order to rewrite this energy identity in terms of the twisting function $\beta=t^{-q}$, observe that
    \begin{align}
    \begin{split}\label{eq:sup:engbound-untwst-m}
        |\d_m(t^q \psi)|^2 \lesssim t^{4q-2} |\d_m(t^{1-q} \psi)|^2+ |\d_m t^{2q-1}|^2  (t^{1-q}\psi)^2,\\
          t^{4q-2} |\d_m(t^{1-q} \psi)|^2\lesssim|\d_m(t^q \psi)|^2+ |\d_m t^{2q-1}|^2  (t^{1-q}\psi)^2,
    \end{split}
    \end{align}
    in which $m$ is either $v,u,$ or $t$. Therefore, if we let
    \begin{align*}
        \check\E_q^\varepsilon[\psi](\tau):= &\int_{\S_{\tau}}  t^{-\varepsilon} \left( t^{2-2q} |\d_t(t^q \psi)|^2 +t^{2-2q}|\nabla \psi|^2 + \psi^2\right) \,dx \\
        +&\int_{\C_{\tau}}t^{-\varepsilon} \left( t^{2-4q}|\d_v(t^q \psi)|^2 + \frac{t^{2-2q}}{r^2}|\rslnabla \psi|^2 + \psi^2 \right)  \, r^2dvd\omega\,,
    \end{align*}
    we then infer that
    \begin{align}\label{eq:sup:check-energy-estimate}
        \begin{split}
            \check\E_q^\varepsilon[\psi](\tau_2)&\lesssim\check\E_q^\varepsilon[\psi](\tau_1)\,,\\
            \int_{\I^{\tau_2}_{\tau_1}} t^{-\varepsilon}\left( t^{2-4q}|\d_u(t^q \psi)|^2 + \frac{t^{2-2q}}{r^2}|\rslnabla \psi|^2 + \psi^2 \right)  \, r^2dud\omega &\lesssim\check\E_q^\varepsilon[\psi](\tau_1)\,.
        \end{split}
    \end{align}
    Moreover, the basic twisting inequality
    $$ |\d_t(t^q \psi)|^2 \lesssim t^{2q} |\d_t \psi|^2+ q^2 t^{2q-2}|\psi|^2,$$
    enable us to also control the term $|\d_t(t^q \psi)|^2$ in \eqref{eq:sup:Eng1-bulk}. Since other bulk terms are already estimated in the above energy identity, for $\varepsilon>0$, we get
    \begin{align*}
        \int_{\R^{\tau_2}_{\tau_1}}  t^{1-2q-\varepsilon}\left( |\nabla \psi|^2+ |\d_t (t^q\psi)|^2 + t^{2q-2} \psi ^2 \right)\, dxdt \lesssim \check \E_q^\varepsilon[\psi](\tau_1)\,.
    \end{align*}
    The above finishes the proof of \eqref{eq:sup:Eng1-bulk} for $\varepsilon>0$ as
    $\check\E_q^\varepsilon[\psi](\tau)\lesssim \E_q^\varepsilon[\psi](\tau)$. Before proving \eqref{eq:sup:Eng1-bulk} with $\varepsilon=0$, we estimate complete the proof of \eqref{eq:sup:Eng1-boundedness} and \eqref{eq:sup:Eng1-Null}.
    
    Similarly, we use the twisting function $\beta_2=\frac{t^{q-1}}{r}$ to derive an additional energy estimate. The associated potential remains the same, $\V_2=\V_1=2(q-1)(2q-1)t^{-2}<0$, and we continue to use the same vector field $X=t^{2-3q}\d_t$. Observe that for $\beta_2$, we have
    \begin{equation*}
        \tilde K^{\beta_2,X}[\psi]= (4q-2+\frac{\varepsilon}{2})t^{-q-1-\varepsilon}|\d_t( t^{1-q}\psi)|^2+\frac{\varepsilon}{2}\frac{ t^{1-5q-\varepsilon}}{r^2}|\d_r\psi|^2+ \frac{\varepsilon}{2}\frac{ t^{1-5q-\varepsilon}}{r^2}|\rslnabla\psi|^2 +(1-q+\frac{\varepsilon}{2})t^{1-3q}\V_1 \psi^2\,.
    \end{equation*}
    Because of identity \eqref{eq:sup-engbound-untwst}, we are again able to untwist $\tilde K^{\beta_2,X}[\psi]$ and obtain non-negative terms. We therefore infer
    \begin{align*}
        &\int_{\S_{\tau_2}} t^{-\varepsilon}\left( t^{2q}|\d_t(t^{1-q}\psi)|^2+\frac{t^{2-2q}}{r^2}|\d_r(r\psi)|^2+\frac{t^{2-2q}}{r^2}|\rslnabla\psi|^2+\psi^2\right)\,dx \\
        &+\int_{\C_{\tau_2}} t^{-\varepsilon}\left(\frac{1}{r^2}|\d_v(rt^{1-q}\psi)|^2+\frac{t^{2-2q}}{r^2}|\rslnabla\psi|^2+ \psi^2\right)\,r^2dvd\omega\\
        &+\int_{\I^{\tau_2}_{\tau_1}} t^{-\varepsilon}\left(\frac{1}{r^2}|\d_u(rt^{1-q}\psi)|^2+\frac{t^{2-2q}}{r^2}|\rslnabla\psi|^2+ \psi^2\right)\,r^2dud\omega\\ &\lesssim\int_{\S_{\tau_1}} t^{-\varepsilon}\left( t^{2q}|\d_t(t^{1-q}\psi)|^2+\frac{t^{2-2q}}{r^2}|\d_r(r\psi)|^2+\frac{t^{2-2q}}{r^2}|\rslnabla\psi|^2+\psi^2\right)\,dx \\
        &+\int_{\C_{\tau_1}} t^{-\varepsilon}\left(\frac{1}{r^2}|\d_v(rt^{1-q}\psi)|^2+\frac{t^{2-2q}}{r^2}|\rslnabla\psi|^2+ \psi^2\right)\,r^2dvd\omega\,.
    \end{align*}
    Notice that \eqref{eq:sup:engbound-untwst-m} is still true if we replace $\psi$ with $r\psi$. Therefore, we can again rewrite the above estimate to arrive at an energy estimate twisted with $\beta=t^{-q}$. Combining this twisted energy estimate with \eqref{eq:sup:check-energy-estimate} leads to
    \begin{align}\label{eq:sup:EngTwsq}
        \begin{split}
            &\int_{\S_{\tau_2}} t^{-\varepsilon}\left( t^{2-2q}|\d_t(t^{q}\psi)|^2+t^{2-2q}|\d_r\psi|^2+\frac{t^{2-2q}}{r^2}|\d_r(r\psi)|^2+\frac{t^{2-2q}}{r^2}|\rslnabla\psi|^2+\psi^2 \right)\,dx \\
        &+\int_{\C_{\tau_2}} t^{-\varepsilon}\left(\frac{t^{2-4q}}{r^2}|\d_v(rt^q\psi)|^2+t^{2-4q}|\d_v(t^q\psi)|^2+\frac{t^{2-2q}}{r^2}|\rslnabla\psi|^2+ \psi^2\right)\,r^2dvd\omega\\
        &+\int_{\I^{\tau_2}_{\tau_1}} t^{-\varepsilon}\left(\frac{t^{2-4q}}{r^2}|\d_u(rt^q\psi)|^2+t^{2-4q}|\d_u(t^q\psi)|^2+\frac{t^{2-2q}}{r^2}|\rslnabla\psi|^2+ \psi^2\right)\,r^2dud\omega\\ &\lesssim\int_{\S_{\tau_1}} t^{-\varepsilon}\left( t^{2-2q}|\d_t(t^{q}\psi)|^2+t^{2-2q}|\d_r\psi|^2+\frac{t^{2-2q}}{r^2}|\d_r(r\psi)|^2+\frac{t^{2-2q}}{r^2}|\rslnabla\psi|^2+\psi^2\right)\,dx \\
        &+\int_{\C_{\tau_1}} t^{-\varepsilon}\left(\frac{t^{2-4q}}{r^2}|\d_v(rt^q\psi)|^2+t^{2-4q}|\d_v(t^q\psi)|^2+\frac{t^{2-2q-\varepsilon}}{r^2}|\rslnabla\psi|^2+ \psi^2\right)\,r^2dvd\omega\,.
        \end{split}
    \end{align}
    To finish the proof of \eqref{eq:sup:Eng1-boundedness} and \eqref{eq:sup:Eng1-Null}, notice that we can estimate the zeroth-order term $\frac{t^{2-2q}}{1+r^2}\psi^2$ in $\E_q^\varepsilon[\psi]$ by recalling \eqref{eq:sub:engbound-untwst}. 

    In order to establish \eqref{eq:sup:Eng1-bulk} for $\varepsilon=0$, first let $\beta_3=1$ and $w=t^{1-3q}$. We then are left with
    \begin{align*}
        K^w[\psi]&= t^{1-3q}\left(-|\d_t\psi|^2+\frac{1}{t^{2q}} |\nabla \psi|^2\right)\,,\qquad
        &J^w_t[\psi]=t^{1-3q}\psi \d_t\psi - \frac{1-3q}{2t^{3q}}\psi^2\,,\\
        J^w_v[\psi]&=t^{1-3q}\psi \d_v\psi - \frac{1-3q}{2t^{2q}}\psi^2\,,\qquad
        &J^w_u[\psi]=t^{1-3q}\psi \d_u\psi - \frac{1-3q}{2t^{2q}}\psi^2\,,
    \end{align*}
    which lead to
    \begin{align}
        \begin{split}\label{eq:sup:EngEps0_1}
            \int_{\R^{\tau_2}_{\tau_1}} t^{1-2q} |\nabla \psi|^2 \,dtdx\lesssim \int_{\R^{\tau_2}_{\tau_1}} t|\d_t\psi|^2 \,dtdx +\int_{\I^{\tau_2}_{\tau_1}} \left(t^{2-2q}|\d_u\psi|^2+\psi^2\right)\,dvd\omega \\+\sum_{i=1,2}\left( \int_{\S_{\tau_i}} \left(t^{2}|\d_t\psi|^2+\psi^2\right)\,dx
        +\int_{\C_{\tau_i}} \left(t^{2-2q}|\d_v\psi|^2+\psi^2\right)\,dvd\omega\right)\,.
        \end{split}
    \end{align}
    All the terms on the right-hand side can be controlled by identity \eqref{eq:sup:EngTws1-q} with $\varepsilon=0$, thanks to the fact that
    \begin{align*}
        t^{2-2q}|\d_m \psi|^2 \lesssim |\d_m (t^{1-q}\psi)|^2 + |\d_m(t^{1-q})|^2 \psi^2\,,
    \end{align*}
    with $m$ denoting either $t, v$ or $u$.
    Next, we again use the twisting function $\beta=t^{-q}$ along with the modification function $w=t^{1-3q}$ to arrive at 
    \begin{align*}
        \tilde K^{\beta,w}[\psi]&= t^{1-5q}\left(-|\d_t(t^q\psi)|^2+|\nabla \psi|^2\right)+\frac{q(3q-1)}{t^{1+3q}}\psi^2\,,\qquad
        &\tilde J^{\beta,w}_t[\psi]=t^{1-4q}\psi \d_t(t^q\psi) - \frac{1-3q}{2t^{3q}}\psi^2\,,\\
        \tilde J^{\beta,w}_v[\psi]&=t^{1-4q}\psi \d_v(t^q\psi) - \frac{1-3q}{2t^{2q}}\psi^2\,,\qquad
        &\tilde J^{\beta,w}_u[\psi]=t^{1-4q}\psi \d_u(t^q\psi) - \frac{1-3q}{2t^{2q}}\psi^2\,.
    \end{align*}
    After integration over $\R^{\tau_2}_{\tau_1}$, we have
    \begin{align}\label{eq:sup:EngEps0_2}
        \begin{split}
            \int_{\R^{\tau_2}_{\tau_1}} \left(t^{1-2q}|\d_t(t^q\psi)|^2 + \frac{1}{t}\psi^2\right) \,dtdx \lesssim \int_{\R^{\tau_2}_{\tau_1}} t^{1-2q} |\nabla \psi|^2 \,dtdx+\int_{\I^{\tau_2}_{\tau_1}} \left(t^{2-4q}|\d_u(t^q\psi)|^2+\psi^2\right)\,dvd\omega \\+\sum_{i=1,2}\left( \int_{\S_{\tau_i}} \left(t^{2-2q}|\d_t(t^q\psi)|^2+\psi^2\right)\,dx
        +\int_{\C_{\tau_i}} \left(t^{2-4q}|\d_v(t^q\psi)|^2+\psi^2\right)\,dvd\omega\right)\,.
        \end{split}
    \end{align}
    Note that all the boundary terms are bounded by \eqref{eq:sup:EngTwsq} with $\varepsilon=0$. Therefore, \eqref{eq:sup:EngEps0_1} and \eqref{eq:sup:EngEps0_2} imply \eqref{eq:sup:Eng1-bulk} for $\varepsilon=0$.
\end{proof}
The next proposition establishes the ILED estimate for the case $q\geq \frac{1}{2}$.
\begin{prop}[Integrated local energy decay on $\Sigma_\tau$ for waves on FLRW with $q>\frac{1}{2}$]\label{prop:sup:Morawetz}
Let $\varepsilon\geq0$ be any non-negative number and  $\delta>0$ be any small positive numbers. If $\psi$ solves the initial value problem \eqref{eq:intro:IVP-wave} for $\frac{1}{2}<q<1-\varepsilon$, then we have
    \begin{align}\label{eq:sup:Morawetz}
        \int_{\R^{\tau_2}_{\tau_1}} &\left(\frac{t^{2-3q-\varepsilon}}{  (1+r)^{1+\delta}}(|\d_t(t^q\psi)|^2+|\d_r\psi|^2)+  \frac{t^{2-3q-\varepsilon}}{r^3}|\rslnabla \psi |^2+\frac{t^{2-3q-\varepsilon}}{1+r^{3+\delta}}\psi^2+ \frac{t^{-q-\varepsilon}}{r} \psi^2\right)\,dxdt\lesssim \E_q^\varepsilon[\psi](\tau_1)\,.
    \end{align}
\end{prop}
\begin{proof}
    To begin with, let $\beta=t^{-q}$, $Y_1=t^{2-4q-\varepsilon}\d_r$, and $z_1=\frac{t^{2-4q-\varepsilon}}{r}$ and recall that $\V=\frac{-q(2q-1)}{t^2}$. With this setting, we have $Y_1^\nu\tilde S_\nu =0$ and
    \begin{align*}
        &\tilde K^{Y_1,z_1} [\psi] =  \frac{t^{2-4q-\varepsilon}}{t^{2q} r^3}|\rslnabla \psi |^2+ (\frac{1-3q-\varepsilon}{2}) (2-4q-\varepsilon) \frac{t^{-4q-\varepsilon}}{r} \psi^2 + (2-4q-\varepsilon) t^{1-5q-\varepsilon}\d_r\psi \d_t(t^q\psi)\,.
    \end{align*}
    Moreover,
    \begin{align*}
        \begin{split}
            \tilde J^{Y_1,z_1}_t=&\frac{t^{2-5q-\varepsilon}}{r} \psi \d_t(t^q \psi)+ (\frac{-2+4q+\varepsilon}{2})\frac{t^{1-4q-\varepsilon}}{r}\psi^2 + t^{2-5q-\varepsilon} \d_t(t^q \psi) \d_r \psi\,,\\
            \tilde J^{Y_1,z_1}_v=&\frac{t^{2-6q-\varepsilon}}{2}|\d_v(t^q \psi)|^2 - \frac{t^{2-4q-\varepsilon}}{2r^2}|\rslnabla\psi|^2+\frac{1}{2}t^{2-2q-\varepsilon}\V\psi^2\\&+ \frac{t^{2-4q-\varepsilon}}{r} \psi\d_v(t^q \psi)+( \frac{t^{2-4q-\varepsilon}}{2r^2}- \frac{2-4q-\varepsilon}{2r}t^{1-3q-\varepsilon})\psi^2\,,\\
            \tilde J^{Y_1,z_1}_u=&-\frac{t^{2-6q-\varepsilon}}{2}|\d_u(t^q \psi)|^2 + \frac{t^{2-4q-\varepsilon}}{2r^2}|\rslnabla\psi|^2-\frac{1}{2}t^{2-2q-\varepsilon}\V\psi^2\\&+ \frac{t^{2-4q-\varepsilon}}{r} \psi\d_v(t^q \psi)+( -\frac{t^{2-4q-\varepsilon}}{2r^2}- \frac{2-4q-\varepsilon}{2r}t^{1-3q-\varepsilon})\psi^2\,.
        \end{split}
    \end{align*}
    Integration over $\R^{\tau_2}_{\tau_1}\setminus \{r<\epsilon\}$ and letting $\epsilon\to 0$ give rise to
    \begin{align}
    \begin{split}
        \int_{\R^{\tau_2}_{\tau_1}}& \left( \frac{t^{2-3q-\varepsilon}}{r^3}|\rslnabla \psi |^2+ \frac{t^{-q-\varepsilon}}{r} \psi^2\right)\,dxdt+\int_{\S_{\tau_2}}  \frac{t^{1-q-\varepsilon}}{r}\psi^2\,dx+\int_{\C_{\tau_2}}\ \frac{t^{1-q-\varepsilon}}{r}\psi^2\,r^2dvd\omega+\int_{\I^{\tau_2}_{\tau_1}}  \frac{t^{1-q-\varepsilon}}{r}\psi^2\,r^2dud\omega\\&\lesssim \int_{\R^{\tau_2}_{\tau_1}}  t^{1-2q-\varepsilon}\left( |\nabla \psi|^2+  |\d_t (t^q\psi)|^2  \right)\, dxdt+ \int_{\S_{\tau_1}} \frac{t^{1-q-\varepsilon}}{r}\psi^2\,dx+\int_{\C_{\tau_1}} \frac{t^{1-q-\varepsilon}}{r}\psi^2\,r^2dvd\omega\\
        &+\int_{\I^{\tau_2}_{\tau_1}} \left(t^{2-4q-\varepsilon}|\d_u(t^q\psi)|^2+\frac{t^{2-2q-\varepsilon}}{r^2}|\rslnabla\psi|^2+t^{-\varepsilon} \psi^2+ \frac{t^{2-2q-\varepsilon}}{r^2}\psi^2\right)\,r^2dud\omega+\E_q^\varepsilon[\psi](\tau_2)+\E_q^\varepsilon[\psi](\tau_1) \,.  \label{eq:sup:Morawetz-y1z1}
    \end{split}
    \end{align}
   Due to Proposition \ref{prop:sup:energy-bnd}, all the terms on the right-hand side above can be directly controlled by $\E_q^\varepsilon[\psi](\tau_1)$, except the initial boundary terms that contain $\frac{t^{1-q-\varepsilon}}{r}\psi^2$. Nevertheless, those terms are also bounded by the zeroth-order term in $\E_q^\varepsilon[\psi](\tau_1)$. Indeed, for large $r$, in view of Lemma \ref{lem:general:t-r}, we have
   \begin{align*}
        \int_{\S_{\tau_1}} \frac{t^{1-q-\varepsilon}}{r}\psi^2\,dx+\int_{\C_{\tau_1}} \frac{t^{1-q-\varepsilon}}{r}\psi^2\,r^2dvd\omega\lesssim  \int_{\S_{\tau_1}} \frac{t^{2-2q-\varepsilon}}{1+r^2}\psi^2\,dx+\int_{\C_{\tau_1}} \frac{t^{2-2q-\varepsilon}}{1+r^2}\psi^2\,r^2dvd\omega\lesssim \E_q^\varepsilon[\psi](\tau_1)\,.
   \end{align*}
   And for small $r$, we recall \eqref{eq:sub:engbound-untwst} to get
\begin{align*}
    \int_{\S_{\tau_1}} \frac{t^{1-q-\varepsilon}}{r}\psi^2\,dx\lesssim \int_{\S_{\tau_1}}\left(t^{1-q-\varepsilon}r|\d_r \psi|^2 +\frac{t^{1-q-\varepsilon}}{r}|\d_r( r\psi)|^2\right)\,dx\lesssim \int_{\S_{\tau_1}}t^{2-2q-\varepsilon}\left(|\d_r \psi|^2 +\frac{1}{r^2}|\d_r( r\psi)|^2\right)\,dx\,.
\end{align*}
The other term on $\C_{\tau_1}$ can also be controlled similarly.
   
    Controlling other derivatives over the region $\R^{\tau_2}_{\tau_1}$ requires multiplying with new vector field. To do so, let $\delta>0$, $Y_2=\frac{t^{2-4q-\varepsilon}}{(1+r)^\delta}\d_r$, and $z_2=\frac{t^{2-4q-\varepsilon}}{r(1+r)^\delta}$, and observe that
    \begin{align*}
                \tilde K^{Y_2} [\psi] =&  (4q-2+\varepsilon)\frac{t^{2-4q-\varepsilon}}{t (1+r)^\delta} t^{-q}\d_r\psi \d_t(t^q\psi)+ \frac{t^{2-4q-\varepsilon}}{t^{2q} r (1+r)^\delta}\left(|\d_t(t^q\psi)|^2-|\d_r\psi|^2-t^{2q}\V\psi^2\right)\\&-\frac{\delta t^{2-4q-\varepsilon}}{2t^{2q}  (1+r)^{1+\delta}}\left(|\d_t(t^q\psi)|^2+|\d_r\psi|^2-\frac{1}{r^2}|\rslnabla\psi|^2-t^{2q}\V\psi^2\right) \,,\\
                \tilde K^{z_2} [\psi] =& \frac{t^{2-4q-\varepsilon}}{t^{2q} r (1+r)^\delta}\left(-|\d_t(t^q\psi)|^2+|\d_r\psi|^2+\frac{1}{r^2}|\rslnabla\psi|^2\right)\\&+\left(\frac{t^{2-4q-\varepsilon}}{r(1+r)^\delta}\V -\frac{1}{2}\Box z_2 + q(4q-2+\varepsilon)\frac{t^{-4q-\varepsilon}}{r(1+r)^\delta}\right)\psi^2\,,
    \end{align*}
   Furthermore, we have
    \begin{align*}
        \begin{split}
            \tilde J^{Y_2,z_2}_t=&\frac{t^{2-5q-\varepsilon}}{r(1+r)^\delta} \psi \d_t(t^q \psi)+ (\frac{-2+4q+\varepsilon}{2})\frac{t^{1-4q-\varepsilon}}{r(1+r)^\delta}\psi^2 + \frac{t^{2-5q-\varepsilon}}{(1+r)^\delta} \d_t(t^q \psi) \d_r \psi\,,\\
            \tilde J^{Y_2,z_2}_v=&\frac{t^{2-6q-\varepsilon}}{2(1+r)^\delta}|\d_v(t^q \psi)|^2 - \frac{t^{2-4q-\varepsilon}}{2r^2(1+r)^\delta}|\rslnabla\psi|^2+\frac{1}{2(1+r)^\delta}t^{2-2q-\varepsilon}\V\psi^2\\&+ \frac{t^{2-4q-\varepsilon}}{r(1+r)^\delta} \psi\d_v(t^q \psi)+( \frac{t^{2-4q-\varepsilon}}{2r^2(1+r)^\delta}+\frac{\delta t^{2-4q-\varepsilon}}{2r(1+r)^{1+\delta}} - \frac{2-4q-\varepsilon}{2r(1+r)^\delta}t^{1-3q-\varepsilon})\psi^2\,,\\
            \tilde J^{Y_2,z_2}_u=&-\frac{t^{2-6q-\varepsilon}}{2(1+r)^\delta}|\d_u(t^q \psi)|^2 + \frac{t^{2-4q-\varepsilon}}{2r^2(1+r)^\delta}|\rslnabla\psi|^2-\frac{1}{2(1+r)^\delta}t^{2-2q-\varepsilon}\V\psi^2\\&+ \frac{t^{2-4q-\varepsilon}}{r(1+r)^\delta} \psi\d_v(t^q \psi)+( -\frac{t^{2-4q-\varepsilon}}{2r^2(1+r)^\delta}-\frac{\delta t^{2-4q-\varepsilon}}{2r(1+r)^{1+\delta}}- \frac{2-4q-\varepsilon}{2r(1+r)^\delta}t^{1-3q-\varepsilon})\psi^2\,.
        \end{split}
    \end{align*}
    As before, we integrate over $\R^{\tau_2}_{\tau_1}\setminus \{r<\epsilon\}$, let $\epsilon\to 0$, and use \eqref{eq:TEMT:DivergenceIdentity} to derive
    \begin{align*}
        \int_{\R^{\tau_2}_{\tau_1}} &\left(\frac{t^{2-3q-\varepsilon}}{  (1+r)^{1+\delta}}(|\d_t(t^q\psi)|^2+|\d_r\psi|^2)+t^{3q}\psi^2 \Box z_2\right)\,dxdt\\\lesssim& \int_{\R^{\tau_2}_{\tau_1}} \left(\frac{t^{2-3q-\varepsilon}}{  r^3(1+r)^\delta}|\rslnabla\psi|^2+ \frac{t^{1-2q-\varepsilon}}{(1+r)^\delta} (|\d_r \psi|^2+ |\d_t (t^q\psi)|^2)+\frac{t^{-q-\varepsilon}}{r(1+r)^\delta}\psi^2\right)\,dxdt\\
        &+\sum_{i=1,2} \left(\int_{\S_{\tau_i}}  |\tilde J^{Y_2,z_2}_t|  \,t^{3q} dx + \int_{\C_{\tau_i}}  |\tilde J^{Y_2,z_2}_v| \, t^{2q}r^2 dvd\omega \right)+ \int_{\I^{\tau_2}_{\tau_1}}  |\tilde J^{Y_2,z_2}_u| \, t^{2q}r^2 dud\omega\,.
    \end{align*}
    Now thanks to \eqref{eq:sup:Eng1-boundedness}, \eqref{eq:sup:Eng1-Null}, \eqref{eq:sup:Eng1-bulk}, and \eqref{eq:sup:Morawetz-y1z1} we infer
    \begin{align*}
        \int_{\R^{\tau_2}_{\tau_1}} &\left(\frac{t^{2-3q-\varepsilon}}{  (1+r)^{1+\delta}}(|\d_t(t^q\psi)|^2+|\d_r\psi|^2)+  \frac{t^{2-3q-\varepsilon}}{r^3}|\rslnabla \psi |^2+ \frac{t^{-q-\varepsilon}}{r} \psi^2+t^{3q}\psi^2 \Box z_2\right)\,dxdt\lesssim \E_q^\varepsilon[\psi](\tau_1)\,.
    \end{align*}
    To finish, just note that similar to \eqref{eq:sub:Boxz2}, we have 
    \begin{align*}
       (4q-2+\varepsilon)(1-q-\varepsilon)\frac{t^{-q-\varepsilon}}{r(1+r)^\delta} \psi^2 +\frac{t^{2-3q-\varepsilon}}{1+r^{3+\delta}}\psi^2\lesssim t^{3q}\psi^2\Box z_2\,,
    \end{align*}
    which is positive for $\frac{1}{2}<q<1-\varepsilon$.
\end{proof}

Before stating the next proposition, recall that $\sigma_q=2\sqrt \frac{q|1-2q|}{|1-q|^2}\geq 1$ for $q\geq \frac{1+2\sqrt 2}{7}$. 
\begin{prop}[$r^p$-estimates for waves on FLRW with $q>\frac{1}{2}$]\label{prop:sup:rp} 
Let $\varepsilon=0$ and $\mu_q=\mu_q^0=4q-2$
and for $\frac{1}{2}<q<1$, set $\sigma_q=2\sqrt \frac{q|1-2q|}{|1-q|^2}>0$. Then for any solution $\psi$ to the wave equation \eqref{eq:intro:IVP-wave}, any $\tau_2>\tau_1\geq \tau_0$, and any
$$0<p< \max\{1, 2-\sigma_q\}\,,$$
we have
    \begin{align}\label{eq:sup:rp}
    \begin{split}
        &\int_{\D^{\tau_2}_{\tau_1}} t^{-{\mu_q}}r^{p-1} \left( |\d_v \varphi|^2 + r^2|\d_v(t^q\psi)|^2+\frac{1}{r^2}|\rslnabla\varphi|^2+\frac{1}{r^2}\varphi^2\right) \, dudvd\omega + \int_{\C_{\tau_2}} t^{-{\mu_q}}r^p|\d_v\varphi|^2 \, dvd\omega \\&+  \int_{\I^{\tau_2}_{\tau_1}} t^{-{\mu_q}}r^{p-2}\left(|\rslnabla\varphi|^2 +\varphi^2\right)\, dud\omega
        \lesssim \int_{\C_{\tau_1}} t^{-{\mu_q}}r^p|\d_v\varphi|^2 \, dvd\omega +\E_q^0[\psi](\tau_1)\,,
    \end{split}
\end{align}
in which $\varphi=rt^q\psi$.
\end{prop}
\begin{proof}
    We start by considering $\varphi=t^qr\psi$ as in the proof of Proposition \ref{prop:sub:rp}. Recall that $\varphi$ satisfies 
    \begin{equation*}
        -\d_{uv} \varphi + \frac{1}{r^2} \rslDelta\varphi + \X \varphi=0\,,
    \end{equation*}
    where $\X=-t^{2q}\V=q(2q-1)t^{2q-2}>0$. Unlike in \eqref{eq:sub:rpmethod-pointwise1}, here we need a family of vector fields whose weights depend on $t$ as well. Therefore, multiply the above equation by $-2t^{-{\mu_q}}r^p \d_v \varphi$ to obtain
    \begin{align}
        \d_u(t^{-{\mu_q}}r^p |\d_v \varphi|^2)  - \rsldiv(2 t^{-{\mu_q}}r^{p-2}\d_v\varphi \rslnabla\varphi) + \d_v(t^{-{\mu_q}}r^{p-2}|\rslnabla \varphi|^2)\notag +{\mu_q} t^{q-1-{\mu_q}}r^p( |\d_v \varphi|^2+r^{-2} |\rslnabla \varphi|^2)\notag \\
        + pt^{-{\mu_q}} r^{p-1}|\d_v \varphi|^2+ (2-p)t^{-{\mu_q}}r^{p-3}|\rslnabla\varphi|^2=2t^{-{\mu_q}}r^p \X \varphi\d_v \varphi\,.\label{eq:sup:rp-pointwise1}
    \end{align}
    On the other hand, we can add the following identity
    \begin{equation*}
        \d_v(t^{-{\mu_q}}r^{p-2} \varphi^2)+(2-p)t^{-{\mu_q}}r^{p-3}\varphi^2 + {\mu_q} t^{q-1-{\mu_q}}r^{p-2}\varphi^2 = 2t^{-{\mu_q}}r^{p-2}\varphi\d_v \varphi\,.
    \end{equation*}
    times $\gamma>0$ to \eqref{eq:sup:rp-pointwise1} to get
    \begin{align}
        \begin{split}
            &\d_u(t^{-{\mu_q}}r^p |\d_v \varphi|^2)  - \rsldiv(2 t^{-{\mu_q}}r^{p-2}\d_v\varphi \rslnabla\varphi) + \d_v\left(t^{-{\mu_q}}r^{p-2}|\rslnabla \varphi|^2+\gamma t^{-{\mu_q}}r^{p-2} \varphi^2\right) \\&+{\mu_q} t^{q-1-{\mu_q}}r^p( |\d_v \varphi|^2+r^{-2} |\rslnabla \varphi|^2) + pt^{-{\mu_q}} r^{p-1}|\d_v \varphi|^2+(2-p)t^{-{\mu_q}}r^{p-3}\left(|\rslnabla\varphi|^2+ \gamma\varphi^2\right) \\&+ \gamma{\mu_q} t^{q-1-{\mu_q}}r^{p-2}\varphi^2=2t^{-{\mu_q}}r^{p-2} (\gamma+r^2\X )\varphi\d_v \varphi\,.
        \end{split}\label{eq:sup:rp-pointwise-gamma}
    \end{align}
    We remark that the right-hand side of the above equation is considered as an error term, and we aim to control it by the terms on the other side. We do this similar to what we do in the proof of Proposition \ref{prop:sub:rp}. So first assume $\frac{\rho}{2}\leq \tau_1<\tau_2$ or equivalently $u\geq 0$. Lemma \ref{lem:general:t-r} once again implies $r^2t^{2q-2}\leq \frac{1}{(1-q)^2}$ in the region $\U^+$. If we let $\gamma=\frac{q|1-2q|}{|1-q|^2}$, we then have
    \begin{align*}
         2t^{-{\mu_q}}r^{p-2}\left(\gamma+ r^2|\X|  \right) |\varphi||\d_v \varphi|&\leq 2 t^{-{\mu_q}}r^{p-2}\left(\gamma+ \frac{q|2q-1|}{|1-q|^2}  \right)|\varphi||\d_v \varphi|\leq  2\gamma^\frac{3}{2}t^{-{\mu_q}} r^{p-3}|\varphi|^2+2 \gamma^ \frac{1}{2}t^{-{\mu_q}}r^{p-1}|\d_v \varphi|^2\,,
    \end{align*}
    for any point in $\U^+$. Suppose $2\sqrt \gamma< p<2 -2\sqrt \gamma$. Integrating \eqref {eq:sup:rp-pointwise-gamma} over $\D^{\tau_2}_{\tau_1}$  for $\frac{\rho}{2}\leq \tau_1<\tau_2$ and using the above inequality lead to
    \begin{align*}
    \begin{split}
        &\int_{\D^{\tau_2}_{\tau_1}}( t^{q-1-{\mu_q}}r^p+ t^{-{\mu_q}}r^{p-1}) \left( |\d_v \varphi|^2 + \frac{1}{r^2}|\rslnabla\varphi|^2+\frac{1}{r^2}\varphi^2\right) \, dudvd\omega \\
        & + \int_{\I^{\tau_2}_{\tau_1}} t^{-{\mu_q}}r^{p-2}(|\rslnabla\varphi|^2+\varphi^2) \, dud\omega +\int_{\C_{\tau_2}} t^{-{\mu_q}}r^p|\d_v\varphi|^2 \, dvd\omega\\&\lesssim \int_{\C_{\tau_1}} t^{-{\mu_q}}r^p|\d_v\varphi|^2 \, dvd\omega +\underbrace{\int_{\{r=\rho,\tau_1\leq \tau\leq \tau_2\}} t^{-{\mu_q}}r^{p-2}(|\rslnabla\varphi|^2 +\varphi^2) \, dud\omega- \int_{\{r=\rho,\tau_1\leq \tau\leq \tau_2\}} t^{-{\mu_q}}r^p|\d_v\varphi|^2 \, dvd\omega}_{I}\,.
    \end{split}
    \end{align*}
    The error term $I$  above should be compared to that in \eqref{eq:sub:rp-one-errors}, and the ILED estimate \eqref{eq:sup:Morawetz} can indeed control it. Indeed, similar to \eqref{eq:sub:MeanValueThm} using the mean value theorem, we find
    \begin{align}\label{eq:sup:rp_I}
        I&\lesssim\int_{\R^{\tau_2}_{\tau_1}} \left(\frac{t^{2-3q}}{  (1+r)^{1+\delta}}(|\d_t(t^q\psi)|^2+|\d_r\psi|^2)+  \frac{t^{2-3q}}{r^3}|\rslnabla \psi |^2+\frac{t^{2-3q}}{1+r^{3+\delta}}\psi^2\right)\,dxdt\lesssim\E_q^0[\psi](\tau_1)\,.
    \end{align}
    Now suppose $\tau_0\leq \tau_1<\tau_2\leq \frac{\rho}{2}$. According to Lemma \ref{lem:general:t-r}, there exists $\fC$ such that $rt^{q-1}\leq \fC$. Therefore, for any small $\epsilon$, there exists $\fC^*$ such that
    \begin{align*}
         2t^{-{\mu_q}}r^{p-2}\left(\gamma+ r^2|\X|  \right) |\varphi||\d_v \varphi|&\leq 2 t^{-{\mu_q}}r^{p-2}\left(\gamma+ \fC^2 \frac{q|2q-1|}{|1-q|^2}  \right)|\varphi||\d_v \varphi|\leq  \epsilon t^{-{\mu_q}} r^{p-3}|\varphi|^2+\fC^* t^{-{\mu_q}}r^{p-1}|\d_v \varphi|^2\,.
    \end{align*}
    Thus, as long as $\epsilon < (2-p)\gamma$ and $\tau_0\leq \tau_1<\tau_2\leq \frac{\rho}{2}$, the above inequality and \eqref{eq:sup:rp-pointwise-gamma} give
    \begin{align}
    \begin{split}\label{eq:sup:rp:finite}
        &\int_{\D^{\tau_2}_{\tau_1}}( t^{q-1-{\mu_q}}r^p+ t^{-{\mu_q}}r^{p-1}) \left( |\d_v \varphi|^2 + \frac{1}{r^2}|\rslnabla\varphi|^2+\frac{1}{r^2}\varphi^2\right) \, dudvd\omega \\
        & + \int_{\I^{\tau_2}_{\tau_1}} t^{-{\mu_q}}r^{p-2}(|\rslnabla\varphi|^2+\varphi^2) \, dud\omega +\int_{\C_{\tau_2}} t^{-{\mu_q}}r^p|\d_v\varphi|^2 \, dvd\omega\\&\leq \fC^*\left(\int_{\D^{\tau_2}_{\tau_1}} \left( t^{-{\mu_q}}r^{p-1}|\d_v \varphi|^2\right) \, dudvd\omega+ \int_{\C_{\tau_1}} t^{-{\mu_q}}r^p|\d_v\varphi|^2 \, dvd\omega +\E_q^0[\psi](\tau_1)\right)\,.
    \end{split}
    \end{align}
    where by abusing the notation, we have used the same constant $\fC^*$. As in the proof of \ref{prop:sub:rp}, we use Grönwall's inequality, Lemma \ref{lem:Gronwall}, for the function $f(\tau)=\int_{\C_{\tau}} t^{-{\mu_q}}r^p|\d_v\varphi|^2 \, dvd\omega$ to get 
    \begin{align*}
        f(\tau_2)\leq \fC^*\left(\int_{\C_{\tau_1}} t^{-{\mu_q}}r^p|\d_v\varphi|^2 \, dvd\omega +\E_q^0[\psi](\tau_1)\right)e^{\fC^*(\tau_2-\tau_1)}\lesssim\int_{\C_{\tau_1}} t^{-{\mu_q}}r^p|\d_v\varphi|^2 \, dvd\omega +\E_q^0[\psi](\tau_1)\,.
    \end{align*}
    The above estimate holds for any $\tau_0\leq \tau_1<\tau_2\leq \frac{\rho}{2}$. Therefore, the bulk term on the right-hand side of \eqref{eq:sup:rp:finite} is also bounded by the boundary terms.
    
    We also need to estimate $ t^{-{\mu_q}} r^{p+1}|\d_v(t^q\psi)|^2$. So, note that
    \begin{align*}
        t^{-{\mu_q}} r^{p+1}|\d_v(t^q\psi)|^2+{\mu_q} t^{q-1-{\mu_q}}r^{p-2}\varphi^2=t^{-{\mu_q}}r^{p-1}|\d_v \varphi|^2+ (p-1) t^{-{\mu_q}}r^{p-3}\varphi^2 -\d_v(t^{-{\mu_q}}r^{p-2}\varphi^2)\,,
    \end{align*}
    which implies
    \begin{align}\label{eq:sup:untwisting-r}
        \int_{\D^{\tau_2}_{\tau_1}} t^{-{\mu_q}}r^{p+1}|\d_v(t^q\psi)|^2\,dudvd\omega &\lesssim \int_{\D^{\tau_2}_{\tau_1}} t^{-{\mu_q}}r^{p-1}|\d_v \varphi|^2+ t^{-{\mu_q}}r^{p-3}\varphi^2 \,dudvd\omega + I\notag\\
        &\lesssim \int_{\C_{\tau_1}} r^p|\d_v\varphi|^2 \, dvd\omega + \E_q^0[\psi](\tau_1)\,,
    \end{align}
    for all $0<p<2$. This completes the proof of \eqref{eq:sup:rp} for $\sigma_q<p<2-\sigma_q$ if we let $\sigma_q=2\sqrt \gamma= 2\sqrt \frac{q|1-2q|}{|1-q|^2}$. Note that this range of $p$ is empty if $q\geq \frac{1+2\sqrt 2}{7}$. 
    
    To prove $r^p$-estimates for $0<p<1$ across the sup-radiation regime ($\frac{1}{2}<q<1$), we untwist a few terms in identity \eqref{eq:sup:rp-pointwise1}. In fact, by expanding the following terms we obtain
    \begin{align*}
        -2t^{\mu_q}r^p \X \varphi \d_v\varphi&= -q(2q-1)\d_v(r^{p+2}\psi^2)-2q^2(2q-1)t^{q-1}r^{p+2}\psi^2+pq(2q-1)r^{p+1}\psi^2\,,\\
        pt^{-\mu_q}r^{p-1}|\d_v\varphi|^2&=pt^{2-2q}r^{p-1}|\d_v(r\psi)|^2+pq(2q-1)r^{p+1}\psi^2+qp(1-p)t^{1-q}r^p \psi^2 + qp\d_v(t^{1-q}r^{1+p}\psi^2)\,\\
        \mu_q t^{q-1-\mu_q}r^p |\d_v\varphi|^2&=2(2q-1)t^{1-q}r^p |\d_v r\psi|^2 + 2q^2 (2q-1)t^{q-1}r^{p+2} \psi^2\\&+2q(2q-1)\d_v(r^{p+2}\psi^2)-2pq(2q-1) r^{p+1} \psi^2\,.
    \end{align*}
    Thus,
    \begin{align*}
        &pt^{-\mu_q}r^{p-1}|\d_v\varphi|^2+\mu_q t^{q-1-\mu_q}r^p |\d_v\varphi|^2  -2t^{\mu_q}r^p \X \varphi \d_v\varphi\\&= \left(pt^{2-2q}r^{p-1}+2(2q-1)t^{1-q}r^p\right)|\d_v(r\psi)|^2+\d_v\left(q(2q-1)r^{p+2}\psi^2+qpt^{1-q}r^{1+p}\psi^2\right)+ qp(1-p)t^{1-q}r^p \psi^2\,.
    \end{align*}
    Therefore, \eqref{eq:sup:rp-pointwise1} turns to 
    \begin{align*}
        \d_u(t^{-{\mu_q}}r^p |\d_v \varphi|^2)  - \rsldiv( 2t^{-{\mu_q}}r^{p-2}\d_v\varphi \rslnabla\varphi) + \d_v\left(t^{-{\mu_q}}r^{p-2}|\rslnabla \varphi|^2+q(2q-1)r^{p+2}\psi^2+qpt^{1-q}r^{1+p}\psi^2\right) \\+{\mu_q} t^{q-1-{\mu_q}}r^{p-2} |\rslnabla \varphi|^2
        + (2-p)t^{-{\mu_q}}r^{p-3}|\rslnabla\varphi|^2+
        pt^{2-2q}r^{p-1}|\d_v(r\psi)|^2+2(2q-1)t^{1-q}r^p|\d_v(r\psi)|^2\\+ qp(1-p)t^{1-q}r^p \psi^2=0\,.
    \end{align*}
    Integrating over $\D^{\tau_2}_{\tau_1}$, for any $0<p\leq 1$ gives
    \begin{align*}
    \begin{split}
        &\int_{\D^{\tau_2}_{\tau_1}} r^{p-1}\left(t^{2-2q}|\d_v(r\psi)|^2+\frac{1}{t^{\mu_q}r^2}|\rslnabla \psi|^2 + (1-p)rt^{1-q}\psi^2\right) \, dudvd\omega  + \int_{\I^{\tau_2}_{\tau_1}} t^{-{\mu_q}}r^{p-2}|\rslnabla\varphi|^2 \, dud\omega\\& +\int_{\C_{\tau_2}} t^{-{\mu_q}}r^p|\d_v\varphi|^2 \, dvd\omega\lesssim \int_{\C_{\tau_1}} t^{-{\mu_q}}r^p|\d_v\varphi|^2 \, dvd\omega - \int_{\{r=\rho,\tau_1\leq \tau\leq \tau_2\}} t^{-{\mu_q}}r^p|\d_v\varphi|^2 \, dvd\omega\\& +\int_{\{r=\rho,\tau_1\leq \tau\leq \tau_2\}}\left( t^{-{\mu_q}}r^{p-2}|\rslnabla\varphi|^2 +r^{p+2}\psi^2+t^{1-q}r^{p+1} \psi^2\right) \, dud\omega\,.
    \end{split}
    \end{align*}
    The last two terms on the right-hand side can again be controlled by $\E^0_q[\psi](\tau_1)$, similar to \eqref{eq:sup:rp_I}. Also observe that
    \begin{align*}
        \d_v(r^p t^{2-2q} \psi^2) + \frac{2-p}{2} r^{p-1}t^{2-2q}\psi^2 \leq (2-2q)t^{1-q}r^p \psi^2 + 2 r^{p-1} t^{2-2q} |\d_v (r\psi)|^2\,,
    \end{align*}
    and
    \begin{align*}
            t^{2-4q}r^{p-1}|\d_v\varphi|^2\lesssim t^{2-2q}r^{p-1}|\d_v(r\psi)|^2 + r^{p+1}\psi^2\,.
        \end{align*}
    Using the above inequalities and \eqref{eq:general:t-r-lemma}, we obtain
    \begin{align*}
    \begin{split}
        &\int_{\D^{\tau_2}_{\tau_1}} r^{p-1}\left(t^{2-4q}|\d_v\varphi|^2+t^{2-2q}|\d_v(r\psi)|^2+\frac{1}{t^{\mu_q}r^2}|\rslnabla \psi|^2 + rt^{1-q}\psi^2+t^{2-2q}\psi^2\right) \, dudvd\omega \\
        & + \int_{\I^{\tau_2}_{\tau_1}} t^{-{\mu_q}}r^{p-2}\left(|\rslnabla\varphi|^2+\varphi^2\right) \, dud\omega +\int_{\C_{\tau_2}} t^{-{\mu_q}}r^p|\d_v\varphi|^2 \, dvd\omega\\&\lesssim \int_{\C_{\tau_1}} t^{-{\mu_q}}r^p|\d_v\varphi|^2 \, dvd\omega+ \E^0_q [\psi](\tau_1)\,.
    \end{split}
    \end{align*}
    Finally, \eqref{eq:sup:untwisting-r} provides the required estimate for the last remaining term \eqref{eq:sup:rp}. This completes the proof of \eqref{eq:sup:rp} for $0<p<1$ and $\frac{1}{2}<q<1$.
    \end{proof}
\section{Decay of waves}\label{sec:decay}
This section is dedicated to prove Corollary \ref{cor:general:decay}. The first part of the corollary follows from Proposition \ref{Prop:decay-energy}, while the second part is given by Proposition \ref{prop:decay:PointwiseDecay} (note that in view of \eqref{eq:general:t-r-lemma}, we have $\log(1+r)\lesssim t^{\frac{\varepsilon}{2}}$, and therefore \eqref{eq:decay:point-wise2} is sufficient for the case $q\geq \frac{1+2\sqrt 2}{7}$). The decay of first-order derivatives, as stated in the third part of Corollary \ref{cor:general:decay}, is proved in Propositions \ref{prop:decay:commuting-radiation} and \ref{prop:decay:derivatives}. All the propositions in this section can be regarded as consequences of the energy estimates in Section \ref{sec:EnergyEstimates} and general functional inequalities provided in Section \ref{sec:functional_ineq}.

Throughout this section, we set
\begin{equation}\label{eq:decay:parameters-q}
    \begin{aligned}
        \varepsilon&=0, \qquad &\mu_q^\varepsilon &=\mu_q=0, \qquad & \sigma &=0,\qquad & \mathrm{ if }&\, q=\frac{1}{2}\,,\\
        \varepsilon&=0, \qquad &\mu_q^\varepsilon &=\mu_q=0, \qquad & \sigma_q &<\sigma<1,\qquad & \mathrm{ if }&\, q<\frac{1}{2}\,,\\
        \varepsilon&=0, \qquad &\mu_q^\varepsilon &=\mu_q=4q-2, \qquad & \sigma_q &<\sigma<1,\qquad & \mathrm{ if }&\, \frac{1}{2}<q<\frac{1+2\sqrt 2}{7}\,,\\
        0<\varepsilon&\ll (1-q), \qquad &\mu_q^\varepsilon &=4q-2+\varepsilon, \qquad & \sigma &=1,\qquad & \mathrm{ if }&\, \frac{1+2\sqrt 2}{7}\leq q<1\,,
        \end{aligned}
\end{equation}
unless stated otherwise.
\subsection{Decay of the energy flux}
We start with showing the decay of energy flux $\E^\varepsilon_q [\psi](\tau)$, defined in \eqref{eq:intro:energy-sigma-tau}. Recall that $\sigma_q=2(\frac{q|1-2q|}{|1-q|^2})^\frac{1}{2}$. 

To state the proposition, first recall that for a solution $\psi$ to the initial value problem \eqref{eq:intro:IVP-wave}, $\varphi$ is defined as $t^qr \psi$. Now, for any $0<q<1$, $0\leq i\leq 2$ and $\varepsilon\geq 0$ define
\begin{align}
    f_q^\varepsilon[\psi](\tau)&:= \E_q^\varepsilon[\psi](\tau)\,,\label{eq:decay:def-f}\\
    h_{q,i}^\varepsilon[\psi](\tau)&:=\int_{\C_{\tau}} t^{-{\mu_q^\varepsilon}}r^i|\d_v\varphi|^2 \, dvd\omega\,,\label{eq:decay:def-hi}
\end{align}
and
    \begin{equation}\label{eq:decay:initial-energy}
        \begin{aligned}
        \E_q[\psi_0,\psi_1]&:=h_1^0[\psi](\tau_0)+h_{2-\sigma_q}^0[\psi](\tau_0)+f_q^0[\psi](\tau_0)\qquad& {\rm{ if }} \, \sigma_q&<1\,,\\
        \E_q[\psi_0,\psi_1]&:=h^0_{1-\frac{\varepsilon}{1-q}}[\psi](\tau_0)+f_q^0[\psi](\tau_0)\qquad&{\rm{ if }}\, \sigma_q&\geq1\,,
    \end{aligned}
    \end{equation}
   
Henceforth, we suppress the dependence of $f^\varepsilon_q[\psi]$ and $h_i^\varepsilon[\psi]$ on $\psi$, $q$, and $\varepsilon$ whenever it is clear from the context. Observe that in view of Lemma \ref{eq:general:t-r-lemma}, we have
\begin{equation}\label{eq:decay:h0-hEps}
    h_1^\varepsilon(\tau) \lesssim h^0 _{1-\frac{\varepsilon}{1-q}}(\tau)\,.
\end{equation}
\begin{prop}[Decay of twisted energy of waves on FLRW]\label{Prop:decay-energy}
    For $q\neq \frac{1}{3}$, let $\varepsilon$, $\sigma$, and $\mu_q^\varepsilon$ be as defined in \eqref{eq:decay:parameters-q}. Assume that $\psi$ solves the initial value problem \eqref{eq:intro:IVP-wave}. The energy flux $ \E^\varepsilon_q [\psi](\tau)$ defined in \eqref{eq:intro:energy-sigma-tau} decays as

    \begin{equation*}
        \begin{aligned}
        \E_q^0 [\psi](\tau)&\lesssim \frac{\E_q[\psi_0,\psi_1]}{\tau^{2-\sigma}} \quad&{\rm{ if }}\, \sigma_q&<1 \quad(\mathrm{\textit{i.\@e.\@}\ }q<\frac{1+2\sqrt 2}{7})\,,\\
        \E^\varepsilon_q [\psi](\tau)&\lesssim \frac{\E_q[\psi_0,\psi_1]}{\tau}\qquad &{\rm{ if }}\, \sigma_q&\geq1\quad(\mathrm{\textit{i.\@e.\@}\ }q\geq\frac{1+2\sqrt 2}{7})\,,
    \end{aligned}
    \end{equation*}
     where $\E_q[\psi_0,\psi_1]$ is defined in \eqref{eq:decay:initial-energy}.
\end{prop}
\begin{proof}
    First, assume $q<\frac{1+2\sqrt 2}{7}$. The result follows from Lemma \ref{lem:general:hierarchy-estimate}, and we just need to show that its assumptions are valid. The key elements to do so are the energy estimates in Propositions \ref{prop:sub:EB}, \ref{prop:sub:ILED} and \ref{prop:sub:rp} in case $q\leq \frac{1}{2}$ and Propositions \ref{prop:sup:energy-bnd}, \ref{prop:sup:Morawetz} and \ref{prop:sup:rp} when $q>\frac{1}{2}$.
    
    First we consider $0<q\leq \frac{1}{2}$ and $q\neq \frac{1}{3}$ where $0<\sigma_q<1$. We therefore can let $p=1$ in \eqref{eq:sub:rpEstimate} to obtain
    \begin{align*}
        &\int_{\D^{\tau_2}_{\tau_1}}  \left( |\d_v \varphi|^2 +r^2|\d_v(t^q \psi)|^2+ \frac{1}{r^2}|\rslnabla\varphi|^2+ \frac{1}{r^2}\varphi^2\right) \, dudvd\omega + \int_{\C_{\tau_2}} r|\d_v\varphi|^2 \, dvd\omega \\&+  \int_{\I^{\tau_2}_{\tau_1}} \frac{1}{r^2}\left(|\rslnabla\varphi|^2 +\varphi^2\right)\, dud\omega
        \lesssim \int_{\C_{\tau_1}} r|\d_v\varphi|^2 \, dvd\omega + \E_q [\psi](\tau_1)\,,
    \end{align*}
for any $\tau_2>\tau_1\geq \tau_0$.
Next, we show how to control the bulk terms on the compact region $\R=\R^{\tau_2}_{\tau_1}\setminus D^{\tau_2}_{\tau_1}$. The ILED estimate \eqref{eq:sub:Morawetz} allow us to control weighted norms of $t^q\psi$ and its first derivatives, and since $r<\rho$ on this region, the power of $r$-weight is not an issue. However, we also need to control derivatives of $\varphi=rt^q\psi$. To do so, recall that the identity
\begin{equation}\label{eq:sub:untwist-r}
    |\d_v \varphi|^2= r^2|\d_v(  t^q\psi)|^2 + 2 r t^q \psi\d_v( t^q\psi)+ t^{2q}\psi^2\,,
\end{equation}
 leads to
 \begin{align*}
     \int_{\R^{\tau_2}_{\tau_1}\setminus D^{\tau_2}_{\tau_1}} |\d_v \varphi|^2 \, dudvd\omega \lesssim \int_{\R^{\tau_2}_{\tau_1}\setminus D^{\tau_2}_{\tau_1}} \left(|\d_v ( t^q \psi)|^2 + \frac{t^{2q}}{r^2}\psi^2 \right) \, r^2dudvd\omega\,.
 \end{align*}
The same is also true for the $u$-derivative. So summing them up an using \eqref{eq:general:t-r-lemma}, we arrive at
\begin{align}\label{eq:sub:p=1}
    \begin{split}
        &\int_{\R^{\tau_2}_{\tau_1}\setminus D^{\tau_2}_{\tau_1}} \left(|\d_v(t^q \psi)|^2+ |\d_u (t^q \psi)|^2 +\frac{1}{r^2}|\d_v\varphi|^2+ \frac{1}{r^2}|\d_u \varphi|^2 + \frac{t^{2q}}{r^2}|\rslnabla\psi|^2+ \frac{t^{2q}}{r^2}\psi^2+t^{4q-2}\psi^2 \right)\, r^2dudvd\omega \\
        &+\int_{\D^{\tau_2}_{\tau_1}}\left( \frac{1}{r^2} |\d_v \varphi|^2+|\d_v (t^q \psi)|^2 + \frac{t^{2q}}{r^2}|\rslnabla\psi|^2 + \frac{t^{2q}}{r^2}\psi^2\right) \, r^2 dudvd\omega +\int_{\C_{\tau_2}} r|\d_v\varphi|^2+t^{4q-2}\psi^2 \, dvd\omega\\
        &\lesssim \int_{\C_{\tau_1}} r|\d_v\varphi|^2 \, dvd\omega + \E_q [\psi](\tau_1)\,.
    \end{split}
\end{align}
 Thus, with the choices  $f=f^0$ and $h_i=h_i^0$ as defined in \eqref{eq:decay:def-f} and \eqref{eq:decay:def-hi}, estimate \eqref{eq:general:1-EE} follows directly from \eqref{eq:sub:p=1}. Moreover, \eqref{eq:sub:Eng-boundedness} coincides exactly with \eqref{eq:general:0-EE}.

To derive the last estimate, \textit{i.\@e.\@}\ \eqref{eq:general:2-EE}, we let $p=2-\sigma$ in \eqref{eq:sub:rpEstimate} to deduce
\begin{align*}
    \int^{\tau_2}_{\tau_1} h_{1-\sigma}(\tau)\,d\tau + h_{2-\sigma}(\tau_2) \lesssim h_{2-\sigma}(\tau_1) +f(\tau_1)\,.
\end{align*}
When $q=\frac{1}{2}$ and $\sigma=0$, we instead use \eqref{eq:radiation:r2Estimate}to derive the above estimate. Furthermore, with the help of Lemma \ref{lem:interpolation-r^p}, we get 
\begin{align}\label{eq:sub:interpolating_h1}
\begin{split}
    \int^{\tau_2}_{\tau_1} h_1^{\frac{1}{1-\sigma}}(\tau)\,d\tau &\lesssim \int^{\tau_2}_{\tau_1} h_{1-\sigma}(\tau)h^{\frac{\sigma}{1-\sigma}}_{2-\sigma}(\tau)\,d\tau \lesssim \left(h_{2-\sigma}(\tau_1) +f(\tau_1)\right)^{\frac{\sigma}{1-\sigma}}\int^{\tau_2}_{\tau_1} h_{1-\sigma}(\tau)\,d\tau\\&\lesssim\left(h_{2-\sigma}(\tau_1) +f(\tau_1)\right)^{\frac{\sigma}{1-\sigma}}(h_{2-\sigma}(\tau_1) +f(\tau_1))=\left(h_{2-\sigma}(\tau_1) +f(\tau_1)\right)^{\frac{1}{1-\sigma}}\,,
\end{split}
\end{align}
from which \eqref{eq:general:2-EE} follows for $\tau_2>\tau_1>\tau_0$. Thus, Lemma \ref{lem:general:hierarchy-estimate} gives
\begin{equation*}
    f(\tau)\lesssim \frac{\E_q[\psi_0,\psi_1]}{\tau^{2-\sigma}}\,,
\end{equation*}
in which $\sigma>\sigma_q=2(\frac{q|1-2q|}{|1-q|^2})^\frac{1}{2}$ and $\E_q[\psi_0,\psi_1]=h_1(\tau_0)+h_{2-\sigma}(\tau_0)+f(\tau_0)$.

Now suppose $\frac{1}{2}<q<\frac{1+2\sqrt 2}{7}$. 
Note that \eqref{eq:sup:Eng1-boundedness} provide
    \begin{align*}
        f(\tau_2)\lesssim f(\tau_1)\,,
    \end{align*}
    for $\tau_1<\tau_2$. So we recover \eqref{eq:general:0-EE}.

    To derive \eqref{eq:general:1-EE}, we let $p=1$ in \eqref{eq:sup:rp}. In addition, Lemma \ref{lem:general:t-r} gives us control over the term containing $t^{-\varepsilon}\psi^2$ as well. Therefore,
    \begin{align*}
        &\int_{\D^{\tau_2}_{\tau_1}}  t^{-{\mu_q^\varepsilon}}\left( |\d_v \varphi|^2 +r^2|\d_v(t^q \psi)|^2+ \frac{1}{r^2}|\rslnabla\varphi|^2+ \frac{1}{r^2}\varphi^2+ t^{2q-2} \varphi^2\right) \, dudvd\omega + \int_{\C_{\tau_2}} t^{-{\mu_q^\varepsilon}}r|\d_v\varphi|^2 \, dvd\omega \\&+  \int_{\I^{\tau_2}_{\tau_1}} t^{-{\mu_q^\varepsilon}} \frac{1}{r^2}\left(|\rslnabla\varphi|^2 +\varphi^2\right)\, dud\omega
        \lesssim \int_{\C_{\tau_1}} t^{-{\mu_q^\varepsilon}}r|\d_v\varphi|^2 \, dvd\omega + \E_q^\varepsilon[\psi](\tau_1)\,,
    \end{align*}
    with ${\mu_q^\varepsilon}=4q-2+\varepsilon>0$. So now we just need to use the ILED estimate for the compact region $\R^{\tau_2}_{\tau_1}\setminus D^{\tau_2}_{\tau_1}$. Before doing that, recall \eqref{eq:sub:untwist-r} to untwist $\varphi$ and derive
    \begin{align*}
         \int_{\R^{\tau_2}_{\tau_1}\setminus D^{\tau_2}_{\tau_1}} t^{2-2q-\varepsilon}|\d_v \varphi|^2 \, dudvd\omega \lesssim \int_{\R^{\tau_2}_{\tau_1}\setminus D^{\tau_2}_{\tau_1}} t^{2-2q-\varepsilon}\left(|\d_v ( t^q \psi)|^2 + \frac{t^{2q}}{r^2}\psi^2 \right) \, r^2dudvd\omega\,.
    \end{align*}
    Inequality above and \eqref{eq:sup:Morawetz} give
    \begin{align*}
        \int_{\R^{\tau_2}_{\tau_1}\setminus\D^{\tau_2}_{\tau_1}} &t^{2-4q-\varepsilon}\left(|\d_v(t^q\psi)|^2+|\d_u(t^q\psi)|^2+  \frac{t^{2q}}{r^2}|\rslnabla \psi |^2+t^{2q}\psi^2+ \frac{t^{4q-2}}{r} \psi^2\right)\,r^2dudvd\omega\lesssim  \E_q^\varepsilon[\psi](\tau_1)\,.
    \end{align*}
    So we obtain \eqref{eq:general:1-EE} as required.

    The last piece of Lemma \ref{lem:general:hierarchy-estimate} is obtained by setting $\sigma>\sigma_q=2(\frac{q|1-2q|}{|1-q|^2})^\frac{1}{2}$ and letting $p=2-\sigma$ in \eqref{eq:sup:rp}. Indeed, \eqref{eq:general:2-EE} follows similar to \eqref{eq:sub:interpolating_h1} with the help of interpolation Lemma \ref{lem:interpolation-r^p}. 
    
    Now suppose $q\geq\frac{1+2\sqrt 2}{7}$ and set $\epsilon=\frac{\varepsilon}{1-q}$. Recall that if we let $p=1-\epsilon$, then  from \eqref{eq:sup:rp} we infer that
    \begin{align*}
        &\int_{\D^{\tau_2}_{\tau_1}} t^{2-4q} r^{-\epsilon}\left( |\d_v \varphi|^2 + r^2|\d_v(t^q\psi)|^2+\frac{1}{r^2}|\rslnabla\varphi|^2+t^{2q}\psi^2\right) \, dudvd\omega + \int_{\C_{\tau_2}} t^{2-4q}r^{1-\epsilon}|\d_v\varphi|^2 \, dvd\omega \\&\lesssim \int_{\C_{\tau_1}} t^{2-4q}r^{1-\epsilon}|\d_v\varphi|^2 \, dvd\omega +\E_q^0[\psi](\tau_1)\,.
    \end{align*}
     Using estimate  \eqref{eq:general:t-r-lemma}, we can also control the term involving $\psi^2$ with different weight factors. Moreover, \eqref{eq:general:t-r-lemma} gives $t^{-\varepsilon}\lesssim\frac{1}{r^\epsilon}$. Thus, we obtain
    \begin{align*}
        &\int_{\D^{\tau_2}_{\tau_1}} t^{2-4q-\varepsilon}\left( |\d_v \varphi|^2 + r^2|\d_v(t^q\psi)|^2+\frac{1}{r^2}|\rslnabla\varphi|^2+t^{2q}\psi^2+t^{4q-2}r^2\psi^2\right) \, dudvd\omega + \int_{\C_{\tau_2}} t^{2-4q}r^{1-\frac{\varepsilon}{1-q}}|\d_v\varphi|^2 \, dvd\omega \\&\lesssim\int_{\C_{\tau_1}} t^{2-4q}r^{1-\frac{\varepsilon}{1-q}}|\d_v\varphi|^2 \, dvd\omega +\E_q^0[\psi](\tau_1)\,.
    \end{align*}
    For the compact region $\R^{\tau_2}_{\tau_1}\setminus\D^{\tau_2}_{\tau_1}$, we again use the ILED estimate \eqref{eq:sup:Morawetz} for to derive 
    \begin{align*}
        \int_{\R^{\tau_2}_{\tau_1}\setminus\D^{\tau_2}_{\tau_1}} &t^{2-2q-\varepsilon}\left(|\d_t(t^q\psi)|^2+|\d_r\psi|^2+  \frac{1}{r^2}|\rslnabla \psi |^2+\frac{1}{1+r^2}\psi^2+ t^{2q-2} \psi^2\right)\,r^2dudvd\omega\lesssim \E_q^\varepsilon[\psi](\tau_1)\lesssim \E_q^0[\psi](\tau_1)\,.
    \end{align*}
    So if we consider $f^\varepsilon(\tau)=f^\varepsilon_q[\psi](\tau)$ and $h_{1-\epsilon}(\tau)=h^0_{q,1-\epsilon}[\psi](\tau)$ as defined in \eqref{eq:decay:def-f} and \eqref{eq:decay:def-hi}, we get
    \begin{align*}
        f^\varepsilon(\tau_2)&\lesssim f^\varepsilon(\tau_1)\,,\\
        \int^{\tau_2}_{\tau_1} f^\varepsilon(\tau)\,d\tau + h_{1-\epsilon}(\tau_2)&\lesssim f^0(\tau_1) + h_{1-\epsilon} (\tau_1)\,, 
    \end{align*}
    for any $\tau_2>\tau_1\geq \tau_0$.
    Therefore, similar to the first part of the proof of Lemma \ref{lem:general:hierarchy-estimate}, the mean value theorem gives
    \begin{equation*}
        f^\varepsilon(\tau)\lesssim \frac{1}{\tau} \E_q[\psi_0,\psi_1]\,,
    \end{equation*}
    for any $\tau>\tau_0$ and with $\E_q[\psi_0,\psi_1]$ as defined in \eqref{eq:decay:initial-energy}.
\end{proof}
\begin{rmk}\label{rm:decay:hi}
    It is worth mentioning that for all $0<q<1$ and $q\neq\frac{1}{3}$, Lemma \ref{lem:general:hierarchy-estimate} and $r^p$-estimates along with \eqref{eq:decay:h0-hEps} imply that for $\varepsilon$ and $\sigma$ satisfying \eqref{eq:decay:parameters-q}, we also have
    \begin{align*}
        \begin{split}h^\varepsilon_0(\tau)\lesssim\frac{\E_q[\psi_0,\psi_1]}{\tau^{2-\sigma}}\,,\qquad h^\varepsilon_1(\tau)\lesssim\frac{\E_q[\psi_0,\psi_1]}{\tau^{1-\sigma}}\,,\qquad h^\varepsilon_{2-\sigma}(\tau) \lesssim \E_q[\psi_0,\psi_1]\,.
        \end{split}
    \end{align*}
\end{rmk}
\begin{rmk}\label{rmk:decay:1/3}
    When $\frac{1}{3}$, the $r^p$-estimate \eqref{eq:sub:rpEstimate} remains valid for $0<p<1$, which is the same range as in case $q\geq \frac{1+2\sqrt 2}{7}$. Therefore, similar results can be established by considering an energy flux including an arbitrarily small $\varepsilon>0$ loss in the power of $t$-weights. After establishing the energy boundedness and ILED estimate for the modified energy flux, we can use \eqref{eq:decay:h0-hEps} to obtain the energy decay.
\end{rmk}
\subsection{Pointwise decay}
The next statement concerns the pointwise decay of the solutions. To begin with, define
\begin{align}\label{eq:decay:h-initial-energy}
    \bar{\E}_q[\psi_0,\psi_1]=\E_q[\psi_0,\psi_1]+\sum_{j=0,1,2}\sum_{i=1,2,3}\E_q[\Omega_i^j\psi|_{\Sigma_{\tau_0}},\d_t\Omega_i^j\psi|_{\S_{\tau_0}}]+ \sum_{i=1,2,3} \E_q[\d_{x^i}\psi|_{\Sigma_{\tau_0}}, \d_t\d_{x^i}\psi|_{\S_{\tau_0}}]\,.
\end{align}
\begin{prop}[Pointwise decay of waves on FLRW ]\label{prop:decay:PointwiseDecay} For $q\neq \frac{1}{3}$, let $\varepsilon$, $\sigma$, and $\mu_q^\varepsilon$ be as defined in \eqref{eq:decay:parameters-q}. If $\psi$ solves the initial value problem \eqref{eq:intro:IVP-wave}, then
    \begin{align}\label{eq:decay:point-wise1}
        |\psi|\lesssim \frac{t^\frac{{\mu_q^\varepsilon}}{2} \tau^{\frac{\sigma}{2}}}{\tau t^q  \sqrt{1+r} }\bar\E_q[\psi_0,\psi_1]^\frac{1}{2}\,,
    \end{align}
    where $\bar \E_q[\psi_0,\psi_1]$ is defined as \eqref{eq:decay:h-initial-energy}. Moreover,
    \begin{equation}\label{eq:decay:point-wise2}
        \begin{aligned}
            |\psi|&\lesssim \frac{t^\frac{{\mu_q^\varepsilon}}{2} \tau^{\frac{\sigma}{2}}}{(1+r)  t^q  \sqrt \tau}\bar\E_q[\psi_0,\psi_1]^\frac{1}{2}\,,\qquad &\mathrm{ if }\quad &0<q<\frac{1+2\sqrt 2}{7}\,,\\
            |\psi|&\lesssim \frac{t^\frac{{\mu_q^\varepsilon}}{2} \log (1+r)}{(1+r)  t^q  }\bar\E_q[\psi_0,\psi_1]^\frac{1}{2}\,,\qquad &\mathrm{ if }\quad &\frac{1+2\sqrt 2}{7}\leq q <1\,.
        \end{aligned}
    \end{equation}
    In particular,
    \begin{align}\label{eq:decay:particular-t}
        |\psi|\lesssim \frac{1}{t^{1-(1-q)\frac{\sigma}{2}-\frac{{\mu_q^\varepsilon}}{2}}}\bar\E_q[\psi_0,\psi_1]^\frac{1}{2}\,.
    \end{align}
\end{prop}
\begin{proof}
    First, we prove \eqref{eq:decay:point-wise1} and \eqref{eq:decay:point-wise2} with $r^{-1}$ instead of $(1+r)^{-1}$, which are equivalent for large $r$. Later we show the decay rate for small $r$ to avoid the degeneracy around $r=0$.
    
    In order to prove the \eqref{eq:decay:point-wise1}, consider any point $(t_\circ,r_\circ,\omega_\circ)$ in $\M$ with $r_\circ \geq \rho$. We then integrate 
    $$\d_v(t^{q-\frac{{\mu_q^\varepsilon}}{2}}\psi)=t^{-\frac{{\mu_q^\varepsilon}}{2}}\d_v(t^{q}\psi)-\frac{{\mu_q^\varepsilon}}{2} t^{2q-\frac{{\mu_q^\varepsilon}}{2}-1}\psi\,,$$ 
    with respect to $v$ from $v(r_\circ)$ to $v^*$ for fixed $u$ to get
    \begin{align}\label{eq:radiation:FirstPointwise1}
    \begin{split}
        | t_\circ^{q-\frac{{\mu_q^\varepsilon}}{2}} \psi(t_\circ,r_\circ,\omega_\circ)|^2&\leq  \Big| t^{q-\frac{{\mu_q^\varepsilon}}{2}} \psi\vert_{v=v^*}\Big|^2+\left(\int^{v^*}_{v(r_\circ)} (t^{-\frac{{\mu_q^\varepsilon}}{2}}\d_v (t^q \psi) + \frac{{\mu_q^\varepsilon}}{2} t^{-\frac{{\mu_q^\varepsilon}}{2}+2q-1}\psi )\; dv\right)^2 \\&\lesssim \Big| t^{q-\frac{{\mu_q^\varepsilon}}{2}} \psi\vert_{v=v^*}\Big|^2+ \frac{1}{r_\circ}\left(\int^{v^*}_{v(r_\circ)} t^{-{\mu_q^\varepsilon}}|\d_v(t^q\psi)|^2\;r^2 dv + \int^{v^*}_{v(r_\circ)} {\mu_q^\varepsilon}^2t^{-\varepsilon} \psi^2\; r^2dv\right)\,,
    \end{split}
    \end{align}
    remembering that ${\mu_q^\varepsilon}=0$ when $q\leq\frac{1}{2}$. If we now integrate over $\SS^2$ and let $v^*\to \infty$, the term on future null infinity vanishes because the energy boundedness estimates \eqref{eq:sub:Eng-boundedness} and \eqref{eq:sup:Eng1-boundedness} assert that there exists a sequence $\{v_n\}_n$ for each $u$ such that 
    \begin{align*}
        \lim_{n\to \infty} \int_{\SS^2\times\{v=v_n\}} t^{2q-\mu_q^\varepsilon}\psi^2 \;d\omega =0\,.
    \end{align*}
    Thus,    \begin{align}\label{eq:radiation:sphereIneq}
    \begin{split}
        \int_{\SS^2} | \psi(t_\circ,r_\circ,\omega)|^2 \;d\omega &\lesssim \frac{t_\circ^{\mu_q^\varepsilon}}{t^{2q}_\circ}\int_{\SS^2\times\{v=\infty\}} t^{2q-\mu_q^\varepsilon}|\psi|^2 \;d\omega+   \frac{t_\circ^{\mu_q^\varepsilon}}{r_\circ t^{2q}_\circ}\int_{\C_{\tau_\circ}} \left(t^{-{\mu_q^\varepsilon}}|\d_v(t^q\psi)|^2+ {\mu_q^\varepsilon}^2 t^{-\varepsilon} \psi^2\right)\;r^2 dvd\omega\\&\lesssim \frac{t_\circ^{\mu_q^\varepsilon}}{r_\circ t^{2q}_\circ}  \E_q^\varepsilon[\psi](\tau_\circ)\lesssim\frac{t_\circ^{\mu_q^\varepsilon}\E_q[\psi_0,\psi_1]}{r_\circ  t^{2q}_\circ \tau^{2-\sigma}_\circ}\,,
    \end{split}
    \end{align}
    thanks to Proposition \ref{Prop:decay-energy}. 
    
    Now recall that the wave equation commutes with angular operators $\Omega_i$, defined in Section \ref{sec:angular-derivatives}, for $i=1, 2, 3$. Hence, the above is also true for $\Omega_i^j \psi$ with $k=1,2,3$ and with suitable initial data sets. The Sobolev inequality in Proposition  \ref{prop:general:Sobolev} then gives the pointwise decay
    \begin{align}\label{eq:decay:sobolev}
        \|\psi(t_\circ, r_\circ, \cdot)\|_{L^\infty(\SS^2)}\lesssim \|\psi(t_\circ, r_\circ, \cdot)\|_{H^2(\SS^2)}\stackrel{\eqref{eq:general:angularderivatives-sphere}}{\lesssim}\sum_{j=0,1,2}\sum_{i=1,2,3}\|\Omega_i^j\psi\|_{L^2(\SS^2)}\lesssim \frac{t_\circ^{\frac{\mu_q^\varepsilon}{2}}\bar\E_q[\psi_0,\psi_1]^\frac{1}{2}}{\sqrt r_\circ t_\circ^q \tau_\circ^{1-\frac{\sigma}{2}}}\,.
    \end{align}
    
    For the other case, $r_\circ \leq \rho$, we integrate with respect to $r$ from $r_\circ$ to $r=\rho$ where the decay is known. In fact,  
    \begin{align*}
        | t_\circ^{q-\frac{{\mu_q^\varepsilon}}{2}} \psi(t_\circ,r_\circ,\omega_\circ)|^2&\lesssim \left(\int^\rho_{r_\circ} (t_\circ^{-\frac{{\mu_q^\varepsilon}}{2}}\d_v (t_\circ^q \psi) + \frac{{\mu_q^\varepsilon}}{2} t_\circ^{-\frac{{\mu_q^\varepsilon}}{2}+2q-1}\psi) \; dr\right)^2+| t_\circ^{q-\frac{{\mu_q^\varepsilon}}{2}} \psi(t_\circ,\rho,\omega_\circ)|^2 \\&\lesssim \frac{1}{r_\circ}\int^\rho_{r_\circ}( t_\circ^{2q-{\mu_q^\varepsilon}}|\d_t(t_\circ^q\psi)|^2+ t_\circ^{2q-{\mu_q^\varepsilon}}|\d_r\psi|^2+{\mu_q^\varepsilon}^2 t_\circ^{-\varepsilon} \psi^2)\;r^2 dr ,\\& + \frac{1}{\rho}\left(\int^\infty_{v(\rho)}( t^{-{\mu_q^\varepsilon}}|\d_v(t^q\psi)|^2\;r^2 dv + {\mu_q^\varepsilon}^2 t^{-\varepsilon} \psi^2)\; r^2dv\right)\,,
    \end{align*}
    in which we have used \eqref{eq:radiation:FirstPointwise1} to bound the term at $r=\rho$. Observe that the $\SS^2$-integral of the right-hand side above is bounded by $\frac{1}{r_\circ}  \E_q^\varepsilon[\psi](\tau_\circ)$. Commuting with angular operators and using the Sobolev inequality then again leads to the desired decay.

    To prove \eqref{eq:decay:point-wise2}, we can assume $r_\circ>\rho$ because otherwise the result simply follows by \eqref{eq:decay:point-wise1}. In addition, suppose $r_\circ\geq \frac{1}{4}u_\circ$ as the other case is trivial as well. Let $(t_a,r_a,\omega_a)$ denote the intersection of $\Sigma_{\tau_\circ}$ and the timelike hypersurface $\{r=\frac{1}{4}u\}$. We consider $\{r=\frac{1}{4}u\}$ because the pointwise inequalities are equivalent at this hypersurface. 
    Now we integrate $\d_v(\varphi)$  with respect to $v$ from $v(r_a)$ to $v(r_\circ)$ to derive
    \begin{align*}
        |\varphi(t_\circ,r_\circ, \omega_\circ)|^2 \lesssim& |\varphi(t_a,r_a, \omega_a)|^2 + \left(\int_{v(r_a)}^{v(r_\circ)} \d_v\varphi \, dv\right)^2 \\
        \lesssim& t_a^{2q}u^2_a |\psi(t_a,r_a, \omega_a)|^2 + t_\circ^{\mu_q^\varepsilon} \left(\int_{v(r_a)}^{v(r_\circ)} \frac{1}{r^{2-\sigma}}\, dv\right)\left(\int_{v(r_a)}^{v(r_\circ)} t^{-{\mu_q^\varepsilon}}r^{2-\sigma}|\d_v \varphi|^2 \, dv\right)\,.
    \end{align*}
    If $\sigma<1$, then the integral of $\frac{1}{r^{2-\sigma}}$ is bounded by $\frac{1}{r_a^{1-\sigma}}$. And if $\sigma=1$, it is bounded by $\log (1+r_\circ)$. For $\sigma<1$, if we integrate over $\SS^2$ and use Remark \ref{rm:decay:hi} and estimate \eqref{eq:radiation:sphereIneq} at point $(t_a,r_a, \omega_a)$, we find
    \begin{align*}
        \int_{\SS^2}  |\varphi(t_\circ,r_\circ, \omega)|^2\, d\omega&\lesssim t_\circ^{{\mu_q^\varepsilon}} t_a^{-{\mu_q^\varepsilon}}u_a^2 \frac{\E_q[\psi_0,\psi_1]}{r_a \tau^{2-\sigma}_a} + \frac{t_\circ^{{\mu_q^\varepsilon}}}{r_a^{1-\sigma}} h^\varepsilon_{2-\sigma}(\tau_\circ)\lesssim \frac{t_\circ^{\mu_q^\varepsilon}}{\tau_\circ^{1-\sigma}}\E_q[\psi_0,\psi_1]\,,
    \end{align*}
    thanks to the facts that $\tau_\circ=\tau_a\sim r_a$ and $u_a<\tau_a$. Once again, commuting with angular derivatives and Sobolev's inequality finally gives
    $$|r_\circ  t_\circ^q \psi(t_\circ,r_\circ, \omega_\circ)|\lesssim \frac{t_\circ^\frac{{\mu_q^\varepsilon}}{2} }{\sqrt {\tau_\circ^{1-\sigma}}}\bar \E_q[\psi_0,\psi_1]^\frac{1}{2}\,,$$
    which proves the required result for points with $r_\circ$ away from zero. The other case, $\sigma=1$, follows similarly.

    Now to remove the degeneracy at $r=0$, set $0<s<\frac{\rho}{2}$ and assume $r_\circ<s$. The local Sobolev inequality in Proposition \ref{prop:general:Sobolev} implies 
    \begin{align}\label{eq:decay:small-r}
        \|\psi(t_\circ,\cdot)\|^2_{L^\infty(B_s)}&\lesssim s^{-3}\int_{B_{2s}\times{\{t_\circ\}}} (|\psi|^2+s^2|\nabla\psi|^2+s^4|\nabla^2 \psi|^2)\,dx\\
        &\lesssim_s t^{{\mu_q^\varepsilon}-2q}_\circ \int_{\S_{\tau_\circ}} t_\circ^{2q-{\mu_q^\varepsilon}} (\frac{1}{(1+r_\circ)^2}|\psi|^2+|\nabla\psi|^2+|\nabla^2 \psi|^2)\,dx
        \\&\lesssim \frac{t^{{\mu_q^\varepsilon}-2q}_\circ}{\tau_\circ^{2-\sigma}} \left(\E_q [\psi_0,\psi_1]+\sum_{i=1,2,3} \E_q[\d_{x^i}\psi|_{\Sigma_{\tau_0}}, \d_t\d_{x^i}\psi|_{\Sigma_{\tau_0}}]\right)\,,
    \end{align}
    in which we have used the fact that the wave operator on FLRW commutes with $\d_{x^i}$ for $i=1,2,3$, and the decay estimates for twisted energy currents are also true for $\d_{x^i}\psi$ with appropriate initial data sets. Therefore, we have
    \begin{align*}
         \|\psi(t_\circ,\cdot)\|_{L^\infty(B_s)}&\lesssim \frac{t^\frac{{\mu_q^\varepsilon}}{2}_\circ\tau_\circ^\frac{\sigma}{2}}{t_\circ^q\tau_\circ}\bar\E_q[\psi_0,\psi_1]^\frac{1}{2}\lesssim \frac{t^\frac{{\mu_q^\varepsilon}}{2}_\circ\tau_\circ^\frac{\sigma}{2}}{(1+r)t_\circ^q\tau_\circ}\bar\E_q[\psi_0,\psi_1]^\frac{1}{2}\,,
    \end{align*}
    for $r\in B_s$ which completes the proof of \eqref{eq:decay:point-wise1} and \eqref{eq:decay:point-wise2}.

    Finally, to establish \eqref{eq:decay:particular-t}, first note that always $\tau\lesssim t^{1-q}$. When $r_\circ\leq\rho$, by definition $\tau_\circ\sim t_\circ^{1-q}$, and estimate \eqref{eq:decay:particular-t} easily follows from \eqref{eq:decay:point-wise1}. In the other case, we have $2\tau=\frac{t^{1-q}}{1-q}-r+\rho$. So if $r_\circ-\rho <2\tau_\circ$, then $\tau_\circ^{-1}\lesssim t_\circ ^{q-1}$, and the result is achieved by \eqref{eq:decay:point-wise1}. Otherwise, if $r_\circ-\rho\geq2\tau_\circ$, we then have $r^{-1}_\circ \lesssim t_\circ^{q-1}$, and \eqref{eq:decay:point-wise2} leads to \eqref{eq:decay:particular-t} (we also use \eqref{eq:general:t-r-lemma} when $q\geq \frac{1+2\sqrt 2}{7}$).
\end{proof}
\begin{rmk}[Pointwise decay in general case]\label{rmk:decay:pointwise-general}
    In general, for $0<q<1$, we can also prove a pointwise decay estimate only using the boundedness of the energy flux $\E^\varepsilon_q [\psi](\tau)$; see Remark \ref{rmk:intro:decay-from-EB}. In fact, we can follow the approach in the proof of Proposition \ref{prop:decay:PointwiseDecay} up to the estimate \eqref{eq:radiation:sphereIneq}, and then instead of using Proposition \ref{Prop:decay-energy}, we use just the boundedness of the energy. In fact, we have
    \begin{align*}
        \|\psi(t, r, \cdot)\|_{L^\infty(\SS^2)}\lesssim \frac{t^{\frac{\mu_q}{2}}\bar\E_q[\psi_0,\psi_1]^\frac{1}{2}}{\sqrt {1+r} t^q }\,,
    \end{align*}
    with $\mu_q=\mu^0_q$.
\end{rmk}
\subsection{Decay of first order derivatives}
In this section, we present the pointwise decay estimates for the first order derivatives of the solution $\psi$. The main result of this section is Proposition \ref{prop:decay:derivatives}, which proves these estimates in the general case.  
Nevertheless, in view of Remark \ref{rmk:radiation:commuting}, commuting with $\d_t (\sqrt t\, \cdot\, )$ in the radiation case yields an improved decay rate for the time derivative of $\psi$. This approach is detailed in Proposition \ref{prop:decay:commuting-radiation}. Note that Propositions \ref{prop:decay:commuting-radiation} and \ref{prop:decay:derivatives} together prove the last part of Corollary \ref{cor:general:decay}.

For the initial value problem \eqref{eq:intro:IVP-wave} with $q=\frac{1}{2}$, define
\begin{equation}\label{eq:decay:hh-initial-energy-radiation}
    \bar {\bar {\E}}[\psi_0,\psi_1]= \bar {\E}_\frac{1}{2}[\d_t (\sqrt t \psi)|_{\Sigma_{\tau_0}},\d_{tt} (\sqrt t \psi)|_{\S_{\tau_0}}]\,.
\end{equation}
\begin{prop}[Decay for the time derivative of waves with $q=\frac{1}{2}$]\label{prop:decay:commuting-radiation}
    Suppose $\psi$ solves the wave equation for $q=\frac{1}{2}$. Then the time derivative $\d_t\psi$ decays as
    \begin{equation*}
        |\d_t\psi|\lesssim \bar {\bar {\E}}[\psi_0,\psi_1]^\frac{1}{2} \frac{1}{t(1+r)\sqrt \tau}\,,\qquad|\d_t\psi|\lesssim \bar {\bar {\E}}[\psi_0,\psi_1]^\frac{1}{2} \frac{1}{t\tau \sqrt{1+r}}\,,\qquad|\d_t\psi|\lesssim \bar {\bar {\E}}[\psi_0,\psi_1]^\frac{1}{2} t^{-\frac{3}{2}}\,,
    \end{equation*}
    in which $\bar {\bar {\E}}[\psi_0,\psi_1]$ is defined in \eqref{eq:decay:hh-initial-energy-radiation}.
\end{prop}
\begin{proof}
    Commute the wave equation with $\d_t(\sqrt t\,\cdot\,)$ and recall that Remark \eqref{rmk:radiation:commuting} indicates $\d_t(\sqrt t \psi)$ satisfies the wave equation as well. Therefore, Proposition \ref{prop:decay:PointwiseDecay} also holds for $\d_t (\sqrt t \psi)$. Indeed we have 
    \begin{align*}
        |\sqrt t \d_t \psi| \lesssim |\d_t(\sqrt t \psi)| + | \frac{1}{\sqrt t} \psi| \lesssim \frac{\bar {\bar {\E}}[\psi_0,\psi_1]^\frac{1}{2}}{t}+ \frac{\bar \E_q[\psi_0,\psi_1]^\frac{1}{2}}{t\sqrt t } \lesssim \frac{\bar {\bar {\E}}_q[\psi_0,\psi_1]^\frac{1}{2}}{t}\,.
    \end{align*}
    The other two estimates also easily follow from Proposition \ref{prop:decay:PointwiseDecay} as well.
\end{proof}
To state the decay estimates for the first derivatives of $\psi$ in the general case, first define the initial energy as
    \begin{equation}\label{eq:decay:hh-initial-energy}
        \begin{aligned}
            \bar{\bar{\E}}_q[\psi_0,\psi_1]:={\bar{\E}}_q[\psi_0,\psi_1]+&\sum_{i,j\in\{1,2,3\}}(\bar{\E}_q[\d_{x^j}\d_{x^i}\psi|_{\Sigma_{\tau_0}},\d_t\d_{x^j}\d_{x^i}\psi|_{\S_{\tau_0}}]+\bar{\E}_q[\Omega_j\Omega_i\psi|_{\Sigma_{\tau_0}},\d_t\Omega_j\Omega_i\psi|_{\S_{\tau_0}}])\\
    +&\sum_{i=1,2}\E_q[\d_v^i\psi|_{\Sigma_{\tau_0}},\d_t\d_v^i\psi|_{\S_{\tau_0}}]\,.
        \end{aligned}
    \end{equation}
\begin{prop}[Decay estimates for the derivatives of waves on FLRW]\label{prop:decay:derivatives} For $q\neq \frac{1}{3}$, let $\varepsilon$, $\sigma$, and $\mu_q^\varepsilon$ be as defined in \eqref{eq:decay:parameters-q}. Suppose $\psi$ is a solution to the covariant wave equation \eqref{eq:intro:IVP-wave}. First derivatives of $\psi$ satisfy
    \begin{align*}
        |\d_v \psi| \lesssim 
            \bar{\bar{\E}}_q[\psi_0,\psi_1]^\frac{1}{2} \frac{ t^\frac{{\mu_q^\varepsilon}}{2}  \tau^{\frac{\sigma}{2}}}{ t^q  (1+r)^2} 
        \,, \qquad |\d_u \psi| \lesssim \bar{\bar{\E}}_q[\psi_0,\psi_1]^\frac{1}{2} \frac{ t^\frac{{\mu_q^\varepsilon}}{2}  \tau^{\frac{\sigma}{2}}}{ t^q  (1+r)}\,,\qquad |\d_t \psi| \lesssim \bar{\bar{\E}}_q[\psi_0,\psi_1]^\frac{1}{2} \frac{ t^\frac{{\mu_q^\varepsilon}}{2}  \tau^{\frac{\sigma}{2}}}{ t^{2q}  (1+r)}\,,
    \end{align*}
    in which $\bar {\bar\E}_q[\psi_0,\psi_1]$ is a high order weighted initial energy of $\psi$ containing derivatives up to order five, defined in \eqref{eq:decay:hh-initial-energy}. In addition, for the case $q=\frac{1}{2}$ with $\sigma={\mu_q^\varepsilon}=0$ we have
    \begin{align*}
       |\d_v \psi| \lesssim \bar{\bar{\E}}_\frac{1}{2}[\psi_0,\psi_1]^\frac{1}{2} \frac{\log \tau}{ \sqrt t  (1+r)^2}\,,\qquad |\d_u \psi| \lesssim \bar{\bar{\E}}_\frac{1}{2}[\psi_0,\psi_1]^\frac{1}{2} \frac{\log \tau}{ \sqrt t  (1+r)}\,,\qquad|\d_t \psi| \lesssim \bar{\bar{\E}}_\frac{1}{2}[\psi_0,\psi_1]^\frac{1}{2} \frac{\log \tau}{ t  (1+r)}\,.
    \end{align*}
\end{prop}\begin{proof}
    First, recall that the covariant wave operator commutes with angular derivatives $\Omega_i$, defined in \ref{sec:angular-derivatives}, for $i=1,2,3$. Therefore, Proposition \ref{prop:decay:PointwiseDecay} implies
    \begin{align}\label{eq:decay:second-angular-derivaives}
        |\Omega_j\Omega_i \psi|\lesssim \frac{t^\frac{{\mu_q^\varepsilon}}{2} \tau^{\frac{\sigma}{2}}}{\tau t^q  \sqrt{1+r} }\bar{\bar{\E}}_q[\psi_0,\psi_1]^\frac{1}{2}\,,
    \end{align}
    for $i,j\in\{1,2,3\}$.
    Now from the wave equation in double null coordinates
    \begin{align*}
        \frac{1}{t^{3q}r}\d_{uv}(t^qr\psi)=\frac{1}{t^{2q}r^2}\rslDelta \psi + \frac{q(2q-1)}{t^2}\psi\,,
    \end{align*}
    and Lemma \ref{lem:general:angular-derivatives}, we conclude that    \begin{align}\label{eq:decay:uv-derivative}
    \begin{split}
        |\d_{uv}(t^qr\psi)|&\lesssim \frac{t^q r}{r^2}|\rslDelta \psi| +t^{3q-2}r |\psi| \lesssim \frac{t^q r}{(1+r)^2} \sum_{i,j=1,2,3} (|\d_{x^i}\d_{x^j}\psi|+|\Omega_j\Omega_i \psi|) + t^{3q-2}r |\psi|\\
        & \lesssim \frac{t^q r}{(1+r)^2} \frac{t^\frac{{\mu_q^\varepsilon}}{2} \tau^{\frac{\sigma}{2}}}{\tau t^q  \sqrt{1+r} }\bar{\bar{\E}}_q[\psi_0,\psi_1]^\frac{1}{2} +  t^{3q-2}r \frac{t^\frac{{\mu_q^\varepsilon}}{2} \tau^{\frac{\sigma}{2}}}{\tau t^q  \sqrt{1+r} }\bar{\E}_q[\psi_0,\psi_1]^\frac{1}{2}\lesssim \frac{t^\frac{{\mu_q^\varepsilon}}{2} \tau^{\frac{\sigma}{2}}}{\tau   (1+r)^\frac{3}{2}} \bar{\bar{\E}}_q[\psi_0,\psi_1]^\frac{1}{2}\,,
    \end{split}
    \end{align}
    thanks to \eqref{eq:decay:second-angular-derivaives}, first estimate of Proposition \ref{prop:decay:PointwiseDecay} and Lemma \ref{lem:general:t-r}. Since the initial data set $(\psi_0,\psi_1)$ is compactly supported on $\Sigma_{\tau_0}$, there exists $v_a$ with $r(v_a,u_0)>\rho$ such that for $\psi|_{\Sigma_{\tau_0}\cap \{v>v_a\}}=0$. So consider the point $(\mathring v,\mathring u,\mathring\omega)$ with $\mathring v>v_a$ and integrate \eqref{eq:decay:uv-derivative} with respect to $u$ from the initial data hypersurface $\Sigma_{\tau_0}$ to $\mathring u$ to find
    \begin{align}
    \begin{split}
        |\d_v (\mathring t^q \mathring r \psi(\mathring v,\mathring u, \mathring \omega))| &\lesssim \bar{\bar{\E}}_q[\psi_0,\psi_1]^\frac{1}{2} |\int_{u_0}^{\mathring u} \frac{t^\frac{{\mu_q^\varepsilon}}{2} \tau^{\frac{\sigma}{2}}}{\tau   (1+r)^\frac{3}{2}} \,du |\lesssim \bar{\bar{\E}}_q[\psi_0,\psi_1] ^\frac{1}{2}\frac{\mathring t ^\frac{{\mu_q^\varepsilon}}{2}}{(1+\mathring r)^\frac{3}{2}}\int_{u_0}^{\mathring u} \tau^{\frac{\sigma}{2}-1}\,du\\ & \lesssim \begin{cases}
            \bar{\bar{\E}}_q[\psi_0,\psi_1]^\frac{1}{2} (1+\mathring r)^\frac{-3}{2} \mathring t ^\frac{{\mu_q^\varepsilon}}{2} \mathring \tau ^{\frac{\sigma}{2}}\,\quad \sigma>0\,,\\
            \bar{\bar{\E}}_q[\psi_0,\psi_1]^\frac{1}{2}(1+\mathring r)^\frac{-3}{2}  \log \mathring \tau \,\quad \sigma={\mu_q^\varepsilon}=0\,.
        \end{cases}
    \end{split}\label{eq:decay:v-derivative-1}
        \end{align} 
    Note that the last inequality holds whether $r>\rho$ or not, and the case $\sigma=0$ corresponds to $q=\frac{1}{2}$. For $(\mathring v,\mathring u,\mathring\omega)$ with $\mathring v\leq v_a$, where the corresponding point on the initial hypersurface is in the support of the initial data, we have
    \begin{align*}
        |\d_v (\mathring t^q \mathring r \psi(\mathring v,\mathring u, \mathring \omega))| &\lesssim \bar{\bar{\E}}_q[\psi_0,\psi_1]^\frac{1}{2} \int_{u_0}^{\mathring u} \frac{t^\frac{{\mu_q^\varepsilon}}{2} \tau^{\frac{\sigma}{2}}}{\tau   (1+r)^\frac{3}{2}} \,du +\Big|\d_v ( t^q r \psi)\vert_{(\mathring v, u_0, \mathring \omega)}\Big|\,. 
        \end{align*} 
        The first term on the right-hand side can be estimated as in \eqref{eq:decay:v-derivative-1}, while the second term is bounded by  $\bar{\bar{\E}}_q[\psi_0,\psi_1]$ using the local Sobolev inequalities \eqref{eq:general:Sobolev:Local} on $S_{\tau_0}$ and $\C_{\tau_0}$ depending on whether $r(\mathring v, u_0)\leq \rho$ or $r(\mathring v, u_0)>\rho$ (recall that $\sum_{i=1,2} \E_q[\d_v^i\psi|_{\Sigma_{\tau_0}},\d_t\d_v^i\psi|_{\S_{\tau_0}}]\lesssim \bar{\bar{\E}}_q[\psi_0,\psi_1] $)\footnote{Note that at the corner point $p\in \S_{\tau_0}\cap \C_{\tau_0}$ with $r(p)=\rho$, we can integrate from $p$ to future null infinity with respect to $v$, and then use the Sobolev inequality on Sphere after integration over $\SS^2$ and commuting with angular derivatives $\Omega_i^j$ for $j=0,1,2$.}. Since $\mathring v$ and $r(\mathring v, u_0)$ are bounded in this case, we can recover the same estimate in \eqref{eq:decay:v-derivative-1}. Therefore, estimate \eqref{eq:decay:v-derivative-1} holds for all points in $\M$.
   
    In order to estimate $\d_u \psi$ at the point $(\mathring v, \mathring u, \mathring \omega)$, we integrate \eqref{eq:decay:uv-derivative} in $v$ from $( \hat v,\mathring u, \mathring \omega)$ to $(\mathring v, \mathring u, \mathring \omega)$ where $( \hat v,\mathring u, \mathring \omega)$ is either lying on the spacelike part of the initial data hypersurface $\S_{\tau_0}$ or on the centre $\{r=0\}$ depending on $\mathring u$. Indeed,
    \begin{align}\label{eq:decay:derivatives:u-twisted}
        |\d_u (\mathring t^q \mathring r \psi(\mathring v,\mathring u, \mathring \omega))| &\lesssim \bar{\bar{\E}}_q[\psi_0,\psi_1]^\frac{1}{2} \int_{\hat v}^{\mathring v} \frac{t^\frac{{\mu_q^\varepsilon}}{2} \tau^{\frac{\sigma}{2}}}{\tau   (1+r)^\frac{3}{2}} \,dv  + \Big|\d_u ( t^q r \psi)\vert_{(\hat v,\mathring u, \mathring \omega)}\Big|\,.
    \end{align}
    The first term above can be handled as follows with $\mathring v_\rho$ denoting the $v-$coordinate of the intersection of $\{r=\rho\}$ and $\{u=\mathring u\}$,
    \begin{align*}
        \int_{\hat v}^{\mathring v} \frac{t^\frac{{\mu_q^\varepsilon}}{2} \tau^{\frac{\sigma}{2}}}{\tau   (1+r)^\frac{3}{2}} \,dv&\lesssim \mathring \tau^{-1+\frac{\sigma}{2}}\int_{\mathring v_\rho}^{\mathring v} \frac{t^\frac{{\mu_q^\varepsilon}}{2} }{ (1+r)^\frac{3}{2}} \,dv+\int_{\hat v}^{\mathring v_\rho} \frac{t^{\frac{{\mu_q^\varepsilon}}{2}+(1-q)(-1+\frac{\sigma}{2})}}{(1+r)^\frac{3}{2}} \,dv\lesssim  \frac{\mathring t ^\frac{{\mu_q^\varepsilon}}{2}\mathring\tau^{\frac{\sigma}{2}}}{\mathring \tau \sqrt \rho} + 
        \begin{cases}
            \mathring t^\frac{{\mu_q^\varepsilon}}{2} \mathring {\tau}^\frac{\sigma}{2}\quad &\sigma>0\\
            \log \mathring \tau\quad &\sigma={\mu_q^\varepsilon}=0
        \end{cases}
        \\&\lesssim\begin{cases}
            \mathring t^\frac{{\mu_q^\varepsilon}}{2} \mathring {\tau}^\frac{\sigma}{2}\quad &\sigma>0\\
            \log \mathring \tau\quad &\sigma={\mu_q^\varepsilon}=0
        \end{cases}
        \,.
    \end{align*}
    For the second term on the right-hand side of \eqref{eq:decay:derivatives:u-twisted}, note that if $( \hat v,\mathring u, \mathring \omega)\in \{r=0\}$ where $\tau=\frac{1}{2}\frac{t^{1-q}}{1-q}$, we have
    \begin{align*}
        \Big|\d_u ( t^q r \psi)|_{(\hat v,\mathring u, \mathring \omega)}\Big|=\Big|t^q \psi\vert_{(\hat v,\mathring u, \mathring \omega)}\Big| \lesssim  ( \frac{t^\frac{{\mu_q^\varepsilon}}{2} \tau^{\frac{\sigma}{2}}}{\tau })|_{(\hat v,\mathring u, \mathring \omega)}\bar\E_q[\psi_0,\psi_1]^\frac{1}{2}\lesssim \mathring t^\frac{{\mu_q^\varepsilon}}{2} \mathring \tau^{\frac{\sigma}{2}}\bar{\bar\E}_q[\psi_0,\psi_1]^\frac{1}{2}\,,
    \end{align*}
    and if $( \hat v,\mathring u, \mathring \omega)\in \S_{\tau_0}$, then the above is trivially correct,  again thanks to the local Sobolev inequality \eqref{eq:general:Sobolev:Local}.
    Thus,
    \begin{align*}
        |\d_t (\mathring t^q \mathring r \psi(\mathring v,\mathring u, \mathring \omega))|\lesssim \frac{1}{\mathring t^q}\left(|\d_v (\mathring t^q \mathring r \psi(\mathring v,\mathring u, \mathring \omega))|+|\d_u (\mathring t^q \mathring r \psi(\mathring v,\mathring u, \mathring \omega))|\right) &\lesssim\begin{cases}
            \frac{\mathring t^\frac{{\mu_q^\varepsilon}}{2} \mathring {\tau}^\frac{\sigma}{2}}{\mathring t ^q}\bar{\bar\E}_q[\psi_0,\psi_1]^\frac{1}{2}\quad &\sigma>0\\
            \frac{\log \mathring \tau}{\mathring t ^q}\bar{\bar\E}_q[\psi_0,\psi_1]^\frac{1}{2}\quad &\sigma={\mu_q^\varepsilon}=0
        \end{cases}
        \,.
    \end{align*}
    To estimate untwisted derivatives, first assume $0<\epsilon<\mathring r$. We therefore find that for $\sigma>0$
    \begin{align*}
        |\d_u \psi(\mathring v,\mathring u, \mathring \omega)|&\lesssim \frac{1}{\mathring t^q \mathring r}\left(|\d_u (\mathring t^q \mathring r \psi)| + \mathring t^{2q-1}\mathring r |\psi| + \mathring t^q |\psi|\right)\lesssim \frac{\mathring t^\frac{{\mu_q^\varepsilon}}{2} \mathring \tau^{\frac{\sigma}{2}}}{\mathring t^q \mathring r}\bar{\bar{\E}}_q[\psi_0,\psi_1] ^\frac{1}{2}\lesssim\frac{\mathring t^\frac{{\mu_q^\varepsilon}}{2} \mathring \tau^{\frac{\sigma}{2}}}{\mathring t^q (1+\mathring r)}\bar{\bar{\E}}_q[\psi_0,\psi_1]^\frac{1}{2} \,,
    \end{align*}
    and for  $\epsilon < \mathring r$ we have
    \begin{align*}
        |\d_v \psi ( \mathring v,\mathring u, \mathring \omega)| \lesssim  \frac{1}{\mathring t^q \mathring r}\left(|\d_v (\mathring t^q \mathring r \psi)| + \mathring t^{2q-1}\mathring r |\psi| + \mathring t^q |\psi|\right)\lesssim  \frac{\mathring t ^\frac{{\mu_q^\varepsilon}}{2}\mathring \tau ^\frac{\sigma}{2}}{\mathring t^q \mathring r^2}(\mathring r^\frac{-1}{2}+\mathring \tau ^\frac{-1}{2})
        \bar{\bar{\E}}_q[\psi_0,\psi_1]^\frac{1}{2}\lesssim\frac{\mathring t ^\frac{{\mu_q^\varepsilon}}{2}\mathring \tau ^\frac{\sigma}{2}}{\mathring t^q (1+\mathring r)^2}\bar{\bar{\E}}_q[\psi_0,\psi_1] ^\frac{1}{2}\,.
    \end{align*}
    Finally, for $\sigma>0$ we have
    \begin{align*}
        |\d_t \psi(\mathring v,\mathring u, \mathring \omega)|&\lesssim \frac{1}{\mathring t^q}\left(|\d_v \psi|+|\d_u \psi|\right) \lesssim  \frac{\mathring t^\frac{{\mu_q^\varepsilon}}{2} \mathring \tau^{\frac{\sigma}{2}}}{\mathring t^{2q} (1+\mathring r)}\bar{\bar{\E}}_q[\psi_0,\psi_1]^\frac{1}{2}\,.
    \end{align*}
    The results for $\sigma=0$, when $q=\frac{1}{2}$ follows similarly.

    For $\mathring r\leq \epsilon \ll \rho$, notice that local Sobolev inequality \eqref{eq:general:Sobolev:Local} gives
    \begin{align*}
        |\d_t(\mathring t^q \psi)|^2\lesssim \sum_{i=0,1,2}\sum_{j=1,2,3} \mathring t ^{\mu_q^\varepsilon-2q}\int_{\{\mathring t\}\times B_1}\mathring t ^{2q-\mu_q^\varepsilon}|\d_t( \mathring t^q  \d_{x^j}^i\psi)|^2\,dx\lesssim  \frac{\mathring t ^{\mu_q^\varepsilon}}{\mathring \tau ^{2-\sigma}\mathring t^{2q}}\bar{\bar{\E}}_q[\psi_0,\psi_1]\,, 
    \end{align*}
    leading to
    \begin{align*}
        |\d_t \psi|\lesssim \frac{|\d_t(\mathring t^q \psi)|}{\mathring t^q} + \frac{|\psi|}{\mathring t}\lesssim \frac{\mathring t ^{\frac{\mu_q^\varepsilon}{2}} }{\mathring \tau ^{1-\frac{\sigma}{2}}\mathring t ^{2q}}\bar{\bar{\E}}_q[\psi_0,\psi_1]^\frac{1}{2} + \frac{\mathring t ^{\frac{\mu_q^\varepsilon}{2}}}{\mathring \tau ^{1-\frac{\sigma}{2}}\mathring t^{1+q}}\bar{\E}_q[\psi_0,\psi_1]^\frac{1}{2}\lesssim  \frac{\mathring t^\frac{{\mu_q^\varepsilon}}{2} \mathring \tau^{\frac{\sigma}{2}}}{\mathring t^{2q} (1+\mathring r)}\bar{\bar{\E}}_q[\psi_0,\psi_1]^\frac{1}{2}\,,
    \end{align*}
    which completes the proof of the desired estimate for the time derivative. Estimates for $\d_v$ and $\d_u$  when $r<\epsilon\ll \rho$ follow similarly by noticing that if $m$ denotes $v$ or $u$, we have
    \begin{align*}
        |\d_m(\mathring t^q \psi)|^2&\lesssim \mathring t ^{2q}\left(|\d_t( \mathring t^q \psi)|^2 +|\d_r\psi|^2\right) \lesssim\sum_{i=0,1,2}\sum_{j=1,2,3} \mathring t ^{\mu_q^\varepsilon}\int_{\{\mathring t\}\times B_1} \mathring t ^{2q-\mu_q^\varepsilon}\left(|\d_t( \mathring t^q \d_{x^j}^i\psi)|^2 +|\nabla\d_{x^j}^i\psi|^2\right)\,dx\\
        &\lesssim \frac{\mathring t ^{\mu_q^\varepsilon}}{\mathring \tau ^{2-\sigma}}\bar{\bar{\E}}_q[\psi_0,\psi_1]\,. 
    \end{align*}
    Finally, untwisting the above estimate would complete the proof.
\end{proof}

\bibliography{Biblio}
\bibliographystyle{amsrefs}
\end{document}